\documentclass[10pt]{article}

\usepackage[letterpaper, hmargin=0.8in, top=1in, bottom=1.2in, footskip=0.6in]{geometry}

\usepackage{titlesec}
\titleformat{\subsection}[runin]{\normalfont\bfseries}{\thesubsection.}{.5em}{}[.]\titlespacing{\subsection}{0pt}{2ex plus .1ex minus .2ex}{.8em}
\titleformat{\subsubsection}[runin]{\normalfont\itshape}{\thesubsubsection.}{.3em}{}[.]\titlespacing{\subsubsection}{0pt}{1ex plus .1ex minus .2ex}{.5em}
\titleformat{\paragraph}[runin]{\normalfont\itshape}{\theparagraph.}{.3em}{}[.]\titlespacing{\paragraph}{0pt}{1ex plus .1ex minus .2ex}{.5em}

\usepackage[labelfont=sc,font=small,labelsep=period]{caption}
\usepackage{amsmath}
\usepackage{amssymb}
\usepackage{amsfonts}
\usepackage{latexsym}
\usepackage{amsthm}
\usepackage{amsxtra}
\usepackage{amscd}
\usepackage{bbm}
\usepackage{mathrsfs}
\usepackage{bm}
\usepackage{graphicx, color}

\definecolor{darkred}{rgb}{0.9,0,0.3}
\definecolor{darkblue}{rgb}{0,0.3,0.9}
\definecolor{orange}{rgb}{0.98, 0.6, 0.01}

\definecolor{vdarkred}{rgb}{0.6,0,0.2}
\definecolor{vdarkblue}{rgb}{0,0.2,0.6}
\usepackage[pdftex, colorlinks, linkcolor=vdarkblue,citecolor=vdarkred]{hyperref}

\usepackage[nottoc,notlof,notlot]{tocbibind}
\usepackage{cite} %Enabling `[1-4]` - style citing
\numberwithin{equation}{section}
\numberwithin{figure}{section}
\theoremstyle{plain} %plain, definition, remark
\newtheorem{theorem}{Theorem}[section]
\newtheorem*{theorem*}{Theorem}
\newtheorem{lemma}[theorem]{Lemma}
\newtheorem*{lemma*}{Lemma}
\newtheorem{corollary}[theorem]{Corollary}
\newtheorem*{corollary*}{Corollary}
\newtheorem{proposition}[theorem]{Proposition}
\newtheorem*{proposition*}{Proposition}

\newtheorem*{conjecture*}{Conjecture}

\theoremstyle{definition} %plain, definition, remark
\newtheorem{definition}[theorem]{Definition}
\newtheorem*{definition*}{Definition}

\newtheorem*{example*}{Example}
\newtheorem{remark}[theorem]{Remark}
\newtheorem*{remark*}{Remark}

\newtheorem*{assumption*}{Assumption}

\renewcommand{\b}[1]{\boldsymbol{\mathrm{#1}}} %bold
\newcommand{\bb}{\mathbb} %blackboard bold
\renewcommand{\cal}{\mathcal} 
 
\newcommand{\fra}{\mathfrak} 
\newcommand{\ol}[1]{\overline{#1} \!\,} %overline
\newcommand{\wh}{\widehat}

\newcommand{\f}[1]{\bm{\mathrm{#1}}}

\renewcommand{\P}{\mathbb{P}}
\newcommand{\E}{\mathbb{E}}

\newcommand{\R}{\mathbb{R}}
\newcommand{\C}{\mathbb{C}}
\newcommand{\N}{\mathbb{N}}
\newcommand{\Z}{\mathbb{Z}}

\newcommand{\ee}{\mathrm{e}}
\newcommand{\ii}{\mathrm{i}}
\newcommand{\dd}{\mathrm{d}}
\newcommand{\col}{\mathrel{\vcenter{\baselineskip0.75ex \lineskiplimit0pt \hbox{.}\hbox{.}}}}
\newcommand*{\deq}{\mathrel{\vcenter{\baselineskip0.65ex \lineskiplimit0pt \hbox{.}\hbox{.}}}=}
\newcommand*{\eqd}{=\mathrel{\vcenter{\baselineskip0.65ex \lineskiplimit0pt \hbox{.}\hbox{.}}}}

\renewcommand{\leq}{\leqslant}
\renewcommand{\le}{\leqslant}
\renewcommand{\geq}{\geqslant}
\renewcommand{\ge}{\geqslant}
\renewcommand{\epsilon}{\varepsilon}

\newcommand{\qq}[1]{[\![{#1}]\!]}

\newcommand{\ind}[1]{\b 1 (#1)}

\newcommand{\p}[1]{({#1})}
\newcommand{\pb}[1]{\bigl({#1}\bigr)}

\newcommand{\pbb}[1]{\biggl({#1}\biggr)}
\newcommand{\pBB}[1]{\Biggl({#1}\Biggr)}

\newcommand{\qB}[1]{\Bigl[{#1}\Bigr]}
\newcommand{\qbb}[1]{\biggl[{#1}\biggr]}

\newcommand{\hb}[1]{\bigl\{{#1}\bigr\}}

\newcommand{\abs}[1]{\lvert #1 \rvert}
\newcommand{\absb}[1]{\bigl\lvert #1 \bigr\rvert}

\newcommand{\norm}[1]{\lVert #1 \rVert}

\newcommand{\avg}[1]{\langle #1 \rangle}

\newcommand{\avgbb}[1]{\biggl\langle #1 \biggr\rangle}

\newcommand{\scalar}[2]{\langle{#1} \mspace{2mu}, {#2}\rangle}

\DeclareMathOperator{\tr}{Tr}

\DeclareMathOperator{\re}{Re}
\DeclareMathOperator{\im}{Im}
\title{Diffusion Profile for Random Band Matrices: a Short Proof}

\author{Yukun He\footnote{University of Z\"{u}rich, Institute of Mathematics, {\tt yukun.he@math.uzh.ch}.} 
\and Matteo Marcozzi\footnote{University of Geneva, Section of Mathematics, {\tt matteo.marcozzi@unige.ch}.}}

\begin{document}
 
\maketitle

\begin{abstract}
Let $H$ be a Hermitian random matrix whose entries $H_{xy}$ are independent, centred random variables with variances $S_{xy} = \E|H_{xy}|^2$, where $x, y \in (\Z/L\Z)^d$ and $d \geq 1$. The variance $S_{xy}$ is negligible if $|x - y|$ is bigger than the band width $W$. 

For $ d = 1$ we prove that if $L \ll W^{1 + \frac{2}{7}}$ then the eigenvectors of $H$ are delocalized and that an averaged version of $|G_{xy}(z)|^2$ exhibits a diffusive behaviour, where $ G(z) = (H-z)^{-1}$ is the resolvent of $ H$. This improves the previous assumption $L \ll W^{1 + \frac{1}{4}}$ of \cite{EKYY13}. In higher dimensions $d \geq 2$, we obtain similar results that improve the corresponding ones from \cite{EKYY13}. Our results hold for general variance profiles $S_{xy}$ and distributions of the entries $H_{xy}$.

The proof is considerably simpler and shorter than that of \cite{EKY13,EKYY13}. It relies on a detailed Fourier space analysis combined with isotropic estimates for the fluctuating error terms. It is completely self-contained and avoids the intricate fluctuation averaging machinery from \cite{EKY13}.
\end{abstract}

\section{Introduction}

Given a large finite graph $\Gamma$, random band matrices $H=(H_{xy})_{x,y\in \Gamma}$ are matrices whose entries $ H_{xy}$ are independent and centred
random variables and the variance $S_{xy} \deq \E |H_{xy}|^2$ typically decays with the distance on a characteristic 
length scale $W$, called the \emph{band width} of $H$.

This name is due to the simplest one-dimensional
model where $\Gamma = \{1, 2, \dots, N\}$ and $H_{xy}=0$ if $|x-y|\geq W$, where $ 1 \leq W \leq L$. As an example of higher-dimensional models, one can take  
$\Gamma$ to be the box 
of linear size $L$ in $\Z^d$, so that the dimension of the matrix is
$N=L^d$. For a more general and extensive presentation of random band matrix models, we refer to \cite{Spe}.

From the physics view point, random band matrices turn out to be very useful to study the disordered systems. In fact, it is conjectured that, depending on the level of energy and disorder strength, all these systems belong to two universality classes: in the strong disorder regime (as for the random Schr\"odinger
operator models such as the Anderson model \cite{And}), the eigenfunctions are localized and the local spectral statistics are Poisson, while in the weak disorder regime (as for the mean-field models such as Wigner matrices \cite{Wig}), the
eigenfunctions are delocalized and the local statistics are those of a mean-field Gaussian matrix ensemble.

As $ W$ varies, random band matrices interpolate between these two classes: in particular, we recover the Wigner matrices by setting $W=N$ and all
variances equal, while for $W = O(1)$ we essentially obtain the Anderson model.  
The \emph{delocalization} property is expressed in the term of the \emph{localization length} $\ell$,
which, in the framework of random matrices, describes the typical length scale of the eigenvectors of $H$: if the localization length is comparable with the system size, $\ell\sim L$, the system is \emph{delocalized}
and it is {\it localized} otherwise. The direct physical interpretation of the delocalization is that delocalized systems describe electric conductors,
while localized systems insulators.
Therefore, random band matrices represent a good model to investigate the \emph{Anderson
metal-insulator phase transition}.

According to nonrigorous results \cite{FM}, for random band matrices the localization length is expected to be $\ell\sim W^2$ in $d=1$, exponentially growing in $W$ in $d=2$, 
and $\ell\sim L$ in $d\ge 3$, i.e.\
the system is delocalized. Notice that the claim that $ \ell\sim W^2$ in $ d = 1$ is equivalent to have the delocalization of the eigenvectors for $ W \geq L^{\frac{1}{2}}$: this is the formulation used in the abstract.  For more details on these conjectures, see \cite{Spe, Sp, Sch}. 

%
%These predictions are in accordance with those for the \emph{Anderson model}, where the random matrix is of the form $-\Delta+\lambda V$; here $\Delta$ is the lattice Laplacian, $V$ a random potential (i.e.\ a diagonal matrix with i.i.d.\ entries), and $\lambda$ a small coupling constant.
%The localization length is $\ell\sim \lambda^{-2}$ in 
%the regime of strong localization, which corresponds to the whole spectrum for $d=1$ and a neighbourhood of the spectral edges for $d > 1$.
%This result follows from the rigorous multiscale analysis
%of Fr\"ohlich and Spencer \cite{FroSpe} as well as 
%from the fractional moment method of Aizenman and Molchanov \cite{AizMol}.
%The two-dimensional Anderson model is conjectured
%to be in the weak localization regime with $\ell\sim \exp (\lambda^{-2})$
%throughout the spectrum \cite{Abr}, but this has so far not been proved.
%In dimensions $d\ge 3$, the prediction is that there is a threshold energy, called the mobility
%edge, $E_0$, that separates the localized regime
%near the band edges from the delocalized regime in the bulk spectrum.
%The localization length is expected to diverge as the energy $E$ approaches
%the mobility edge from the localization side. The increase of the
%localization length as an inverse power of $E-E_0$ has been
%rigorously established up to a certain scale in \cite{Spen, Elg},
%but this analysis does not allow $E$ to actually reach the conjectured value
%of $E_0$.
%A  key  open question for the Anderson model is to establish the metal-insulator
%transition, i.e.\ to show that the mobility edge indeed exists.
%
%
%

Up to now, only lower and upper bounds have been rigorously established for $\ell$. Since the Anderson transition can be studied via random matrices directly in $ d = 1$ by varying $ W$ in the interval $ 1 \leq W \leq L$, here we focus on the results in the one-dimensional setting, but analogous theorems were proved for higher dimensions.

Schenker \cite{Sch} showed that $\ell\le W^8$,
uniformly in the system size. The lower bound $\ell\ge W$ was proved in \cite{EYY1} by using
a self-consistent equation for the diagonal matrix entries $G_{xx}$
of the \emph{resolvent} (or \emph{Green function}) $G=G(z)=(H-z)^{-1}$, $z = E + \ii \eta \in \C$. In a series of papers Erd\H{o}s and collaborators improved this lower bound: in \cite{EK1, EK2} it was proved that $\ell \ge W^{1 + \frac{1}{6}}$ 
by using diagrammatic perturbation theory and then in \cite{EKYY13} the authors showed that $\ell \ge W^{1 + \frac{1}{4}}$  by employing a self consistent equation for an averaged version of $G$ and the so called \emph{fluctuation averaging estimates} \cite{EKY13}.

In \cite{EK1, EK2} the unitary time evolution, $\ee^{\ii tH}$, was analysed and it was shown to  
behave diffusively on the spatial scale $W$, i.e.\ the typical propagation
distance is $\sqrt{t} W$. From this, one deduces that a (superposition of) random walk with step size of order $W$ is responsible for the delocalization of random band matrices. The barrier at $W^{1 + \frac{1}{6}}$ is due to the fact, in order to see that the localization length $ \ell$ is greater than the naive size $W$, a control of the random walk for large times is needed, but the diagrammatic expansion used in \cite{EK1, EK2} allows to control the time evolution only up to time $t\ll W^{1/3}$. Hence, the delocalization occurs on the scale $\sqrt{t} W = W^{1 + \frac{1}{6}}$.

In \cite{EKYY13, EKY13} the same result was translated in term of the \emph{resolvent} $G$: as explained in Remark 2.7 in \cite{EKYY13}, controlling $\ee^{\ii tH}$ up to $t \ll W^{\nu}$ for some exponent $\nu > 0$ is equivalent to controlling $G(z)$ for $\eta \gg W^{-\nu}$. In \cite{EKYY13, EKY13} the authors managed to overcome the technical barrier of \cite{EK1, EK2} and they proved the diffusion behaviour of $G(z)$ when $\eta \gg W^{-1/2}$. 

We emphasize that the delocalization results \cite{EK1, EK2, EKYY13, EKY13} hold for general variance profiles $S_{xy}$ and distributions of the entries $H_{xy}$. Our results hold under similarly general assumptions.

In this paper we improve the results of \cite{EKYY13, EKY13}: we prove that $\ell \gg W^{1 + \frac{2}{7}}$ and that $G(z)$ exhibits a diffusive behaviour for $\eta \gg W^{-4/7}$. As in \cite{EKYY13}, the main object of interest is the matrix $T$, whose entries are local averages of $|G_{xy}|^2$:
$$
   T_{xy} \;\deq\; \sum_i S_{xi} |G_{iy}|^2.
$$
The band structure of $H$ is obtained by assuming that $S_{xy} : = \E |H_{xy}|^2 \approx W^{-d} f((i - j) / W)$ where $f$ is a symmetric probability density on $\R^d$. Note that in this way $ S$ is translation invariant.

We will show that $T$ satisfies a self-consistent equation of the form
\begin{equation} \label{self-const intro}
   T_{xy} = |\mathfrak{m}|^2\sum_{i} S_{xi} T_{iy} + |\mathfrak{m}|^2 S_{xy} + \cal E_{xy}\,,
\end{equation}
where 
$\mathfrak m \equiv \mathfrak m(z)$ is an explicit function of the spectral parameter $z\in \C$
and $\cal E$ is an error term. 

Translation invariance of $S$ implies that its Fourier transform $\wh s(p)$ for $\abs{p} \ll W^{-1}$ reads
\begin{equation} \label{low-p exp}
\wh s(p) \;\approx\; 1 - W^2 (p \cdot D p) + \cdots\,,
\end{equation}
where $D$ is the matrix of second moments of $f$ (see \eqref{def D}).

Neglecting the error term $\cal E$, we get from \eqref{self-const intro} that
\begin{equation} \label{T_approx}
  T \;\approx\; \frac{|\mathfrak m|^2S}{1-|\mathfrak m|^2S}\,.
\end{equation}
By using $|\fra m|^2= 1-\alpha\eta+ O(\eta^2)$ (see \eqref{m2} below), where
\begin{equation} \label{def alpha}
\alpha \;\equiv\; \alpha(E) \;\deq\; \frac{2}{\sqrt{4 - E^2}} \qquad (E = \re z)\,,
\end{equation} 
we obtain that the Fourier transform of $T$ is approximately given by
\begin{equation} \label{approx T hat}
    \frac{\alpha^{-1}}{\eta +  W^2(p \cdot D_{\rm eff} \, p)} \qquad \text{where} \qquad  D_{\rm eff} \;\deq\;
  \frac{D}{\alpha}\,,
\end{equation}
for $\abs{p}\ll W^{-1}$  and $\eta\ll 1$. 
This corresponds to the  diffusion approximation on scales larger than $W$
with  an effective diffusion constant  $D_{\rm eff}$.
As in \cite{EKYY13,EKY13}, the main work of the proof is to estimate the error term $\cal E_{xy}$ from \eqref{self-const intro}, in order to make the approximation \eqref{T_approx} rigorous.

The crucial difference between our approach and that of \cite{EKYY13} is how this error term is estimated. In \cite{EKYY13} it is bounded by the fluctuation averaging estimates relying on very intricate expansions derived using the Schur's complement formula \cite{EKY13}. In contrast, in this paper we first perform a careful analysis of the error term in Fourier space, and then apply isotropic error estimates to the resulting Fourier coefficients. Here, \emph{isotropic} refers to generalized matrix entries $\scalar{\b v}{\cal E \b w}$ where the vectors $\b v, \b w$ do not lie in the direction of the standard coordinate axes. A trivial but essential observation behind our proof is that the Fourier basis is completely delocalized. The isotropic error estimates are derived using the \emph{cumulant expansion formula} (see Lemma \ref{lem:cumulant_expansion} below) in the formulation given in \cite{HK17, HKR17}. We do not use Schur's complement formula at all. Thus, apart from sharing the basic starting point and self-consistent equation \eqref{self-const intro} with \cite{EKYY13, EKY13}, our proof differs fundamentally from that of \cite{EKYY13, EKY13}. For more details on the proof strategy, see Section \ref{sec:fourier}. We would also like to emphasis that our proof is considerably simpler than that of \cite{EKYY13, EKY13}: the arguments in \cite{EKYY13} need \cite{EKY13} as an input, and they are altogether 120 pages; on the other hand, our proof is completely self-contained. 

More recently, the series of works \cite{BYY18,BYYY18,YY18} further improves the result of this paper. For $W \gg L^{\frac{3}{4}}$, it was proved that the $L^{\infty}$-norm of the bulk eigenvectors of $H$ are all simultaneously bounded by $N^{-\frac{1}{2}+\varepsilon}$ with overwhelming probability. The authors also prove bulk universality under the same condition. The proof uses a strong (high probability) version of the quantum unique ergodicity property of random matrices, as well as estimates on the ``generalized resolvent" of band matrices. The later is obtained by the intricate fluctuation averaging machinery similar to \cite{EYY1}, but the estimates are done much more carefully. As a result, \cite{BYY18,BYYY18,YY18} are also more lengthy and complicated than this paper.

\subsection*{Acknowledgments}
We thank Antti Knowles, who motivated us to study this problem,
for useful discussions and suggestions on the topic. This project has received funding from the European Research Council (ERC) under the European Union’s Horizon 2020 research and innovation programme (grant agreement No.\ 715539\_RandMat) and the Swiss National Science Foundation (SNF).

\section{The main results}

%\subsection{Setup} \label{sec:setup}
Let $f:\R^d \to \R$ be a smooth and symmetric probability density for some fixed $d \geq 1$.
Let $L, W \in \N$ such that
\begin{equation} \label{lower bound on W}
L^\delta \;\leq\; W \;\leq\; L
\end{equation}
for some fixed $\delta > 0$. 
%The parameter $L$ is the fundamental large quantity of our model. 
We set $\bb T^d_L \deq [-L/2, L/2)^d \cap \Z^d $ to be the $d$-dimensional discrete torus, 
so that $\bb T^d_L$ has $N \deq L^d$ lattice points. For the following we fix an (arbitrary) ordering of $\bb T_L^d$, which allows us to identify it with $\{1, \dots, N\}$. Let the canonical representative of $i \in \Z^d$ be
\begin{equation*}
[i]_L \;\deq\; (i + L \Z^d) \cap \bb T^d_L\,,
\end{equation*}
and the periodic distance $\abs{i}_L \;\deq\; \absb{[i]_L}$,
where $\abs{\cdot}$ denotes Euclidean distance in $\R^d$.

Moreover, define the $N \times N$ matrix $S(L,W) \equiv S = (S_{ij} \col i,j \in \bb T_L^d)$ through
\begin{equation}\label{sij}
S_{ij} \;\deq\; \frac{1}{Z_{L,W}} \, f \pbb{\frac{[i - j]_L}{W}}\,,
\end{equation}
where $Z_{L,W}$ is a normalization constant such that for all $i \in \bb T_L^d$
\begin{equation}\label{stoch}
\sum_j S_{ij} \;=\; 1\,,
\end{equation}
i.e.\ $S$ is a stochastic matrix. Here we adopted the convention that summations
are always over the set $\bb T_L^d$, unless specified otherwise.
The symmetry of $f$ implies that $S$ is also symmetric: $S_{ij} = S_{ji}$. 
%Furthermore, 
%As a stochastic matrix, the spectrum of $S$ lies in $[-1,1]$.
%in Lemma A.1 of \cite{EYY1} it is proved that there exists a positive constant $\nu$, depending only on $f$,
%such that the spectrum of $ S$ lies in $[-1 + \nu, 1]$. 

Let $(\zeta_{ij} \col i \leq j)$, where $i,j \in \bb T_L^d$, be a family of independent, 
complex-valued, centred random variables $\zeta_{ij} \equiv \zeta_{ij}^{(N)}$ satisfying
\begin{equation}\label{zetacond}
\E \zeta_{ij} \;=\; 0\,, \qquad \E \abs{\zeta_{ij}}^2 \;=\; 1\,, \qquad \zeta_{ii}\in \R, \qquad \zeta_{ij} \deq \bar \zeta_{ji} \mbox{ for } i > j.
\end{equation}
The band matrix $H = (H_{ij})_{i,j\in \bb T_L^d}$ is defined as
\begin{equation}\label{def:band}
H_{ij} \;\deq\; (S_{ij})^{1/2} \, \zeta_{ij}\,,
\end{equation}
By definition we have $H = H^*$ and $\E \abs{H_{ij}}^2 \;=\; S_{ij}$. Moreover, we assume that the random variables $\zeta_{ij}$ have finite moments, uniformly in $N$, $i$, and $j$, in the sense that for all $p \in \N$ there is a constant $\mu_p$ such that
\begin{equation} \label{finite moments}
\E \abs{\zeta_{ij}}^p \;\leq\; \mu_p
\end{equation}
for all $N$, $i$, and $j$. Furthermore, we set the parameter
\begin{equation}  \label{s leq W}
M \;\equiv\; M_N \;\deq\; \frac{1}{\max_{i, j} S_{ij}}\,.
\end{equation}
Note that $M = \pb{W^d + O(W^{d - 1})} / \norm{f}_\infty$ since the definition of $S$ implies that $Z_{N,W} = W^d +O(W^{d-1})$. 
Here we used  the  usual $O(\cdot)$ notation. Furthermore, we will write $O_{\alpha}(\cdot)$ if  the  implicit  constant
depends on the parameter $\alpha$ which can never depend on $N$. 

The following definition introduces a notion of a high-probability bound that will be used throughout the whole paper.
\begin{definition}[Stochastic domination]\label{sdI}
Let 
$$
X = \pb{X^{(N)}(u) \col N \in \N, u \in U^{(N)}}, \ \ \ Y = \pb{Y^{(N)}(u) \col N \in \N, u \in U^{(N)}}
$$
be two families of random variables, where the $Y^{(N)}(u)$ are nonnegative and $U^{(N)}$ is a possibly $N$-dependent parameter set.  We say that $X$ is \emph{stochastically dominated by $Y$, uniformly in $u$,} if for all $\epsilon > 0$ and  $D > 0$ we have
\begin{equation*}
\sup_{u \in U^{(N)}} \P \qB{|X^{(N)}(u)| > N^\epsilon Y^{(N)}(u)} \;\leq\; N^{-D}
\end{equation*}
for large enough $N\ge N_0(\epsilon, D)$. Unless stated otherwise, 
throughout this paper the stochastic 
domination will always be uniform in all parameters apart from the parameter $\delta$ in \eqref{lower bound on W} and the sequence of constants $\mu_p$ in \eqref{finite moments}; thus, $N_0(\epsilon, D)$ also depends on $\delta$ and $\mu_p$.
If $X$ is stochastically dominated by $Y$, uniformly in $u$, we use the equivalent notations
\begin{equation*}
X \;\prec\; Y \qquad \text{and} \qquad X \;=\; O_\prec(Y)\,.
\end{equation*}
%If for some complex family $X$ we have $|X| \prec Y$ we also write $X = O\prec(Y )$. 
\end{definition}

%For example, using Chebyshev's inequality and \eqref{finite moments} one easily finds that 
%\begin{equation}\label{hsmallerW}
%H_{ij} \;\prec\; (S_{ij})^{1/2} \;\leq\; M^{-1/2}\,,
%\end{equation}
%so that we may also write $H_{ij} = O_\prec((S_{ij})^{1/2})$.
As stated in Lemma \ref{lemma:basic_properties_of_prec} below, the relation $\prec$ satisfies the familiar algebraic rules of order relations. Moreover, note that for deterministic $X$ and $Y$, $X = O_{\prec}(Y)$ implies $X = O_{\epsilon}(N^{\epsilon}Y )$ for any $\epsilon > 0$. 

We note that, if \eqref{finite moments} only holds for some large but fixed $p$, Definition \ref{sdI} should be slightly modified so that our results would still remain true, even though in a slightly weaker sense. In fact, in this case the exponents $\epsilon$ and $D$ in the definition of $\prec$ must depend on $p$, therefore one should keep track of the $p$ dependencies in all the stochastic domination bounds involved in the proof. 

We use the spectral parameter $z = E + \ii \eta \in \C$ with $ \eta := \mathrm{Im} \,z > 0$ and we fix two arbitrary (small) global constants $\gamma > 0$ and $\kappa > 0$.
 All of our estimates will depend on $\kappa$ and $\gamma$, and we shall often 
omit the explicit mention of this dependence. We set
\begin{equation}\label{def:S}
\mathbf S \;\equiv\; \mathbf S^{(N)}(\kappa, \gamma) \;\deq\; \hb{E + \ii \eta \col -2 + \kappa \leq E \leq 2 - \kappa \,,\, M^{-1 + \gamma} \leq \eta \leq 10}\,.
\end{equation}
and we shall always assume that $z \in \mathbf S(\kappa, \gamma)$. 
%In this paper we always consider families $X^{(N)}(u) = X^{(N)}_i(z)$ indexed by $u = (z,i)$, where $z \in \mathbf S(\kappa,\gamma)$ and $i$ takes on values in some finite (possibly $N$-dependent or empty) index set.

%\cob [Postpone into proof.] \nc We shall call \emph{control parameter} any positive and deterministic quantity $\Psi^{(N)}(z)$ and we shall call \emph{admissible} any control parameter $\Psi^{(N)}(z)$ such that
%\begin{equation} \label{admissible Psi}
%M^{-1/2} \;\leq\; \Psi^{(N)}(z) \;\leq\; M^{-\gamma/2}
%\end{equation}
%for all $N$ and $z \in \b S$, where $\gamma$ is the same fixed number as in \eqref{def:S}.
%A typical example of an admissible control parameter is $\Psi(z) = \frac{1}{\sqrt{M \eta}}$.

We will use the Stieltjes transform of Wigner's semicircle law, which is defined by
\begin{equation} \label{definition of msc}
\mathfrak m(z) \;\deq\; \frac{1}{2 \pi} \int_{-2}^2 \frac{\sqrt{4 - \xi^2}}{\xi - z} \, \dd \xi\,
\end{equation}
and is characterized by the unique solution of
\begin{equation} \label{eq:m_equation}
\mathfrak m(z) + \frac{1}{\mathfrak m(z)} + z \;=\; 0
\end{equation}
with $\im \mathfrak m(z) > 0$ for $\im z > 0$, i.e.\ 
\begin{equation} \label{explicit m}
\mathfrak m(z) \;=\; \frac{-z + \sqrt{z^2 - 4}}{2}\,.
\end{equation}
%We also recall that in the bulk there exist a small constant $ c$ such that $ c \leq |\frak m| \leq 1$. 
%To avoid confusion, we remark that the Stieltjes transform $\mathfrak m$ was denoted by $m_{sc}$ 
%in the papers \cite{ESY1, ESY2, ESY3, ESY4,ESY5,ESY6, ESY7, ESYY, EYY1, EYY2, EYY3, EKYY1, EKYY2}, in which $\mathfrak m$ had a different meaning from \eqref{definition of msc}. 
We define the \emph{resolvent} of $H$ through
\begin{equation*}
G \;\equiv\; G(z) \;\deq\; (H - z)^{-1}\,,
\end{equation*}
and denote its entries by $G_{ij}(z)$. 

Moreover, throughout the paper $C$ and $c$ will denote a generic large and small positive constant respectively, which may depend on some fixed parameters and whose value may change from one expression to the next. Given two positive quantities $A_N$ and $B_N$, the notation $A_N \asymp B_N$ means $c A_N \leq B_N \leq C A_N$, while $A_N \ll B_N$ means that there exists a constant $c > 0$ such that $A_N \leq N^{-c} B_N$; we also use $A_N \gg B_N$ to denote $B_N \ll A_N$. We also set $\avg{x} \deq \sqrt{1 + \abs{x}^2}$. Notice that we will often drop the $z$ dependence from the notation, even though most quantities 
in this paper depend on it.

Finally, we assume the following decay condition on $f$ (and therefore on $S$)
\begin{equation} \label{decay of f}
\abs{f(x)} \;\leq\; C_n \avg{x}^{-n} \qquad \text{for all } n \in \N\,.
\end{equation}
We are now ready to state our main theorems in $d=1$. The analogous results for $d >1$ are discussed in Section \ref{sec:high_dim}.

\subsection{Local semicircle law and delocalization}
We introduce the control parameter $ \Phi$ through
\begin{align}\label{eq:def_Phi}
\Phi^2 := \frac{1}{L\eta} + \frac{1}{W\sqrt{\eta}}\,.
\end{align}
Our first theorem is a local semicircle law for the resolvent entries: more precisely we get a bound for the random variable 
\begin{equation*}
\Lambda(z) \;\deq\; \max_{x,y} \absb{G_{xy}(z) - \delta_{xy} \mathfrak m(z)}\,.
\end{equation*}

\begin{theorem}[Local semicircle law]\label{th:local_law_d=1}
Assume \eqref{decay of f} and that
\begin{equation} \label{cond on N eta}
L \ll W^{1 + \frac{1}{3}}, \ \ \ \   \eta \gg \frac{L^{3/2}}{W^{5/2}}\,.
\end{equation}
Then for $z\in \mathbf S$ we have
\begin{equation}\label{decc}
 \Lambda^2 \;\prec\; \Phi^2\,.
\end{equation}
\end{theorem}
This improves Theorem 2.2 in \cite{EKYY13} where the assumptions are 
$ L \ll W^{1 + \frac{1}{4}}$ and $\eta \gg \frac{L^{2}}{W^3} = \frac{L^{3/2}}{W^{5/2}} \frac{L^{1/2}}{W^{1/2}}$. 
%Moreover, under the conditions \eqref{cond on N eta}, theorem \ref{th:local_law_d=1} improves the earlier result 
%\begin{equation}\label{bandbound}
%\Lambda^2 \;\prec\; \frac{1}{M\eta}
%\end{equation}
%proved in \cite{EYY1} (see Lemma \ref{lm:lsc} below).
%
%Moreover, as explained in Remark 2.6 in \cite{EKYY13}, the optimal range for the parameters would be
%\begin{align}
%L \ll W^{2}, \ \ \ \ \eta \gg L^{-1}.
%\end{align}
Note also that, as one can see from \eqref{prof1} and \eqref{Tprec} below, \eqref{decc} is optimal. 

We observe that by spectral decomposition of $G$ one gets that for any $ d$
\begin{align}\label{eq:average_G}
\frac{1}{N^2}  \sum_{x,y} |G_{xy}|^2 \;=\; \frac{1}{N^2} \tr G^* G \;=\; \frac{1}{N\eta} \im \frac{\tr G}{N} 
   \;=\; \frac{\im m}{N\eta} + O_\prec \pbb{ \frac{\Lambda}{N\eta}}\,.
\end{align}
Furthermore, from \eqref{eq:def_Phi} we have $\Phi^2 \leq C (L\eta)^{-1}$ for $\eta \le (W/L)^2$, thus \eqref{decc} shows that all off-diagonal entries of $G$ have a magnitude comparable with the average of their magnitudes computed in \eqref{eq:average_G}: in other words, $|G_{xy}|^2$ is essentially flat. In this case we say that the resolvent is \emph{completely delocalized} and this implies that the eigenvectors of $H$ are delocalized.

More precisely, let us denote the eigenvalues of $H$ by $\lambda_1\leq \lambda_2 \leq \cdots \leq \lambda_N$, and the associated normalized eigenvectors by $\b{u}_1,\b{u}_2...,\b{u}_N$. We use the notation $\b{u}_{\alpha}=(u_{\alpha}(x))_{x=1}^N$. We shall only consider eigenvectors associated with eigenvalues lying in the interval
$$\cal I \deq [-2 + \kappa, 2 - \kappa]\,,$$ 
where $\kappa > 0$ is fixed. For $\ell\equiv\ell(N)$, define the characteristic function $P_{x,\ell}$ projecting onto the complement of the $\ell$-neighborhood of $x$,
\[
P_{x,\ell}(y)\deq  \mathbf{1}(|x-y|\geq \ell)\,.
\]
Let $\varepsilon>0$ and define the random subset of eigenvector indices through
\[
\cal A_{\varepsilon,\ell}\deq \Big\{\alpha:\lambda_{\alpha}\in \cal I, \sum_{x}|u_{\alpha}(x)|\|P_{x,\ell}\b{u}_{\alpha}\|\leq \varepsilon  \Big\}
\]
which indexes the set of eigenvectors localized on scale $\ell$ up to an error $\varepsilon$. 

As stated in \cite[Proposition 7.1]{EKYY13}, given the interval  the condition for the eigenvector delocalization is that
\begin{equation} \label{large eta estimate}
\sup_{E \in \cal I} \abs{G_{xy}(E + \ii \eta)}^2 \;\prec\; \frac{1}{N \eta} + \delta_{xy}\,
\end{equation}
for $\eta $ such that $M^{-1+\gamma} \leq \eta \ll 1$.  Hence by choosing $\eta$ such that $L^{3/2} W^{-5/2} \ll \eta\le W^2 L^{-2}$, \eqref{large eta estimate} and Theorem \ref{th:local_law_d=1} imply the following corollary.

\begin{corollary}[Eigenvector delocalization]\label{cor:deloc_d=1}
Assume \eqref{decay of f} and $ L \ll W^{1 + \frac{2}{7}}$. Then the eigenvectors of $H$ are completely delocalized in the sense that for any $\ell \ll N$ and fixed $\varepsilon>0$, we have
\[
\frac{|A_{\varepsilon,\ell}|}{N}\leq C\sqrt{\varepsilon}+O_{\prec}(N^{-c})\,.
\]
\end{corollary}
This improves \cite[Corollary 2.4]{EKYY13}, where the condition is $ L \ll W^{1 + \frac{1}{4}}$. 

\subsection{Diffusion profile}\label{sec:diff_d=1}
Set $\bb T_L^1 \equiv \bb T_L$.
In the previous section we saw that for $\eta\le (W/L)^2$ the profile of $|G_{xy}|^2$
 is essentially flat, while we will see that for $\eta\ge (W/L)^2$, an averaged version of $|G_{xy}|^2$ is well approximated by a diffusion profile
\begin{equation}\label{Ydet}
\Theta_{xy} \;\deq\; \pbb{\frac{\abs{\fra m}^2 S}{1 - \abs{\fra m}^2 S}}_{xy}, \qquad x,y \in \bb T_L \,.
\end{equation}
Note that the matrix $\Theta $ is the solution of 
$$ \Theta = |\fra m|^2 S \Theta + |\fra m|^2S, 
$$
which is obtained from \eqref{self-const intro} by dropping the error term $\cal E$.
For a precise formulation of the result, we consider the weighted average
\begin{equation}\label{def:T}
    T_{xy} \;\deq\; \sum_{i} S_{xi}|G_{iy}|^2\,.
\end{equation}

\begin{theorem}[Diffusion profile]\label{th:diffusion_d=1}
Assume \eqref{decay of f} and  
$$
(W/L)^2 \le \eta \le 1\,, \ \ \ L \ll W^{1 + \frac{2}{7}}.
$$
Then
\begin{align}
\label{Tprec}
 T_{xy} -\Theta_{xy} & \;\prec\; \frac{1}{L\eta}\,, \\
 \label{Tfin}
  T_{xy} & \;\prec\;  \Upsilon_{xy}\,, \\ 
\label{Tfin1}
\absb{G_{xy}-\delta_{xy} \mathfrak m}^2 & \;\prec\;  \Upsilon_{xy}\,,
\end{align}
where we defined
\begin{equation}\label{Tdetbound}
\Upsilon_{xy}\; \equiv \; \Upsilon_{xy}^{(K)} \;\deq\; \frac{1}{L\eta} +\frac{1}{W\sqrt{\eta}}
\exp \qbb{ -\frac{\sqrt{\alpha \eta}}{W\sqrt{D}}|x-y|_L} 
   +\frac{1}{W} \avgbb{ \frac{\sqrt{\eta}|x-y|_L}{W}}^{-K}\,.
\end{equation}
Here  $K$ is an arbitrary, fixed, positive integer.
All estimates are uniform in $z\in \mathbf S$ and $x,y \in \bb T_L$.
\end{theorem}

\begin{remark} Theorem \ref{th:diffusion_d=1} improves the analogous result in \cite{EKYY13} where it was assumed that $ L \ll W^{1 +\frac{1}{4}}$. However, one expects that \eqref{Tprec} and the local semicircle law \eqref{decc} should in fact hold under the weaker conditions $\eta\gg 1/{L}$ and $L \ll W^2$, from which one could deduce the complete delocalization of the eigenvectors for all $L \ll W^2$.
%The main problem in proving this claim is that is not clear how to control $\Lambda$ when $\eta \le 1/{W}$.
\end{remark}

As explained in Section 2.2 in \cite{EKYY13}, $\Theta$ can be interpreted in terms of random walks. In fact, since $S$ is translation invariant, also $ \Theta$ is so and it can be written as 
\begin{equation} \label{random walk picture}
\Theta_{xy} \;=\; \sum_{n \geq 1} \abs{\fra m}^{2n} (S^{n})_{xy}\, \approx \sum_{n \geq 1} \ee^{-n \alpha \eta} (S^n)_{xy},
\end{equation}
where we used \eqref{m2} below and $ \alpha$ is defined in \eqref{def alpha}. Moreover, $S$ is the transition matrix of a random walk on $\bb T_L$ whose steps are of size $W$, therefore $\Theta_{xy}$ is a superposition of random walks up to times of order $(\alpha \eta)^{-1}$.

The normalized variance of each step is given by the unrenormalized diffusion constant $D$ appearing in \eqref{low-p exp}:
\begin{equation} \label{def D}
D \;\deq\; \frac{1}{2} \sum_{u} \pbb{\frac{u}{W}}^2 S_{u0} = D_\infty + O(W^{-1}) \qquad \text{where} \qquad D_\infty \;\deq\; \frac{1}{2} \int x^2 f(x) \, \dd x\,. 
\end{equation}
Proposition 2.8 in \cite{EKYY13} provides an explicit formula for $\Theta$: for each $K \in \N$ we have
\begin{equation}\label{prof1}
\Theta_{xy} \;=\;
\theta_{x - y}
+
 O\pbb{\frac{1}{W^2}}
+O_K \pBB{\frac{1}{W} \avgbb{\frac{\sqrt{\eta}\,|x-y|_L}{W}}^{-K}}
\end{equation}
where
\begin{equation} \label{def fra P}
\theta_x \;\deq\; \frac{\abs{\fra m}^2}{L}\sum_{p \in \frac{2\pi}{ L} \bb Z }  \ee^{\ii px} \frac{1}{\alpha \eta + W^2 D p^2} \;=\; \frac{|\fra m|^2}{2W\sqrt{D\alpha \eta}} \,  \sum_{k\in \bb Z}
 \exp \qbb{- \frac{\sqrt{\alpha \eta}}{W \sqrt{D}} \, \absb{x+kL}}\,.
\end{equation}
\\
From \eqref{prof1}, \eqref{def fra P} and \eqref{Tdetbound} we see that $\Theta_{xy} \leq C \Upsilon_{xy} \leq C' \Phi^2 $. 
The total mass of the profile $\sum_{x} \theta_x$ is simply given by 
\begin{equation} \label{total mass 1}
\sum_{x} \theta_x \;=\; \frac{\abs{\fra m}^2}{\alpha \eta} \;=\; \frac{\im \fra m}{\eta} \pb{1 + O(\eta)}\,,
\end{equation}
where in the last step we used Lemma \ref{lemma: msc} below.  
%In fact, the calculation \eqref{total mass 1} is a mere consistency check (to leading order) since $\sum_x \Theta_{x0} = \frac{\im \fra m}{\eta}$; see \eqref{total mass 3} below. 
Hence, the average height of the profile is of order $(L \eta)^{-1}$, while its maximum is of order $(W \sqrt{\eta})^{-1}$: this means that when $\eta \gg (W/L)^2$ the profile is concentrated in the region
$|x-y| \leq W\eta^{-1/2} \ll L$ and the complete delocalization has not taken place. For more details about the physical interpretation of $\Theta$ we refer to \cite{EKYY13, EK1}.
% The total weight of $\Theta_{u0}$ is 
%\begin{equation} \label{total mass 0}
%\sum_u \Theta_{u0} \;=\; \sum_{n \geq 1} \abs{\fra m}^{2n} \;\approx\; (\alpha \eta)^{-1}\,;
%\end{equation}
%a precise computation is given in \eqref{total mass 3} below.
%Note that the total mass of the distribution $\abs{G_{x0}}^2$ may be computed explicitly by spectral decomposition of $G$: assuming $\Lambda \prec \Psi$ we have
%\begin{equation} \label{total mass 2}
%\sum_{x} T_{x0} \;=\; \sum_x \abs{G_{x0}}^2 \;=\;  \frac{\im G_{00}}{\eta} \;=\; \frac{\im \fra m}{\eta} (1 + O_\prec(\Psi))\,,
%\end{equation}
%in agreement with the corresponding statement \eqref{total mass 0} for the deterministic limiting profile.

\section{Preliminaries} \label{sec:prelim}
In this Section we work in the $d$-dimensional setting. The following lemma collects basic algebraic properties of stochastic domination $\prec$. 
\begin{lemma}[Basic properties of $\prec$]\label{lemma:basic_properties_of_prec} $~$
\begin{enumerate}
	\item Suppose that $X(v) \prec Y (v)$ for all $v\in V$. If $|V|\leq N^C$ for some constant $C$ then $\sum_{v\in V}X(v)\prec \sum_{v\in V}Y(v)$.
	\item Suppose that $X_1 \prec Y_1$ and $X_2 \prec Y_2$. Then $X_1 X_2 \prec Y_1 Y_2$.
	\item Suppose that $X \le N^C$ and $Y \ge N^{-C}$ for some constant $C>0 $. Then $X \prec Y$ implies $\E[X] \prec \E[Y]$.
\end{enumerate}
If the above random variables depend on some parameter $u$ and the hypotheses are uniform in $u$ then so are the conclusions.
\end{lemma}

\begin{proof}
	The proof follows from the definition of stochastic domination together with a union bound argument. 
\end{proof}

Note that if for any $\epsilon >0$ and $p\geq 1$ we have $\E |X|^p \leq N^\epsilon \Psi^p$
for large enough $N$ which depends on $\epsilon$ and $p$, then $X\prec\Psi$ by Chebyshev's inequality. 
%Moreover, if $X\leq \Psi$ almost surely,
%then $X\prec \Psi$. Hence $O_\prec(\Psi)$ describes a larger class of random variables than $O(\Psi)$.

A crucial tool in our analysis is the cumulant expansion formula: the following version of the cumulant expansion formula is proved in \cite{HKR17}. Slightly different versions of the same formula can be found in \cite{KKP, Kho2, HK17}).
\begin{lemma}[Cumulant expansion]\label{lem:cumulant_expansion}
	Let $h$ be a complex random variable with all its moments exist. The $(p,q)$-cumulant of $h$ is defined as
		$$
		\mathcal{C}^{(p,q)}(h)\deq (-i)^{p+q} \cdot \left(\frac{\partial^{p+q}}{\partial {s^p} \partial {t^q}} \log \E e^{\mathrm{i}sh+\mathrm{i}t\bar{h}}\right) \bigg{|}_{s=t=0}\,.
		$$
		Let $f:\C^2 \to \C$ be a smooth function, and we denote its holomorphic  derivatives by
		$$
		f^{(p,q)}(z_1,z_2)\deq \frac{\partial^{p+q}}{\partial {z_1}^p \partial {z_2}^q} f(z_1,z_2)\,.
		$$ Then for any fixed $\ell \in \N$, we have
		\begin{equation} \label{eq:cumulant_expansion}
		\E f(h,\bar{h})\bar{h}=\sum\limits_{p+q=0}^{\ell} \frac{1}{p!\,q!}\mathcal{C}^{(p,q+1)}(h)\E f^{(p,q)}(h,\bar{h}) + R_{\ell+1}\,,
		\end{equation}
		given all integrals in \eqref{eq:cumulant_expansion} exists. Here $R_{\ell+1}$ is the remainder term depending on $f$ and $h$, and for any $\tau>0$, we have the estimate
		\begin{align}\label{eq:R_cumulant_expansion}
			\begin{aligned}
			R_{\ell+1}=&\ O(1)\cdot \E \big|h^{\ell+2}\cdot\mathbf{1}_{\{|h|>N^{\tau-1/2}\}}\big|\cdot \max\limits_{p+q=\ell+1}\big\| f^{(p,q)}(z,\bar{z})\big\|_{\infty} \\
			&+O(1) \cdot \E |h|^{\ell+2} \cdot \max\limits_{p+q=\ell+1}\big\| f^{(p,q)}(z,\bar{z})\cdot \mathbf{1}_{\{|z|\le N^{\tau-1/2}\}}\big\|_{\infty}\,.
			\end{aligned}
		\end{align}
\end{lemma}

The following result gives bounds on the cumulants of the entries of $H$.
\begin{lemma}\label{lem:cumulant_factos_estimate} 
If $H$ satisfies \eqref{def:band} and \eqref{finite moments} then for every $i,j\in\qq{N}$ and $k\in \N$ we have
	\begin{equation*}
		\cal C^{(p,q)}(H_{ij})=O_k(S_{ij}^{k/2})=O_k\big(M^{-k/2}\big), \ \ \ p+q = k
	\end{equation*}
	and $\cal C^{(0,1)}(H_{ij})= \cal C^{(1,0)}(H_{ij}) = 0$.
\end{lemma}

\begin{proof} This follows easily by the homogeneity of the cumulants.
\end{proof}

Moreover, we need the following bound on $\Lambda$, which was proved in \cite{EYY1} and \cite{EKYY4} with two different approaches.
\begin{lemma} \label{lm:lsc}
We have
\begin{equation} \label{rough estimate on Lambda}
\Lambda(z) \;\prec\; \frac{1}{\sqrt{M \eta}}\,
\end{equation}
uniformly for $z\in \b S$.
\end{lemma}

Finally, we collect some elementary facts about $\fra m$ which were proved in Lemma 3.5 in \cite{EKYY13}.
\begin{lemma} \label{lemma: msc}
\begin{itemize}
\item[(i)] We have the identity
\begin{equation} \label{id m im m}
1 - \abs{\mathfrak m}^2 \;=\; \frac{\eta \abs{\mathfrak m}^2}{\im \mathfrak m}\,.
\end{equation}
\item[(ii)] There is a constant $c > 0$ such that for $z \in \b S$
\begin{equation} \label{m is bounded}
c \;\leq\; \abs{\mathfrak m} \;\leq\; 1.
\end{equation}
\item[(iii)] For $z \in \b S$ and $\alpha$ given in \eqref{def alpha}
\begin{equation} \label{m2}
|\fra m|^2 \;=\; 1 - \eta\alpha + O(\eta^2).
\end{equation}
\item[(iv)] For $z \in \b S$
\begin{equation}\label{m3}
\im \fra m \;\asymp\; 1 \,, \qquad 1 - |\mathfrak m|^2 \;\asymp\; \eta
\end{equation}
where the implicit constants in the two estimates depend on $\kappa$.
\end{itemize}
\end{lemma}

\section{Self-consistent equation for $ T$ and Fourier space analysis}\label{sec:fourier}
From this section up to Section \ref{sec:nonGaussian_generalComplex} we will work in $d=1$, while in Section \ref{sec:high_dim} we will discuss the high dimensions case.
Following \cite{EKYY13}, we write a self-consistent equation
for $T$ defined in \eqref{def:T} and we control its error term by using a Fourier space argument. Set

\begin{align}\label{eq:error_def}
\cal E_{xy} := T_{xy} - |\mathfrak{m}|^2\sum_{i} S_{xi} T_{iy} - |\mathfrak{m}|^2 S_{xy},
\end{align}
so that $ T$ satisfies \eqref{self-const intro}. 
%\begin{align}\label{eq:T_selfconsistent_equation}
%T_{xy} = |\mathfrak{m}|^2\sum_{i} S_{xi} T_{iy} + |\mathfrak{m}|^2 S_{xy} + \cal E_{xy}
%\end{align}
We introduce the projection $ \Pi := \mathbf{i} \mathbf{i}^*$ with $ \mathbf{i}= L^{-1/2}(1, \ldots, 1)$ and we denote the complementary projection by $\overline{\Pi} := I-\Pi $ where $I$ is the identity matrix. Setting $ \overline{T}_y = (\Pi T)_{xy}$, from Proposition 5.1 in \cite{EKYY13} we know that
\begin{align}\label{eq:T_solved}
T_{xy} = \overline{T}_y + |\mathfrak{m}|^2 \bigg( \frac{S - \Pi}{I-|\mathfrak{m}|^2 S} \bigg)_{xy} + \tilde{\cal E}_{xy}
\end{align}
where
\begin{align}\label{eq:Tbar}
\overline{T}_y = \frac{\mathrm{Im}\, \mathfrak{m}}{L \eta} \bigg[1+O_{\prec}\bigg(\frac{1}{\sqrt{W\eta}} \bigg)  \bigg]
\end{align}
and
\begin{align}\label{eq:tilde_E}
\tilde{\cal E}_{xy} = \bigg(\frac{1}{I-|\mathfrak{m}|^2 S} \overline\Pi \cal E \bigg)_{xy}.
\end{align}
Note that $ \frac{\im \fra m}{L \eta} = \Pi_{xy} \frac{\im \fra m}{\eta}$ and
$$
\frac{\im \fra m}{\eta} \Pi + |\mathfrak{m}|^2 \bigg( \frac{S - \Pi}{I-|\mathfrak{m}|^2 S} \bigg) = \Theta
$$
where we used \eqref{id m im m} and the fact that $ \Pi S = S \Pi = \Pi$. Thus, from \eqref{eq:T_solved} and \eqref{eq:Tbar} we get
\begin{align}\label{eq:T_solved2}
T_{xy} = O_{\prec} \bigg( \frac{1}{L \sqrt{W \eta^3}} \bigg) + \Theta_{xy} + \tilde{\cal E}_{xy}.
\end{align}
%It is instructive to perform the same averaging with the deterministic profile $\Theta$:
%\begin{equation} \label{total mass 3}
%\frac{1}{N} \sum_x \Theta_{xy} \;=\; \pbb{\Pi \, \frac{\abs{\fra m}^2 S}{1 - \abs{\fra m}^2 S}}_{yy} \;=\; \frac{1}{N} \frac{\abs{\fra m^2}}{1 - \abs{\fra m}^2} \;=\; \frac{\im \fra m}{N \eta}\,,
%\end{equation}
%where in the last step we used \eqref{id m im m}.

In order to get a bound for $\tilde{\cal E}$, we will analyze it in Fourier space.  
%let $ \bb T_L^1$ be the one-dimensional lattice defined as
%\begin{align*}
%& \Gamma := \bigg \{-\frac{L-1}{2}, \ldots, \frac{L-1}{2}  \bigg \}, \ \ \ \ \mbox{for $L$ odd} \\\nonumber
%& \Gamma := \bigg \{-\frac{L}{2}+1, \ldots, \frac{L}{2}  \bigg \}, \ \ \ \ \mbox{for $L$ even}
%\end{align*}
We introduce our conventions for the discrete Fourier transform: given $f: \bb T_L \to \C $ and $p \in \frac{2 \pi}{L}\bb T_L =:  \bb T_L^* $, then we set
\begin{align*}
\widehat f(p): = \sum_{x \in \bb T_L} \ee^{-\ii p  x}f_x, \ \ \ f_x= \frac{1}{L} \sum_{p \in \bb T_L^*} \ee^{\ii p x} \widehat f(p).
\end{align*}
%Here $ p \cdot x$ denotes the standard inner product between the $d$-dimensional vectors $p$ and $x$. 

We define the family of vectors $(\b e(p), p \in \bb T_L)$ whose components are $ e_x(p):=L^{-1/2} \ee^{\ii p x}$. Thus, one can write $ \widehat{f}(p)= L^{1/2}\langle \b e(p), f \rangle$, where $ \langle \cdot, \cdot \rangle$ is the standard inner product of $ L^2(\bb T_L)$. Note that $\b e(p)$ is completely delocalized for all $p \in \bb T_L^*$, in the sense that $\norm{\b e(p)}_\infty \leq L^{-1/2}$.
  
%About the $S$, given its definition in terms of a probability distribution $f$, we know that in Fourier space it looks like
%\begin{align}\label{eq:S_approx}
%1-\widehat{S}(p) \simeq (Wp)^2, p < 1/W; \ \ \ \ \ \ \  1-\widehat{S}(p) \simeq 1, p \geq 1/W.
%\end{align}
%where $W$ is the width of the matrix band. 
%
%One can solve the selfconsistent equation \eqref{eq.selfconsistent_equation} by writing
%\begin{align}
%T = \frac{|m|^2 S}{1-|m|^2 S} + \big(\frac{1}{1-|m|^2 S} \big) \cal E
%\end{align} 
%
% Note that $ |m|^2 \leq 1$. 
In the following proposition we get a bound for $ \tilde{\cal E}_{xy}$ via a Fourier space argument. The idea is to split $\tilde{ \cal E}$ in three parts corresponding to the zero mode, the low modes and the high modes contribution. The treatment of the first and the third term is very close to what is done in section 5 of \cite{EKYY13}, while the second term is analyzed by exploring the quadratic behaviour of the Fourier transform of $S_{x0}$ for small momenta (see \eqref{low-p exp}). 

\begin{proposition}\label{prop:fourier_analysis}
Let $\tilde{\cal E}_{xy}$ be defined as in \eqref{eq:tilde_E}. Then
\begin{align*}
\sup_{x,y} |\tilde{\cal E}_{xy}| \prec  \frac{1}{\sqrt{L}}\min \bigg( \frac{L^{2}}{W^2}, \frac{L}{W \sqrt{\eta}} \bigg) \sup_{y} \sup_{p \neq 0} |\langle \b e(p), \cal E \rangle_y| + \sup_{x,y} | \cal E_{xy}|,
\end{align*}
%while for $d = 2$ one has
%\begin{align}
%\sup_{x,y} |\tilde{\cal E}_{xy}| \prec  \frac{\sqrt N}{M} \sup_{y} \sup_{p \neq 0} |\langle \b e(p), \cal E \rangle_y| + \sup_{x,y} | \cal E_{xy}|,
%\end{align}
where $ \langle \b e(p), \cal E \rangle_y = \sum_{x \in \bb T_L} e_{-x}(p)\cal E_{xy}$. Here $ \cal E_{xy}$ is regarded as a vector in $x$, while the $y$'s are just considered as parameters.
\end{proposition}

\begin{proof}
We will use this trivial bound: let $ \b w \in \C^{L}$, then
\begin{align}\label{eq:trivial_norm_bound}
\norm{\b w}_{\infty} \leq \norm{\b w}_2 \leq \sqrt{L} \norm{\b w}_{\infty}.
\end{align}

Let $Q = (Q_{xy}, \,\, x,y \in \bb T_L)$ be a $L \times L$ matrix with translation invariant and $L$-periodic entries, i.e.\ there exists a function $q: \bb T_L \to \C$ such that $q_x = Q_{x0}$. We now specify the form of $ q_x$ via its discrete Fourier transform: for $p \in \bb T_L^*$ we set
\begin{align*}
\widehat q(p) = 1-\chi(pW),
\end{align*}
where $ \chi \in C^{\infty}(\R)$ is a bump function such that $ \chi(r)=1$ for $ |r|^2 \leq 1$ and $ \chi(r) =0$ for $ |r|^2 \geq 2$. Basically, $\widehat q(p) $ is a smoothed version of the indicator function ${\bf 1}(p > W^{-1})$. Furthermore, we introduce the notation $w_x := (\ol \Pi \cal E)_{xy}$, where we regard the index $y$ as a parameter. 

Note that $ Q$ and $ S$ commute because they are translation invariant. Hence the error term \eqref{eq:tilde_E} can be written as
\begin{align}\label{eq:split_error}
\tilde{\cal E} = \frac{I- Q}{I-|\mathfrak{m}|^2 S} \b w +\frac{ Q}{I-|\mathfrak{m}|^2 S} \b w.
\end{align}
From the definition of $\widehat q(p) $, one can easily see that the first term on the right hand side of \eqref{eq:split_error} is the small Fourier modes contribution, while the second one is the large Fourier modes contribution. 
%
%\begin{align}
%\bigg \Vert \frac{ Q + P}{1-|\mathfrak{m}|^2 S} \overline\Pi \mathbf v \bigg\Vert_{\infty} = \bigg \Vert \frac{Q + P}{1-|\mathfrak{m}|^2 S} \mathbf v \bigg \Vert_{\infty} \leq \bigg \Vert \frac{ Q}{1-|\mathfrak{m}|^2 S} \mathbf v  \bigg\Vert_{\infty } + \bigg \Vert \frac{ P }{1-|\mathfrak{m}|^2 S} \mathbf v  \bigg\Vert_{\infty }
%\end{align}
Let us analyze the large modes term: its $\ell^{\infty}$-norm can be bounded as follows
\begin{align*}
\bigg \Vert \frac{ Q}{I-|\mathfrak{m}|^2S} \b w \bigg \Vert_{\infty} \leq &\sum_{k=0}^{K-1} |\mathfrak{m}|^{2k}\norm{ QS^k \b w}_{\infty}+\sum_{k=K}^{\infty} |\mathfrak{m}|^{2k}\norm{ QS^k \b w}_{\infty} \\\nonumber
\leq & \sum_{k=0}^{K-1} |\mathfrak{m}|^{2k}\norm{ QS^k}_{\infty \to \infty}\norm{\b w}_{\infty}+\sum_{k=K}^{\infty} |\mathfrak{m}|^{2k}\norm{ QS^k \b w}_{2} \\\nonumber
\leq & \sum_{k=0}^{K-1} |\mathfrak{m}|^{2k}\norm{ QS^k}_{\infty \to \infty}\norm{\b w}_{\infty}+\sqrt{L}\sum_{k=K}^{\infty} |\mathfrak{m}|^{2k}\norm{ QS^k}_{2 \to 2} \norm{\b w}_{\infty}
\end{align*}
where we used \eqref{eq:trivial_norm_bound} and $ K \in \N$ is going to be chosen later. Here we denoted the $\ell^{\infty} \to \ell^{\infty}$ norm of a matrix $A$ by $\norm{A}_{\infty \to \infty} = \max_i \sum_j |A_{ij}|$ and the Euclidean matrix norm by $\norm{A}_{2 \to 2}$.

We observe that
\begin{align}\label{eq:norm_infinity_est}
&\norm{S^k Q}_{\infty \to \infty} = O(\log L) \\
\label{eq:norm_2_est}
& \norm{S Q^{1/k}}_{2 \to 2} \leq 1-c
\end{align}
for some small positive $c$.

Note that $ \norm{S}_{\infty \to \infty} \leq 1$ since $ \sup_x |\sum_y S_{xy} v_y| \leq \sup_y |v_y| = \norm{\b v}_{\infty} $ for any $\b v \in \C^L$. Thus, to prove \eqref{eq:norm_infinity_est} we need to show that $ \norm{ Q}_{\infty \to \infty} = \norm{ q}_1 = O(\log L) $. 

To do this, we consider $u(r)= 1-\chi(rW)$ as smooth function on the torus, i.e.\ $r \in \tilde{\bb T}=  [-\pi, \pi]$, with Fourier coefficients
\begin{align*}
\widehat u(n):= \frac{1}{2\pi}\int_{\tilde{\bb T}} \dd r \, \ee^{\ii r n} u(r), \ \ \ n \in \Z\,.
\end{align*}
Since $u(r)$ is smooth, then for any $\ell \geq 0$ we have
\begin{align}\label{eq:fast_decay}
|\widehat u(n)| \leq C \Big(\ind{n = 0} + \frac{\ind{n \neq 0}}{W} \Big| \frac{W}{n}\Big|^{\ell} \Big). 
\end{align}
Note that $\widehat q(p)= u(p)$ for $p \in \bb T_L^*$, then we can write 
\begin{align}\label{eq:q_representation}
 q_{x} = \frac{1}{L}\sum_{p \in \bb T_L^*} \ee^{\ii p x}\widehat{ q}(p) = \frac{1}{L} \sum_{p \in \bb T_L^*} \ee^{\ii p x} u(p) = \sum_{m \in \Z} \widehat u(x + L m). 
\end{align}
Using \eqref{eq:fast_decay} and \eqref{eq:q_representation} and approximating the sums via integrals, one can see that
\begin{align*}
\norm{ q}_1 & \leq \sum_{x \in \bb T_L} \sum_{m \in \Z} |\widehat u(x + Lm)| = O(\log L),
% \\\nonumber
%& \leq \sum_{x \in \bb T_L} |u(x )| + \bigg(\sum_{x \in \bb T_L^d} \sup_{\substack{m \in \Z^d \\ m \neq 0}}|u(x + Lm)|^{1/2} \bigg)  \bigg( \sum_{\substack{m \in \Z^d \\ m \neq 0}} |u(x + Lm)|^{1/2} \bigg) \\\nonumber
%& \prec \frac{1}{W} \int_0^{L} \dd t \Big(1 + \Big| \frac{t}{W}\Big|\Big)^{-1} + \frac{1}{W} \int_0^{\frac{3}{2} L} \dd t \Big(1 + \Big| \frac{t}{W}\Big|^2 \Big)^{-1/2} \int_0^{\infty} \dd s \Big( 1 + \Big| \frac{L s}{W}\Big|^4 \Big)^{-1/2} \\\nonumber
%& \prec (1 + (W/L))\log (L/W) \prec  \log (L/W).
\end{align*}
which shows \eqref{eq:norm_infinity_est}. 
%
%\begin{align*}
%\chi(q) \leq C \ee^{-cq^2}
%\end{align*}
%then, if we denote by $ \tilde \chi$ the inverse Fourier transform of $ \chi$, we get that $\norm{ \tilde \chi}_1 $ is finite (check this). Thus
%
%\begin{align}
%\norm{Q}_1 = L^{-1}\sum_{x \in \Lambda}\bigg| \sum_{p \in \Gamma^*}(1-\chi(pW) )\ee^{+\ci px} \bigg| \leq 1+ \norm{\tilde \chi}.
%\end{align}
%(Chek the details of this argument).

On the other hand, to prove \eqref{eq:norm_2_est} we observe that a discrete Fourier analysis argument yields for some small positive $c$ that
\begin{align}\label{eq:below_bound}
\ind {|p|^2 \geq W^{-2}}(1-\widehat{s}(p)) &  
\geq c \ind {|p|^2 \geq W^{-2}},
\end{align}
 where $\widehat s$ is the Fourier transform of the function $s_x: = S_{x0}$ for $x \in \bb T_L$.
Hence, from \eqref{eq:below_bound} we get
\begin{align*}
\norm{S Q^{1/k}}_{2 \to 2} =\sup_{p \in \bb T_L^*} | \widehat{s}(p) (1-\chi(pW))^{1/k}| \leq \sup_{p \in \bb T_L^*} | \widehat{s}(p) \ind {|p|^2 \geq W^{-2}}| \leq 1-c.
\end{align*}
Thus, by \eqref{eq:norm_2_est} and \eqref{eq:norm_infinity_est} we obtain
\begin{align*}
\bigg \Vert \frac{ Q}{I-|\mathfrak{m}|^2S} \mathbf w \bigg \Vert_{\infty} 
\leq \norm{\mathbf w}_{\infty}\bigg[O(K \log L)+\sqrt{L}\sum_{k=K}^{\infty} (1-c)^k \bigg].
\end{align*}
By summing the geometric series and choosing $ K = C  \log L $ for some sufficiently large constant $C$ we get
% $ A >1$ that will be determined later as a function of $ c$, we get
%\begin{align}
%\bigg \Vert \frac{ Q}{1-|\mathfrak{m}|^2S}\overline{\Pi}\mathbf v \bigg \Vert_{\infty} 
%\leq & \norm{\mathbf v}_{\infty}\bigg[ O(\log N)+ \sqrt{N}\sum_{k=\lfloor  A \log L \rfloor}^{\infty} (1-c)^k \bigg]
%\end{align}
%
%We can analyse in more detail the sum on the right hand side:
%\begin{align}
%\sqrt{N}\sum_{k=\lfloor  A\log N \rfloor}^{\infty} (1-c)^k = \sqrt{N}\frac{(1-c)^{\lfloor A \log N \rfloor}}{c} \leq \sqrt{N}\frac{(1-c)^{-1}(1-c)^{ A \log N }}{c} = \frac{N^{1/2+A\log(1-c)}}{c(1-c)} 
%\end{align}
%therefore, in order this term to be $ O(1)$, we just need to choose $ A \geq - \frac{1}{2 \log(1-c)}$. This means that
\begin{align}\label{eq:Q_estimate}
\bigg \Vert \frac{ Q}{1-|\mathfrak{m}|^2S} \b w \bigg \Vert_{\infty} 
= O(\log L)^2 \norm{\mathbf w}_{\infty} \prec \norm{\mathbf w}_{\infty}.
\end{align}
We now need to estimate the second term in \eqref{eq:split_error}, i.e.\ the small modes contribution, which can be written as
\begin{align*}
\bigg( \frac{ I-Q}{I-|\mathfrak{m}|^2S}\overline{\Pi} \cal E \bigg)_{xy}
=   \sum_{\substack{p \in \bb T_L^* \\ p \neq 0}} \frac{\chi(pW)}{1-|\mathfrak{m}|^2 \widehat{s}(p)} e_x(p) \langle \b e(p), \cal E \rangle_y 
\end{align*}
where we recall that 
%used the Fourier space respresentation and we regarded $ \cal E_{xy}$ as a vector in $x$, while the $y$'s are just considered as parameters: explicitely 
$ \langle \b e(p), \cal E \rangle_y = \sum_{x \in \bb T_L} e_{-x}(p)\cal E_{xy}$. Using \eqref{m3} and the constraint on the momentum $ p$ given by $ \widehat{P}(p)$ and $ 1-\widehat s (p) \geq cW^2p^2$ for $ (pW)^2 \leq 2$, one can see that
\begin{align}\label{eq:P_estimate}
\Bigg\Vert \sum_{\substack{p \in \bb T_L^* \\ p \neq 0}} \frac{\chi(pW)}{1-|\mathfrak{m}|^2 \widehat{s}(p)} \b e(p) \langle \b e(p), \cal E \rangle_y \Bigg \Vert_{\infty} & \prec \frac{1}{\sqrt L}  \sup_{p \neq 0} |\langle \b e(p), \cal E \rangle_y| \sum_{\substack{p \in \bb T_L^* \\ p \neq 0}} \frac{\chi(pW)}{\eta + (pW)^2}.
\end{align} 
Note that 
\begin{align}\label{eq:est1}
 \sum_{\substack{p \in \bb T_L^* \\ p \neq 0}} \frac{\chi(pW)}{\eta + W^2 |p|^2} \prec \sum_{\substack{j \in \Z \\0 < |j| \leq L/W}} \frac{1}{\eta + ( \frac{W|j|}{L} )^2} \prec \frac{L^2}{W^2}\sum_{j = 1}^{\infty} |j|^{-2} \prec \frac{L^2}{W^2}.
\end{align}
On the other hand, by estimating the sum with an integral we get
\begin{align}\label{eq:est2}
\sum_{\substack{j \in \Z \\0 < |j| \leq L/W}} \frac{1}{\eta + ( \frac{W|j|}{L} )^2} \prec \int_{0}^{L/W}\frac{\dd x}{\eta + ( \frac{W x}{L} )^2} \prec \frac{ L}{W \sqrt{\eta}}\,.
\end{align}
Thus, by using  \eqref{eq:Q_estimate}, \eqref{eq:P_estimate}, \eqref{eq:est1} and \eqref{eq:est2} we obtain
\begin{align}\label{eq:err_preliminary}
\sup_{x,y} |\tilde{\cal E}_{xy}| \prec \frac{1}{\sqrt{L}} \min \bigg( \frac{L^{2}}{W^2}, \frac{L}{W \sqrt{\eta}} \bigg) \sup_{y} \sup_{p \neq 0} |\langle \b e(p), \cal E \rangle_y| + \sup_{x,y} | \cal E_{xy}|\,.
\end{align}
\end{proof}

\begin{remark} 
For $T'_{xy} = \sum_i |G_{xi}|^2 S_{iy}$ an analogous result holds: it satisfies the self-consistent equation 
$$
T' = |\fra m|^2 T' S + |\fra m|^2 S + \cal E',
$$
so that one gets
$$
T'= O_{\prec} \bigg( \frac{1}{L \sqrt{W \eta^3}} \bigg) + \Theta + \tilde{\cal E}', \ \ \ \  \tilde{\cal E}' = \cal E' \ol \Pi(I - |\fra m|^2 S)^{-1}.
$$
In the proof of Proposition \ref{prop:fourier_analysis} we regarded $\cal E_{xy}$ as a column vector indexed by $x$ and we regarded $y$ as a parameter. Conversely, the analogous statement for $T'$ can be obtained with the same Fourier analysis argument by seeing $\cal E'_{xy}$ as a row vector indexed by $y$ with $x$ as parameter.
\end{remark}

\begin{remark} 
Note that in \cite{EKYY13} the authors derive the bound 
\begin{align}\label{eq:old_err}
\sup_{x,y} |\tilde{ \cal E}_{xy}| \prec \min \bigg( \frac{L^{2}}{W^2}, \frac{L}{W \sqrt{\eta}} \bigg) \sup_{x,y} | \cal E_{xy}|
\end{align}
and then use the fluctuation averaging estimates based on the Schur's complement formula to control $\sup_{x,y} | \cal E_{xy}|$. In \eqref{eq:err_preliminary} we see that in the first term, i.e.\ the small modes contribution, we gain a prefactor $ L^{-1/2}$ compared to the prefactor in \eqref{eq:old_err}, but we have to estimate $ \sup_{y} \sup_{p \neq 0} |\langle \b e(p), \cal E \rangle_y| $ instead of just $ \sup_{x,y} | \cal E_{xy}|$. 

On the other hand, the second term in \eqref{eq:err_preliminary}, corresponding to the large Fourier modes, is going to be always subdominant compared to the first one, since its prefactor is just 1.
\end{remark}

In the following section we discuss how to estimate 
$ \sup_{y} \sup_{p \neq 0} |\langle \b e(p), \cal E \rangle_y|$ and $\sup_{x,y} | \cal E_{xy}|$ and we prove the main results.
As already mentioned in the Introduction, for this task we will avoid using cumbersome expansions based on the Schur's complement formula, but we will rather employ the cumulant expansion method.

\section{Proofs of the main results}

Let us define the following family of vectors:
\begin{align*}
\bb V := \left\lbrace \b v \in \C^L:  \ \ \norm{\mathbf v}_2 =O(1), \ \ \norm{\b v}_{\infty} = O(L^{-{1/2}}) \right\rbrace.
\end{align*}
Note that $\b e(p) \in \bb V$ for any $p \in \bb T_L^*$. Moreover, let us fix the following notation: let $ A$ be a $L \times L$ matrix and $ \b u \in \C^L$, then we will write $ A_{\b u b} := \sum_{a} u_a A_{ab}$. Thus
\begin{align}\label{eq:sup_E_vy}
\sup_{y} \sup_{p \neq 0} |\langle \b e (p), \cal E \rangle_y| \leq \sup_y \sup_{\b v \in \bb V} |\cal E_{\b v y}|.
\end{align}
 Furthermore, it is convenient to split $ \cal E_{xy}$ as $ \cal E_{xy} = \cal P_{xy} + \cal R_{xy}$ where
\begin{align}
\cal P_{xy} = & \,\,-\fra m z T_{xy} - \fra m \sum_{i,j} S_{xi} S_{ij}G_{jj}|G_{iy}|^2 - \fra m \sum_{i,j} S_{xi} S_{ij}\overline{G}_{ii}|G_{jy}|^2 - \fra m S_{xy} \overline{G}_{yy} \label{eq:def_P}
\\
\cal R_{xy} = & \,\fra m \sum_{i,j} S_{xi} S_{ij}(G_{jj}-\mathfrak{m})|G_{iy}|^2 + \fra m \sum_{i,j} S_{xi} S_{ij}(\overline{G}_{ii}-\overline{\mathfrak{m}})|G_{jy}|^2 + \fra m S_{xy} (\overline{G}_{yy}-\overline{\mathfrak{m}}). 
\label{eq:def_R}
\end{align}
By using \eqref{eq:m_equation} in \eqref{eq:def_P} we readily see that indeed
$$
\cal P_{xy} = (1 + \fra m^2) T_{xy} - \fra m \sum_{i,j} S_{xi} S_{ij}G_{jj}|G_{iy}|^2 - \fra m \sum_{i,j} S_{xi} S_{ij}\overline{G}_{ii}|G_{jy}|^2 - \fra m S_{xy} \overline{G}_{yy} = \cal E_{xy} - \cal R_{xy}.
$$ 
The reason why we chose to write $ \cal E_{xy}$ in this complicated way is that, applying the cumulant expansion directly to $ \E \cal E_{xy}$, it is not possible to get any cancellation even at the level of the expectation, while for $ \cal P_{xy}$ one gets $\E \cal P_{xy} = 0$ when $ H_{ij}$ are Gaussian. 

To see that, we recall the basic definition of the Green function $ G(z)=(H - z)^{-1} $ which amounts to
\begin{align}\label{eq:def_G}
z G = H G - z I.
\end{align}
We set the notation $\partial_{ij} g(H) := \frac{\partial}{\partial H_{ij}} g(H)$ and we recall also the differentiation rule when $H$ is complex Hermitian:
\begin{align}\label{eq:diff_G}
\partial_{kl} G_{ij} =  - G_{ik} G_{lj}.
\end{align}
Note that, assuming $H_{ij}$ to be Gaussian with $\E H_{ij}^2 =0$, in the cumulant expansion formula \eqref{eq:cumulant_expansion} only the term dependent on the variance $\E |H_{ij}|^2$ survives:
\begin{align}\label{eq:cum_exp_gaussian}
\E H_{ij} f(H_{ij}, H_{ji}) = \E |H_{ij}|^2 \E \partial_{ji} f(H_{ij}, H_{ji}) = S_{ij} \E \partial_{ji}f(H_{ij}, H_{ji}).
\end{align}
Thus, from \eqref{eq:def_G}, \eqref{eq:diff_G} and \eqref{eq:cum_exp_gaussian}, we easily get
\begin{align*}
z\, \E T_{xy} + S_{xy} \E \ol G_{yy} & = \sum_{i} S_{xi}(zG_{iy})\ol G_{iy} = \E \sum_{i,j} S_{xi} H_{ij} G_{jy} \ol G_{iy} = \E \sum_{i,j} S_{xi} S_{ij} \partial_{ji} (G_{jy} \ol G_{iy}) \\
& = -\E \sum_{i,j} S_{xi} S_{ij} (G_{jj} |G_{iy}|^2 + \ol G_{ii} |G_{jy}|^2)
\end{align*}
which implies that $ \E \cal P_{xy} = 0$. 
As we will see, remarkable cancellations occur also for the moments of $ \cal P_{xy}$ even when $H$ is non Gaussian and $\E H_{ij}^2 \neq 0$, allowing us to control $ \cal P_{xy}$ and $ \cal P_{\b v y}$. 

In the following, we shall call \emph{control parameter} any positive and deterministic quantity $\Psi^{(N)}(z)$ and we shall call \emph{admissible} any control parameter $\Psi^{(N)}(z)$ such that
\begin{equation} \label{admissible Psi}
M^{-1/2} \;\leq\; \Psi^{(N)}(z) \;\leq\; M^{-\gamma/2}
\end{equation}
for all $N$ and $z \in \b S$, where $\gamma$ is the same fixed number as in \eqref{def:S}.
A typical example of an admissible control parameter is $\Psi(z) = \frac{1}{\sqrt{M \eta}}$.

The proposition below states precisely the bounds that we obtain for $ \cal P_{xy}$ and $ \cal P_{\b v y}$, as well as for $\cal R_{\b vy}$ and $ \cal R_{xy}$: it is the technical core of the paper and it replaces the fluctuation averaging estimates of Proposition 3.9 in \cite{EKYY13}. 

\begin{proposition}\label{prop:main_estimates}
Let $\Psi$ be an admissible control parameter as defined in \eqref{admissible Psi} such that $ \Lambda \prec \Psi$ and $ \b v \in \bb V$.
Then the following bounds hold true:
\begin{align} \label{eq:P_vy}
\cal P_{\mathbf v y} & \prec \frac{\Psi^3}{\sqrt{\eta}} + \frac{\Psi}{\sqrt{L \eta}}\,, \\
\label{eq:P_xy}
\cal P_{xy} & \prec \frac{\Psi^3}{W} \sqrt{\frac{L}{\eta}} + \frac{\Psi}{W \sqrt{\eta}}\,, \\ 
\label{eq:R_vy}
\cal R_{\mathbf v x} & \prec \frac{\Psi^2}{\eta \sqrt{L}} + \frac{\Psi}{\sqrt{L}}\,, \\
\label{eq:R_xy}
\cal R_{xy} & \prec \frac{\Psi^2}{W \eta} + \frac{\Psi}{W}\,.
\end{align}
\end{proposition}
\begin{proof}
See Section \ref{sec:proof_main_estimates}.
\end{proof}
%\textcolor{red}{We might get more refined estimates for $ P_{xy}$ and $ R_{xy}$:
%\begin{align}
%P_{xy} \prec \frac{\Psi^3}{\sqrt{M N \eta}} + \frac{\Psi}{M \sqrt{\eta}}\, \ \ \ R_{xy} \prec \Psi^4 + \frac{\Psi^2}{\sqrt{M}},
%\end{align}
%but they would not affect the final results.
%}

To prove the main results we will need also a couple of auxiliary lemmata which show how to combine apriori bounds on $\Lambda$ and $T$ to get a better estimate for $\Lambda$. Note that the first one is basically Lemma 5.3 in \cite{EKYY13} and in the Appendix we give a new proof which does not rely on the averaging fluctuation estimates of \cite{EKY13}.
\begin{lemma}\label{lemma:improv_bound}
Suppose that $\Lambda \prec \Psi$ for some admissible parameter $\Psi$ and $T_{ab}, T'_{ab} \prec \Omega^2_{ab}$ for a family of admissible control parameters $\Omega_{ab}$ indexed by a pair $(a,b)$. Then 
\begin{align}\label{eq:app1}
|G_{ab}- \fra m \, \delta_{ab}|^2 \prec \Omega_{ab}^2 + \Psi^4.
\end{align}
%Moreover, setting $\tilde \Omega = \sup_{ab} \Omega_{ab}$ and assuming that $\tilde \Omega$ is admissible, one has
%\begin{align}\label{eq:app2}
%\Lambda^2 \prec \tilde \Omega^2.
%\end{align}
\end{lemma}

\begin{proof}
See Appendix.
\end{proof}

The second lemma implements the idea of self improving bounds: we start with a rough bound and we improve it by a recursive procedure. 
\begin{lemma}\label{lemma:self-improving_bound}
Suppose that 
\begin{align*}
& \Lambda \prec \Psi, \ \ \ \Omega \leq \Psi \leq \tilde \Psi; \ \ \ \ \ \ T_{ij}, T_{ij}' \prec \Omega^2 + \sum_{k = 1}^K a_k \Psi^k
\end{align*}
where $ \Psi$, $\tilde \Psi$, $\Omega$ are admissible control parameters and $ K$ is some fixed integer. Assume also that $ a_k \geq 0$ for $ k = 1, \ldots, K$ and
\begin{align}\label{eq:a_coeff_condition}
a_1 \Omega^{-1} \ll 1; \ \ \ \ \ a_{\ell} \Psi^{\ell-2} \ll 1, \ \ 2 \leq \ell \leq K.
\end{align}
Then $ \Lambda \prec \Omega$.
\end{lemma}

\begin{proof} Set $\Xi^2 = \Omega^2 + \sum_{k = 1}^K a_k \Psi^k$. From Lemma \ref{lemma:improv_bound}, we easily deduce the implication
\begin{align}\label{eq:app_iteration}
\Lambda \prec \Psi \ \ \ \Longrightarrow \ \ \ \sup_{a,b}|G_{ab} - \fra m \, \delta_{ab}| = \Lambda \prec \Xi + \Psi^2.
\end{align}
After $ k$ iterations of \eqref{eq:app_iteration} we get $ \Lambda \prec \Xi + \Psi^{2^k}$. Since $ \Xi$ and $ \Psi$ are both admissible control parameters, taking $ k \sim |\log \gamma|$ we get the implication
\begin{align}\label{eq:self-improving_bound}
\Lambda^2 \prec \Psi^2 \ \ \ \Longrightarrow \ \ \  \Lambda^2 \prec \Xi^2 =\Omega^2 + \sum_{k = 1}^K a_k \Psi^k.
\end{align}
We can iterate \eqref{eq:self-improving_bound} by defining the recursion relation
\begin{align*}
\Psi_{i+1}^2 := \Omega^2 + \sum_{k = 1}^K a_k \Psi_i^k, \ \ \ \Psi_0 = \tilde \Psi.
\end{align*}
Thus, \eqref{eq:self-improving_bound} implies that $ \Lambda^2 \prec \Psi^2_i$ for any fixed $ i$. The conditions \eqref{eq:a_coeff_condition} and the fact that $ \Omega$ is admissible imply that there is a finite integer $ i$ (depending on the implicit constants involved in the relation ``$ \ll$''), such that $ \Psi^2_i \prec \Omega^2$.
\end{proof}

We are now ready to prove our main results. 

\subsection{Proof of Theorem \ref{th:local_law_d=1}}

From \eqref{eq:sup_E_vy} and propositions \ref{prop:fourier_analysis} and \ref{prop:main_estimates} we get 
\begin{align}\nonumber
\sup_{x,y}|\tilde {\cal E}_{xy}| & \prec \frac{L^{3/2}}{W^2} \sup_{y}\sup_{\b v \in \bb V}(|\cal P_{\b vy}| + |\cal R_{\b vy}|) + \sup_{x,y}|\cal P_{xy}| + \sup_{x,y}|\cal R_{xy}|) \\
 & \prec    \, \frac{L^2}{W^2} \bigg(\frac{\Psi^2}{L \eta} + \frac{\Psi^3}{\sqrt{L \eta}} + \frac{\Psi}{L \sqrt{\eta}} \bigg) + \frac{ \Psi^2}{ W\eta} + \frac{\Psi^3} {W} \sqrt{\frac{L}{\eta}} + \frac{ \Psi}{ W \sqrt{\eta}}\,. \label{eq:error}
\end{align}
By \eqref{eq:T_solved2}, \eqref{prof1} and \eqref{eq:error} we have
\begin{align*}
T_{xy} \prec \Phi^2 + \frac{L^2}{W^2} \bigg(\frac{\Psi^2}{L \eta} + \frac{\Psi^3}{\sqrt{L \eta}} + \frac{\Psi}{L \sqrt{\eta}} \bigg) + \frac{ \Psi^2}{ W\eta} + \frac{\Psi^3} {W} \sqrt{\frac{L}{\eta}} + \frac{ \Psi}{ W \sqrt{\eta}} \,.
\end{align*}
To finish the proof we apply Lemma \ref{lemma:self-improving_bound} with $ \Omega = \Phi$ and, thanks to Lemma \ref{lm:lsc}, $ \tilde \Psi = (W \eta)^{-1/2}$. In this setting we have $\Psi^{-2} \leq L \eta + W \sqrt \eta $. The conditions \eqref{cond on N eta} in Theorem \ref{th:local_law_d=1} arise from the assumption \eqref{eq:a_coeff_condition} in Lemma \ref{lemma:self-improving_bound}.

%
%with initial input $ \Psi_0 = (W\eta)^{-1/2}$.
%Let us now list the conditions induced by the terms $ \epsilon_i$ once they are inserted in the recursive relation with initial condition $ \Psi_0 = (W\eta)^{-1/2}$ and we assume that
%
%\begin{align}
%& \epsilon_1 \Longrightarrow \eta \gg \frac{L}{W^2}     \\\nonumber
%& \epsilon_2 \Longrightarrow \eta \gg \frac{L^4}{W^6} \Longrightarrow  \eta \gg \frac{L^{3/2}}{W^{5/2}}, \ \ L \ll W^{4/3}   \\\nonumber
%& \epsilon_3 \Longrightarrow  L \ll W^{4/3}    \\\nonumber
%& \epsilon_4 \Longrightarrow  \eta \gg  \frac{L^{3/2}}{W^{5/2}}   \\\nonumber
%& \epsilon_5 \Longrightarrow  \eta \ll \frac{W^3}{L^2}    \\\nonumber
%& \epsilon_6 \Longrightarrow \eta \gg W^{-1} \\\nonumber
%& \epsilon_7 \Longrightarrow W \gg 1 \\\nonumber
%& \epsilon_8 \Longrightarrow \eta \gg W^{-1} \\\nonumber
%& \epsilon_{9} \Longrightarrow \eta \gg \frac{L^{1/2}}{W^{3/2}} \\\nonumber
%\end{align}

%\textcolor{red}{Maybe there is some hope to still improve the range of validity of Theorem 2.2 in \cite{EKYY13} by replacing the condition $ \eta \gg \max \bigg(\frac{L^{4/3}}{W^{7/3}}, \frac{L^4}{W^6}  \bigg)$ with 
%$$
%\eta \gg \frac{L}{W^2}.
%$$ 
%This looks plausible because for Wigner matrices (i.e.\ $ S_{ij} = L^{-1}$) one has
%$$
%\epsilon_4 = \epsilon_5 = \frac{1}{(L\eta)^{3/2}} \Longrightarrow \eta \gg \frac{L}{W^2}.
%$$
%}

\subsection{Proof of Theorem \ref{th:diffusion_d=1}}
We are now in the range $ (W/L)^2 \leq \eta \leq 1$, therefore \eqref{eq:T_solved2}, Proposition \ref{prop:fourier_analysis} and \ref{prop:main_estimates} imply that for some admissible control parameter $ \Psi$ such that $ \Lambda \prec \Psi$
\begin{align}\label{eq:diff1}
T_{xy} - \Theta_{xy} \prec \frac{L}{W \sqrt \eta} \bigg(\frac{\Psi^2}{L \eta} + \frac{\Psi^3}{\sqrt{L \eta}} + \frac{\Psi}{L \sqrt{\eta}} \bigg) + \frac{ \Psi^2}{ W\eta} + \frac{\Psi^3} {W} \sqrt{\frac{L}{\eta}} + \frac{ \Psi}{ W \sqrt{\eta}}.
\end{align}
From Theorem \ref{th:local_law_d=1} we see that, in the range $ (W/L)^2 \leq \eta \leq 1$, we have $ \Psi = W^{-1/2} \eta^{-1/4}$ when
$$
\frac{W^2}{L^2} \gg \frac{L^{3/2}}{W^{5/2}} \ \ \ \Longleftrightarrow \ \ \ L \ll W^{9/7}.
$$
It is now easy to check from \eqref{eq:diff1} with $ \Psi = W^{-1/2} \eta^{-1/4}$, $ (W/L)^2 \leq \eta \leq 1$ and $ L \ll W^{9/7}$ that \eqref{Tprec} holds true. 
Note that \eqref{Tfin} follows by using \eqref{Tdetbound} in \eqref{Tprec}. Finally, using Lemma \ref{lemma:improv_bound} with $\Omega_{ij}^2=\Upsilon_{ij}$ and
 $ \Psi = W^{-1/2}\eta^{-1/4}$,
 we obtain
\begin{equation}\label{stro2}
\absb{G_{ij}-\delta_{ij}\fra m}^2 \;\prec\; \Upsilon_{ij} +\Psi^4
 \;\prec\;  \Upsilon_{ij} .
\end{equation}
Here we used that $\Psi^4$ can be absorbed into $ (L\eta)^{-1}\le \Upsilon_{ij}$. This proves \eqref{Tfin1}, and hence concludes the proof of Theorem \ref{th:diffusion_d=1}.

\subsection{Proof of Corollary \ref{cor:deloc_d=1}}
From \eqref{large eta estimate} we see that we need $ \Lambda^2 \prec (L\eta)^{-1}$.  From Theorem \ref{th:local_law_d=1} we know that this is true when $\eta \leq (W/L)^2$ and $\eta \gg L^{3/2} W^{-5/2}$. Therefore, we have to require that
$$
\frac{L^{3/2}}{W^{5/2}} \ll \frac{W^2}{L^2}
$$
which is true when $L \ll W^{9/7}$.

\section{Proof of Proposition \ref{prop:main_estimates}}\label{sec:proof_main_estimates}
In order to avoid useless technical complications we assume that $ H$ is Gaussian and Hermitian and that 
\begin{equation} \label{CH}
\E \zeta_{ij}^2 = 0 \quad \text{for all} \quad i < j\,.
\end{equation}
in addition  to \eqref{zetacond}.
For example, \eqref{CH} is true when the real and imaginary parts of $\zeta_{ij}$
are independent with identical variance. However, our results hold also without these assumptions and in Section \ref{sec:nonGaussian_generalComplex} we sketch how to achieve this generalization.

Let us define a family of matrices $\bb S$ such that
\begin{align} \label{6.1}
\bb S = \left\lbrace \sigma \in \C^{L \times L} : \sup_{x,y} |\sigma_{xy}| \prec W^{-1}, \ \ \sum_{x}|\sigma_{xy}| \prec 1, \ \ \sum_{y}|\sigma_{xy}| \prec 1 \right\rbrace.
\end{align}
Note that $\bb S$ is closed under matrix addition and multiplication.
In particular, we immediately see that $S \in \bb S$. The following lemma, proven in Section \ref{lemma:chain_loop_X_estimates}, collects all the necessary estimates needed to prove Proposition \ref{prop:main_estimates}. 

\begin{lemma} \label{lemma:chain_loop_X_estimates}
	Let be $ \Psi$ an admissible control parameter defined as in \eqref{admissible Psi} and such that $ \Lambda \prec \Psi$.
	\begin{itemize} 
		\item[(i)] Let $ \sigma^{(1)}, \sigma^{(2)}, \ldots, \sigma^{(n)} \in \bb S$, then for $n \geq 1$ we have
		\begin{equation}\label{eq:chain_bound}
		Y_{ab;u_1\cdots u_n}^{(n)} := \sum_{i_1,...,i_n} \sigma^{(1)}_{u_1 i_1}\cdots \sigma^{(n)}_{u_n i_n} G_{ai_1}G_{i_1i_2}\cdots G_{i_{n-1}i_n}G_{i_nb} \prec \Psi^{2n+1}+\delta_{ab}\Psi^{2n}.
		\end{equation}
		\item[(ii)] Let $ \sigma^{(1)}, \sigma^{(2)}, \ldots, \sigma^{(n)} \in \bb S$, then for $n \ge 2$ and for some $\xi \in \bb S$ with nonnegative entries we have
		\begin{equation}\label{eq:loop_bound}
		Z^{(n)}_{ab;u_3 \cdots u_n} := \sum_{i_1,...,i_n} \sigma^{(1)}_{ai_1}\sigma^{(2)}_{bi_2}\sigma^{(3)}_{u_3 i_3}\cdots \sigma^{(n)}_{u_n i_n} G_{i_1i_2}\cdots G_{i_{n-1}i_n}G_{i_ni_1} \prec \Psi^{2n}+ \xi_{ab} \Psi^{2(n-2)}.
		\end{equation}
		\item[(iii)] Let $ \sigma \in \bb S$, then we have
		\begin{align}\label{eq:X_bound}
		X_i := \sum_j \sigma_{ij} (G_{jj} - \mathfrak m) \prec \Psi^2, \ \ \ X_{ij} := \sum_k \sigma_{ik} G_{kj} \prec \Psi^2.
		\end{align}
	\end{itemize}
\end{lemma}

\begin{remark}\label{rmk:chain_loop}
	In what follows we will refer to $Y_{ab;u_1 \ldots u_n}^{(n)}$ and $Z_{ab;u_3 \ldots u_n}^{(n)}$ respectively as \emph{open chain} (or simply \emph{chain}) and \emph{loop} of order $n$. This terminology emphasizes that in $Y^{(n)}$ the extreme indices $a$ and $b$ of the product $G_{ai_1}G_{i_1i_2}\cdots G_{i_{n-1}i_n}G_{i_nb}$ are not summed over, while in $Z^{(n)}$ the extreme indices are identical and they are summed over. The order $n$ refers to the fact that both $Y^{(n)}$ and $Z^{(n)}$ involve $n$ summations. 
\end{remark}
The following lemma translates the control of arbitrary moment of a random variable into stochastic domination bounds. 
\begin{lemma}\label{lemma:moment_to_domination}
Let $ \phi $ be a random variable such that $ 0 \leq \phi \leq L^C$ for some $ C > 0$ and let $\varphi \in \R^+$ be deterministic such that $\varphi \in [L^{-C}, L^C]$. Suppose that there exists $q \in [0,1)$ such that for any deterministic $\vartheta \in [\varphi, L^C]$ and any $p \in \N$ one has the implication
\begin{align}\label{eq:moment_estimate}
\phi \prec \vartheta \quad \Longrightarrow \quad \E|\phi|^{2p} = \sum_{k=1}^{2p}O_{\prec}((\vartheta^q \varphi^{1-q})^k) \E|\phi|^{2p-k},
\end{align} 
then $\phi \prec \varphi$. 
\end{lemma}

\begin{proof}
Applying H\"older inequality to \eqref{eq:moment_estimate} one gets
\begin{align*}
\E|\phi|^{2p} \leq \sum_{k=1}^{2p}O_{\prec}((\vartheta^q \varphi^{1-q})^k)(\E|\phi|^{2p})^{\frac{2p-k}{2p}}
\end{align*}
which implies that 
\begin{align*}
\E|\phi|^{2p} = O_{\prec}((\vartheta^q \varphi^{1-q})^{2p}).
\end{align*}
Then from Markov inequality we deduce the implication
\begin{align*}
\phi \prec \vartheta \, \, \Longrightarrow \,\, \phi \prec \vartheta^p \varphi^{1-p}.
\end{align*}
By invoking Lemma 2.6 in \cite{HKR17}, we conclude the proof.
\end{proof}

We are now ready to prove Proposition \ref{prop:main_estimates}. We will first prove \eqref{eq:P_vy}.

{\bf Step 1.} We recall that by spectral decomposition of $G$ we  can easily get the so called Ward identity
\begin{align}\label{eq:ward_identiy}
\sum_{i} |G_{xi}|^2 = \frac{\mathrm{Im} G_{xx}}{\eta}.
\end{align}
%Green function $ G$:
%\begin{align}\label{eq:def_G}
%G = (H - z)^{-1} \ \ \ \Longleftrightarrow \ \ \ z G = H G - z \id.
%\end{align}
%Setting the notation
%\begin{align}
%\partial_{ij} g(H) := \frac{\partial}{\partial H_{ij}} g(H)
%\end{align}
%we recall also the differentiation rule when $H$ is Hermitian:
%\begin{align}\label{eq:diff_G}
%\partial_{kl} G_{ij} =  - G_{ik} G_{lj},
%\end{align}
%
%and the cumulant expansion formula for Hermitian and Gaussian $ H$
%\begin{align}\label{eq:cum_exp_gaussian}
%\E H_{ij} f(H_{ij}, H_{ji}) = \E |H_{ij}|^2 \E \partial_{ji} f(H_{ij}, H_{ji}) = S_{ij} \E \partial_{ji}f(H_{ij}, H_{ji}).
%\end{align}
Let us now define
\begin{equation} \label{eq:def_Q}
\cal Q_{xy}=\sum_{i,j} S_{\b vi}(H_{ij}G_{jx}\ol{G}_{iy}+S_{ij}G_{jj}G_{ix}\ol{G}_{iy}+S_{ij}\ol G_{ii}G_{jx}\ol{G}_{jy})\,,
\end{equation} 
and accordingly,
\begin{equation*}
\ol{ \cal Q}_{xy}=\sum_{i,j} \ol {S}_{\b vj}(H_{ij} G_{jy}\ol {G}_{ix}+S_{ji}G_{jj}\ol G_{ix}{G}_{iy}+S_{ji}\ol G_{ii}\ol G_{jx}{G}_{jy})\,,
\end{equation*} 
where we recall that $ S_{\b v i} = \sum_x v_x S_{xi}$ and $ \norm{\b v}_{\infty} = O(L^{-1/2})$.
From \eqref{eq:def_P} we see that $\cal P_{\b vy}=-\fra m\cal Q_{yy}$, and by \eqref{eq:diff_G} we have the derivatives
\begin{equation*}
\partial_{ji} \cal Q_{yy} =-\cal Q_{jy}G_{iy}-\cal Q_{yi}\ol G_{jy}+S_{\b vj}G_{iy}\ol G_{jy}-\sum_{k,l}S_{\b vk}S_{kl}G_{lj}G_{il}|G_{ky}|^2-\sum_{k,l}S_{\b v k}S_{kl}\ol G_{ki}\ol G_{jk}|G_{ly}|^2\,,
\end{equation*}
and
\begin{equation*}
\partial_{ji} \ol{\cal Q}_{yy} =-\ol{\cal Q}_{iy}\ol G_{jy}-\ol{\cal Q}_{yj} G_{iy}+\ol S_{\b vj}G_{iy}\ol G_{jy}-\sum_{k,l}\ol S_{\b vl}S_{lk} G_{lj} G_{il}|G_{ky}|^2-\sum_{k,l}\ol S_{\b vl}S_{lk} \ol G_{ki}\ol G_{jk}|G_{ly}|^2\,.
\end{equation*}
Now we fix $p \ge 2$, and by cumulant formula \eqref{eq:cum_exp_gaussian} we have
\begin{equation} \label{eq:Q^2p}
\begin{aligned}
\bb E |\cal Q_{yy}|^{2p}=&\,\bb E \sum_{i_1,j_1} S_{\b vi_1}(H_{i_1j_1}G_{j_1y}\ol G_{i_1y}+S_{i_1j_1}G_{j_1j_1}|G_{i_1y}|^2+S_{i_1j_1}\ol G_{i_1i_1}|G_{j_1y}|^2)\cdot \cal Q_{yy}^{p-1}\ol{\cal Q}_{yy}^p\\
=&\,\bb E \sum_{i_1,j_1} S_{\b vi_1}S_{i_1j_1}G_{j_1y}\ol G_{i_1y}[(p-1) (\partial_{j_1i_1} \cal Q_{yy}) \cal Q_{yy}^{p-2}\ol{\cal Q}_{yy}^p+p (\partial_{j_1i_1} \ol{\cal Q}_{yy}) |\cal Q_{yy}|^{2p-2}]\,.
\end{aligned}
\end{equation}
Now we would like to compute the second line of \eqref{eq:Q^2p} by recursively applying cumulant formula. To this end, we define for each $m \in \{2,3,...,2p-1\}$ the set
\begin{equation} \label{6.15}
\cal V_{m-1,m}=\{\,G_{j_{m}j_{m-1}}G_{i_{m-1}y}\ol G_{i_my},\,\,\ol G_{i_mi_{m-1}}\ol G_{j_{m-1}y}G_{j_my}\}\,.
\end{equation}
For $n \in \{1,2,...,2p-1\}$, let us consider
\begin{equation} \label{eq:Q_term}
\Big(\frac{1}{\sqrt{L}}\Big)^n\bb E \sum_{i_1,j_1,...,i_n,j_n} \sigma^{(1)}_{i_1j_1}\cdots \sigma^{(n)}_{i_nj_n}G_{j_1y}\ol G_{i_1y}V_{1,2}\cdots V_{n-1,n}(\partial_{j_ni_n} \cal Q_{yy})\cal Q_{yy}^{\alpha}\ol{\cal Q}_{yy}^{\beta}\,,
\end{equation}
where $V_{m-1,m} \in \cal V_{m-1,m}$ for $m=2,3,...,n$, $\sigma^{(1)},...,\sigma^{(n)}\in \bb S$, and $\alpha+\beta=2p-n-1$. Formula \eqref{eq:Q_term} is one of the terms produced by applying $n$ times the cumulant expansion to $ \bb E |\cal Q_{yy}|^{2p}$: for instance, if for $n=1$ we set $\sigma^{(1)}\deq \sqrt{L}S_{\b v i_1}S_{i_1j_1}$, \eqref{eq:Q_term} corresponds to the first term on the second line of \eqref{eq:Q^2p}. By the differential rule \eqref{eq:def_Q} we see that \eqref{eq:Q_term} becomes 
\begin{equation}\label{eq:Q_term2}
\begin{aligned} 
& \Big(\frac{1}{\sqrt{L}}\Big)^n\bb E \sum_{i_1,j_1,...,i_n,j_n}\sigma^{(1)}_{i_1j_1}\cdots \sigma^{(n)}_{i_nj_n}G_{j_1y}\ol G_{i_1y}V_{1,2}\cdots V_{n-1,n} \cdot \cal Q_{yy}^{\alpha}\ol{\cal Q}_{yy}^{\beta} \cdot \\
&  \bigg(-\cal Q_{j_ny}G_{i_ny}-\cal Q_{yi_n}\ol G_{j_ny}+S_{\b vj_n}G_{i_ny}\ol G_{j_ny}-\sum_{i,j}S_{\b vi}S_{ij}G_{jj_n}G_{i_nj}|G_{iy}|^2-\sum_{i,j}S_{\b vi}S_{ij}\ol G_{ii_n}\ol G_{j_ni}|G_{jy}|^2\bigg)  \\
&\eqd (A)+(B)+(C)+(D)+(E) \,.
\end{aligned}
\end{equation}
In the remaining proof we look at each term on the above carefully.

{\bf Step 2.} Let us first look at term (C), which is
\begin{equation} \label{eq:A_term}
\begin{aligned}
&\,\Big(\frac{1}{\sqrt{L}}\Big)^n\bb E \sum_{i_1,j_1,...,i_n,j_n}\sigma^{(1)}_{i_1j_1}\cdots \sigma^{(n)}_{i_nj_n}G_{j_1y}\ol G_{i_1y}V_{1,2}\cdots V_{n-1,n}S_{\b vj_n}G_{i_ny}\ol G_{j_ny}\cdot \cal Q_{yy}^{\alpha}\ol{\cal Q}_{yy}^{\beta}\\
=&\,\Big(\frac{1}{\sqrt{L}}\Big)^{n+1}\bb E \sum_{i_1,j_1,...,i_n,j_n}\sigma^{(1)}_{i_1j_1}\cdots \sigma^{(n-1)}_{i_{n-1}j_{n-1}} \sigma^{(n)}_{i_nj_n}G_{j_1y}\ol G_{i_1y}V_{1,2}\cdots V_{n-1,n}G_{i_ny}\ol G_{j_ny}\cdot \cal Q_{yy}^{\alpha}\ol{\cal Q}_{yy}^{\beta}\,,
\end{aligned}
\end{equation}
where in the second line we renamed $ \sqrt{L}S_{\b vj_n}\sigma^{(n)}_{i_nj_n} \in \bb S$ by $\sigma^{(n)}_{i_nj_n}$.
By our definition of $\cal V$ we see that for $m=1,2,...,n-2$, $ V_{m,m+1}$ contains either $\ol G_{i_{m+1} i_m}$ or $ G_{j_{m+1} j_m}$. W.L.O.G. assume $G_{j_2j_1}$ is a factor of $V_{1,2}$, and let $k_1\in \{1,2,3,...,n-2\}$ be the smallest integer such that $\ol G_{i_{k_1+2}i_{k_1+1}}$ is a factor of $V_{k_1+2,k_1+1}$. Thus 
\begin{equation} \label{614}
V_{k_1+1,  k_1}=G_{j_{k_1+1}j_{k_1}}G_{i_{k_1}y}\ol G_{i_{k_1+1}y}\,,
\end{equation} 
and (C) contains
$$
G_{j_{k_1+1}j_{k_1}}\cdots G_{j_2j_1}G_{j_1y}\,.
$$ 
This means that (C) actually contains a chain (in the sense of Remark \ref{rmk:chain_loop}) of order $k_1$. Now let $k_2 \in \{k_1+1,...,n-2\}$ be the smallest integer bigger than $k_1$ such that $G_{j_{k_2+2}j_{k_2+1}}$ is a factor of $V_{k_2+1,k_2+2}$, then
\begin{equation*}
V_{k_2+1,k_2}=\ol G_{i_{k_2+1}i_{k_2}}G_{j_{k_2}y}\ol G_{j_{k_2+1}y}\,,
\end{equation*} 
and \eqref{614} shows that (C) contains
$$
\ol G_{i_{k_2+1}i_{k_2}}\cdots \ol G_{i_{k_1+2}i_{k_1+1}}\ol G_{i_{k_1+1}y}\,.
$$ 
and consequently it contains another chain of order $k_2-k_1$. By continuing this process, we can find in (A) a product of finitely many different chains. Let $\cal T$ be the collection of all such $G$ and $\ol G$ that appear in all these chains, and for $m=1,2,...,n-1$, let $q_m \in \{i_m,j_m\}$ denote the index that appears in one of the chains. Let $\{p_m\}\deq \{i_m,j_m\}/\{q_m\}$ for $m \in \{1,2,...,n-1\}$.

By Lemma \ref{lemma:chain_loop_X_estimates} (i), we see that
\begin{equation*}
\sum_{j_1, \ldots, j_{k_1}} \sigma^{(1)}_{i_1j_1}\cdots  \sigma^{(k_1)}_{i_{k_1}j_{k_1}}G_{j_{k_1+1}j_{k_1}}\cdots G_{j_2j_1}G_{j_1y} \prec \Psi^{2k_1+1}+\Psi^{2k_1}\delta_{j_{k_1+1}y}\,.
\end{equation*}
Assume there are totally $l$ many chains. By applying the above estimate for other chains, we see that
\begin{equation} \label{617}
\sum_{q_1,...,q_{n-1}} \prod_{m=1}^{n-1}  \sigma^{(m)}_{i_mj_m} \prod_{t \in \cal T} t
\prec \big(\Psi^{2k_1+1}+\Psi^{2k_1}\delta_{p_{k_1+1}y})\big(\Psi^{2(k_2-k_1)+1}+\Psi^{2(k_2-k_1)}\delta_{p_{k_2+1}y}\big)\cdots \big(\Psi^{2(k_l-k_{l-1})+1}+\Psi^{2(k_l-k_{l-1})}\delta_{p_{n}y}\big)\,,
\end{equation} 
Note that the LHS of \eqref{617} is contained in (C). Together with \eqref{eq:A_term} we have
\begin{equation} \label{623}
\begin{aligned}
(C) \prec&\, \Big(\frac{1}{\sqrt{L}}\Big)^{n+1} \big(\Psi^{2k_1+1}+\Psi^{2k_1}\delta_{p_{k_1+1}y})\big(\Psi^{2(k_2-k_1)+1}+\Psi^{2(k_2-k_1)}\delta_{p_{k_2+1}y}\big)\cdots \big(\Psi^{2(k_l-k_{l-1})+1}+\Psi^{2(k_l-k_{l-1})}\delta_{p_{n}y}\big)\\
&\,\cdot \bb E \sum_{p_1,...,p_{n-1},i_n,j_n} \Big| \sigma^{(n)}_{i_nj_n} G_{j_1y}\ol G_{i_1y}V_{1,2}\cdots V_{n-1,n}G_{i_ny}\ol G_{j_ny} \Big/\prod_{t \in \cal T}t\Big|\,|\cal Q_{yy}|^{\alpha+\beta}\,.\\
\end{aligned}
\end{equation}
To estimate the above, we need to expand 
\begin{equation} \label{624} \big(\Psi^{2k_1+1}+\Psi^{2k_1}\delta_{p_{k_1+1}y})\big(\Psi^{2(k_2-k_1)+1}+\Psi^{2(k_2-k_1)}\delta_{p_{k_2+1}y}\big)\cdots \big(\Psi^{2(k_l-k_{l-1})+1}+\Psi^{2(k_l-k_{l-1})}\delta_{p_{n}y}\big)
\end{equation}
and consider each term in the result separately. Here we only give estimates of two terms, and other cases follow in a similar fashion.  

Suppose we take the term 
$$
\Psi^{2k_1+1}\Psi^{2(k_2-k_1)+1}\cdots\Psi^{2(k_l-k_{l-1})}
$$from \eqref{624}. Note that $k_l=n-1$, and $k_l+l=|\cal T|$. We have
\begin{equation*}
\begin{aligned}
&\, \Big(\frac{1}{\sqrt{L}}\Big)^{n+1} \Psi^{2k_1+1}\Psi^{2(k_2-k_1)+1}\cdots\Psi^{2(k_l-k_{l-1})}\cdot \bb E \sum_{p_1,...,p_{n-1},i_n,j_n} \Big| \sigma^{(n)}_{i_nj_n} G_{j_1y}\ol G_{i_1y}V_{1,2}\cdots V_{n-1,n}G_{i_ny}\ol G_{j_ny} \Big/\prod_{t \in \cal T}t\Big|\,|\cal Q_{yy}|^{\alpha+\beta}\\
=&\,\Big(\frac{1}{\sqrt{L}}\Big)^{n+1} \Psi^{|\cal T|+n-1}\cdot \bb E \sum_{p_1,...,p_{n-1},i_n,j_n} \Big| \sigma^{(n)}_{i_nj_n} G_{j_1y}\ol G_{i_1y}V_{1,2}\cdots V_{n-1,n}G_{i_ny}\ol G_{j_ny} \Big/\prod_{t \in \cal T}t\Big|\,|\cal Q_{yy}|^{\alpha+\beta}\\
\prec&\,\Big(\frac{1}{\sqrt{L}}\Big)^{n+1} \Psi^{|\cal T|+n-1}\Psi\cdot \bb E \sum_{p_1,...,p_{n-1},q_n} \Big|G_{j_1y}\ol G_{i_1y}V_{1,2}\cdots V_{n-1,n}G_{i_ny}\ol G_{j_ny} \Big/|G_{p_ny}|\Big(\prod_{t \in \cal T}t\Big)\Big|\,|\cal Q_{yy}|^{\alpha+\beta}\\
\prec&\,\Big(\frac{1}{\sqrt{L}}\Big)^{n+1} \Psi^{|\cal T|+n-1}\Psi \cdot L^n \Big(\frac{1}{\sqrt{L\eta}}\Big)^{n+1}\Psi^{3n+1-|\cal T|-1-(n+1)} \bb E |\cal Q_{yy}|^{\alpha+\beta}
=\frac{\Psi^2}{L\eta} \cdot \bigg(\frac{\Psi^3}{\sqrt{\eta}}\bigg)^{n-1}\cdot \bb E |\cal Q_{yy}|^{2p-n-1}\,,
\end{aligned}
\end{equation*}
where in the second step we used $\sum_{p_n}  \sigma^{(n)}_{i_nj_n} |G_{p_{n}y}| \prec \Psi$, and in the third step there are at least $n+1$ many $G,\ol G$ that we can use to apply Ward identity. 

Suppose we take from \eqref{624} the term
\begin{equation*}
\Psi^{2k_1}\delta_{p_{k_1+1}y}\Psi^{2(k_2-k_1)+1}\cdots\Psi^{2(k_l-k_{l-1})}\,,
\end{equation*}
then we have
\begin{equation} \label{620}
\begin{aligned}\Big(\frac{1}{\sqrt{L}}\Big)^{n+1} &\,\Psi^{2k_1}\delta_{p_{k_1+1}y}\Psi^{2(k_2-k_1)+1}\cdots\Psi^{2(k_l-k_{l-1})} \\
& \qquad \cdot \bb E \sum_{p_1,...,p_{n-1},i_n,j_n} \Big| \sigma^{(n)}_{i_nj_n} G_{j_1y}\ol G_{i_1y}V_{1,2}\cdots V_{n-1,n}G_{i_ny}\ol G_{j_ny} \Big/\prod_{t \in \cal T}t\Big|\,|\cal Q_{yy}|^{\alpha+\beta}\\
=&\, \Big(\frac{1}{\sqrt{L}}\Big)^{n+1} \Psi^{|\cal T|+n-2}\cdot \bb E \sum_{p_1,...,p_{n-1},i_n,j_n} \Big|\delta_{p_{k_1+1}y}\sigma^{(n)}_{i_nj_n} G_{j_1y}\ol G_{i_1y}V_{1,2}\cdots V_{n-1,n}G_{i_ny}\ol G_{j_ny} \Big/\prod_{t \in \cal T}t\Big|\,|\cal Q_{yy}|^{\alpha+\beta}\,.
\end{aligned}
\end{equation}
For $k_1+1=n$, we use $\sigma^{(n)}_{i_nj_n} \prec W^{-1}\prec \Psi^2$ and have
\begin{equation*}
\begin{aligned}
\eqref{620} \prec&\, \Big(\frac{1}{\sqrt{L}}\Big)^{n+1} \Psi^{|\cal T|+n}\cdot \bb E \sum_{\substack{p_1,...,p_{n-1},q_n,\\p_n=y}} \Big| G_{j_1y}\ol G_{i_1y}V_{1,2}\cdots V_{n-1,n}G_{i_ny}\ol G_{j_ny} \Big/|G_{p_ny}|\Big(\prod_{t \in \cal T}t\Big)\Big|\,|\cal Q_{yy}|^{\alpha+\beta}\\
\prec&\, \Big(\frac{1}{\sqrt{L}}\Big)^{n+1} \Psi^{|\cal T|+n}\cdot L^{n}\Big(\frac{1}{\sqrt{L\eta}}\Big)^{n+1}\Psi^{3n+1-|\cal T|-1-(n+1)} \bb E |\cal Q_{yy}|^{\alpha+\beta} = \frac{\Psi^2}{L\eta} \cdot \bigg(\frac{\Psi^3}{\sqrt{\eta}}\bigg)^{n-1}\cdot \bb E |\cal Q_{yy}|^{2p-n-1}\,.
\end{aligned}
\end{equation*}
For $k_1+1 \le n-1$, note that there is only one factor $|G_{p_{k_1+1}y}|$ in \eqref{620} that contains the index $p_{k_1+1}$, and we have
\begin{equation*}
\begin{aligned}
\eqref{620} \prec&\, \Big(\frac{1}{\sqrt{L}}\Big)^{n+1} \Psi^{|\cal T|+n-2}\cdot \bb E \sum_{\substack{p_1,...,p_{k_1};p_{k_1+1}=y;\\p_{k_1+2},...,p_{n-1},i_n,j_n} }\Big| \sigma^{(n)}_{i_nj_n}G_{j_1y}\ol G_{i_1y}V_{1,2}\cdots V_{n-1,n}G_{i_ny}\ol G_{j_ny} \Big/|G_{p_{k_1+1}y}|\Big(\prod_{t \in \cal T}t\Big)\Big|\,|\cal Q_{yy}|^{\alpha+\beta}\\
\prec &\,\Big(\frac{1}{\sqrt{L}}\Big)^{n+1} \Psi^{|\cal T|+n-1}\cdot \bb E \sum_{\substack{p_1,...,p_{k_1};p_{k_1+1}=y;\\p_{k_1+2},...,p_{n-1},q_n} }\Big| G_{j_1y}\ol G_{i_1y}V_{1,2}\cdots V_{n-1,n}G_{i_ny}\ol G_{j_ny} \Big/|G_{p_{k_1+1}y}G_{p_ny}|\Big(\prod_{t \in \cal T}t\Big)\Big|\,|\cal Q_{yy}|^{\alpha+\beta}\\
\prec&\, \Big(\frac{1}{\sqrt{L}}\Big)^{n+1} \Psi^{|\cal T|+n-1}\cdot L^{n-1}\Big(\frac{1}{\sqrt{L\eta}}\Big)^{n}\Psi^{3n+1-|\cal T|-2-n} \bb E |\cal Q_{yy}|^{\alpha+\beta} \prec \frac{\Psi^2}{L\eta} \cdot \bigg(\frac{\Psi^3}{\sqrt{\eta}}\bigg)^{n-1}\cdot \bb E |\cal Q_{yy}|^{2p-n-1}\,,
\end{aligned}
\end{equation*}
where in the last step we used the estimate $\frac{\sqrt{L\eta}}{L\Psi} \prec 1$.

One can take other bounds in \eqref{624} and show that we have the same bound. Thus we obtained from \eqref{623} that
\begin{equation} \label{6.24}
(A) \prec \frac{\Psi^2}{L\eta} \cdot \bigg(\frac{\Psi^3}{\sqrt{\eta}}\bigg)^{n-1}\cdot \bb E |\cal Q_{yy}|^{2p-n-1} \prec \bigg(\frac{\Psi}{\sqrt{L\eta}}+\frac{\Psi^3}{\sqrt{\eta}}\bigg)^{n+1} \bb E |\cal Q_{yy}|^{2p-n-1}\,.
\end{equation}

{\bf Step 3.} Now let's look at the first and fourth term on the RHS of \eqref{eq:Q_term2}, which is (A) and (D). It is important to consider these two contributions together because there is a crucial cancellation between them. By writing 
$$
\cal Q_{j_ny}=\sum_{i_{n+1},j_{n+1}} S_{\b vi}(H_{i_{n+1}j_{n+1}}G_{j_{n+1}j_n}\ol{G}_{i_{n+1}y}+S_{i_{n+1}j_{n+1}}G_{j_{n+1}j_{n+1}}G_{i_{n+1}x}\ol{G}_{i_{n+1}y}+S_{i_{n+1}j_{n+1}}\ol G_{i_{n+1}i_{n+1}}G_{j_{n+1}j_n}\ol{G}_{j_{n+1}y})
$$
and using cumulant expansion on $H_{i_{n+1}j_{n+1}}$ we see that
\begin{equation} \label{eq:rhs_Q}
\begin{aligned}
&\,(A)+(D)=-\Big(\frac{1}{\sqrt{L}}\Big)^n\bb E  \sum_{i_1,j_1,...,i_n,j_n} \sigma^{(1)}_{i_1j_1}\cdots  \sigma^{(n)}_{i_nj_n}G_{j_1y}\ol G_{i_1y}V_{1,2}\cdots V_{n-1,n}\cal Q_{j_ny}G_{i_ny}\cdot \cal Q_{yy}^{\alpha}\ol{\cal Q}_{yy}^{\beta}+(D)\\
=&\,-\Big(\frac{1}{\sqrt{L}}\Big)^n\bb E \sum_{i_1,j_1,...,i_{n+1},j_{n+1}}  \sigma^{(1)}_{i_1j_1}\cdots  \sigma^{(n)}_{i_nj_n} \partial_{j_{n+1}i_{n+1}} \big(G_{j_1y}\ol G_{i_1y}V_{1,2}\cdots V_{n-1,n}S_{\b vi_{n+1}}S_{i_{n+1}j_{n+1}}G_{j_{n+1}j_n}\ol G_{i_{n+1}y}G_{i_ny}\cdot \cal Q_{yy}^{\alpha}\ol{\cal Q}_{yy}^{\beta}\big)\\
&\,-\Big(\frac{1}{\sqrt{L}}\Big)^n\bb E \sum_{i_1,j_1,...,i_n,j_n} \sigma^{(1)}_{i_1j_1}\cdots  \sigma^{(n)}_{i_nj_n}G_{j_1y}\ol G_{i_1y}V_{1,2}\cdots V_{n-1,n} \cdot \\
& \Big(\sum_{i_{n+1},j_{n+1}}S_{\b vi_{n+1}}S_{i_{n+1}j_{n+1}}(G_{j_{n+1}j_{n+1}}G_{i_{n+1}j_n}\ol G_{i_{n+1}y} + G_{j_{n+1}j_n}\ol G_{i_{n+1}i_{n+1}}\ol G_{j_{n+1}y})\Big)G_{i_ny}\cal Q_{yy}^{\alpha}\ol{\cal Q}_{yy}^{\beta}+(D)\,.
\end{aligned}
\end{equation}
We see that when the differential $\partial_{j_{n+1}i_{n+1}}$ is applied to $G_{j_{n+1}j_n}$, $\ol G_{i_{n+1}y}$, and $G_{i_ny}$, the result will cancel the second, third and last term on the RHS of \eqref{eq:rhs_Q} respectively. Thus by setting $\sigma^{(n+1)}_{i_{n+1}j_{n+1} }\deq \sqrt{L}S_{\b vi_{n+1}}S_{i_{n+1}j_{n+1}}$ we have
\begin{equation} \label{eq:rhs_Q2}
\begin{aligned}
&\,(A)+(D)\\
=&\,-\Big(\frac{1}{\sqrt{L}}\Big)^{n+1}\bb E \sum_{i_1,j_1,...,i_{n+1},j_{n+1}}  \sigma^{(1)}_{i_1j_1}\cdots  \sigma^{(n+1)}_{i_{n+1}j_{n+1}} \big[\partial_{j_{n+1}i_{n+1}} \big(G_{j_1y}\ol G_{i_1y}V_{1,2}\cdots V_{n-1,n}\cdot \cal Q_{yy}^{\alpha}\ol{\cal Q}_{yy}^{\beta}\big)\big]G_{j_{n+1}j_n}\ol G_{i_{n+1}y}G_{i_ny}\,.\\
\end{aligned}
\end{equation}
Also, note that $V_{n,n+1}\deq G_{j_{n+1}j_n}\ol G_{i_{n+1}y}G_{i_ny} \in \cal V_{n,n+1}$, thus when the differential is applied to $\cal Q_{yy}^{\alpha}\ol{\cal Q}_{yy}^{\beta}$, we return to the same form as \eqref{eq:Q_term}. The term left to be estimated is
\begin{equation} \label{eq:left_term}
-\Big(\frac{1}{\sqrt{L}}\Big)^{n+1}\bb E \sum_{i_1,j_1,...,i_{n+1},j_{n+1}}  \sigma^{(1)}_{i_1j_1}\cdots  \sigma^{(n+1)}_{i_{n+1}j_{n+1}} \big[\partial_{j_{n+1}i_{n+1}} \big(G_{j_1y}\ol G_{i_1y}V_{1,2}\cdots V_{n-1,n}\big)\big]G_{j_{n+1}j_n}\ol G_{i_{n+1}y}G_{i_ny}\cdot \cal Q_{yy}^{\alpha}\ol{\cal Q}_{yy}^{\beta}\,.
\end{equation}
As in the estimation of (A), we can first apply Lemma \ref{lemma:chain_loop_X_estimates} to sum over $q_1,...,q_{n-1},j_{n},q_{n+1}$. Depending on whether $\partial_{j_{n+1}i_{n+1}}$ is applied to a $G$ or $\ol{G}$, the index $q_{n+1}$ will be equal to $j_{n+1}$ or $i_{n+1}$ respectively. Also, in this case we will have one loop (in the sense of Remark \ref{rmk:chain_loop}) if the differential is applied to a $G_{j_k x}$ ($x\in \{y, j_{k-1}\}$) and we have $G_{j_{n+1}j_n}\cdots G_{j_{k+1}j_k}$ in \eqref{eq:left_term}. Other summations in \eqref{eq:left_term} will still give rise to chains. Thus, with the same procedure of Step 2 we can use Lemma \ref{lemma:chain_loop_X_estimates}(i)(ii) to show that
\begin{equation} \label{eq:left_term2}
\eqref{eq:left_term} \prec \bigg(\frac{\Psi}{L\eta}+\frac{\Psi^3}{\sqrt{\eta}}\bigg)^{n+1}\cdot \bb E |\cal Q_{yy}|^{2p-n-1}\,.
\end{equation}
Thus by \eqref{eq:rhs_Q2}--\eqref{eq:left_term2} we have
\begin{multline} \label{eq:semifinal_bound}
(A)+(D)=\Big(\frac{1}{\sqrt{L}}\Big)^{n+1}\bb E \sum_{i_1,j_1,...,i_{n+1},j_{n+1}} \sigma^{(1)}_{i_1j_1}\cdots  \sigma^{(n+1)}_{i_{n+1}j_{n+1}}G_{j_1y}\ol G_{i_1y}V_{1,2}\cdots V_{n,n+1} (\partial_{j_{n+1}i_{n+1}} \cal Q_{yy}^{\alpha}\ol{\cal Q}_{yy}^{\beta})\\+O_{\prec}\bigg(\frac{\Psi}{L\eta}+\frac{\Psi^3}{\sqrt{\eta}}\bigg)^{n+1}\cdot \bb E |\cal Q_{yy}|^{2p-n-1}\,.
\end{multline}
Similarly, 
\begin{multline} \label{eq:semifinal_bound2}
(B)+(E)=\Big(\frac{1}{\sqrt{L}}\Big)^{n+1}\bb E \sum_{i_1,j_1,...,i_{n+1},j_{n+1}}\sigma^{(1)}_{i_1j_1}\cdots \sigma^{(n+1)}_{i_{n+1}j_{n+1}}G_{j_1y}\ol G_{i_1y}V_{1,2}\cdots V_{n-1,n} V'_{n,n+1} (\partial_{j_{n+1}i_{n+1}} \cal Q_{yy}^{\alpha}\ol{\cal Q}_{yy}^{\beta})\\+O_{\prec}\bigg(\frac{\Psi}{L\eta}+\frac{\Psi^3}{\sqrt{\eta}}\bigg)^{n+1}\cdot \bb E |\cal Q_{yy}|^{2p-n-1}\,,
\end{multline}
where  $V'_{n,n+1}\deq \ol{G}_{i_{n+1}i_n}\ol{G}_{j_ny}G_{j_{n+1}y} \in  \cal V_{n,n+1}$.

{\bf Step 4}. Plugging \eqref{6.24}, \eqref{eq:semifinal_bound} and \eqref{eq:semifinal_bound2} into \eqref{eq:Q_term2} gives
\begin{multline} \label{eq:semifinal_bound3}
\Big(\frac{1}{\sqrt{L}}\Big)^n\bb E \sum_{i_1,j_1,...,i_n,j_n} \sigma^{(1)}_{i_1j_1}\cdots  \sigma^{(n)}_{i_nj_n}G_{j_1y}\ol G_{i_1y}V_{1,2}\cdots V_{n-1,n}(\partial_{j_ni_n} \cal Q_{yy})\cal Q_{yy}^{\alpha}\ol{\cal Q}_{yy}^{\beta}\\
=\Big(\frac{1}{\sqrt{L}}\Big)^{n+1}\bb E \sum_{i_1,j_1,...,i_{n+1},j_{n+1}} \sigma^{(1)}_{i_1j_1}\cdots  \sigma^{(n+1)}_{i_{n+1}j_{n+1}}G_{j_1y}\ol G_{i_1y}V_{1,2}\cdots V_{n-1,n} (V_{n,n+1}+V'_{n,n+1}) (\partial_{j_{n+1}i_{n+1}} \cal Q_{yy}^{\alpha}\ol{\cal Q}_{yy}^{\beta})\\+O_{\prec}\bigg(\frac{\Psi}{L\eta}+\frac{\Psi^3}{\sqrt{\eta}}\bigg)^{n+1}\cdot \bb E |\cal Q_{yy}|^{2p-n-1}\,.
\end{multline}
Similar result can also be obtained for
\begin{equation*} 
\Big(\frac{1}{\sqrt{L}}\Big)^n\bb E \sum_{i_1,j_1,...,i_n,j_n} \sigma^{(1)}_{i_1j_1}\cdots  \sigma^{(n)}_{i_nj_n}G_{j_1y}\ol G_{i_1y}V_{1,2}\cdots V_{n-1,n}(\partial_{j_ni_n} \ol{\cal Q}_{yy})\cal Q_{yy}^{\alpha}\ol{\cal Q}_{yy}^{\beta}\,,
\end{equation*}
together with \eqref{eq:Q^2p} we have
\begin{align*}
\bb E |\cal Q_{yy}|^{2p}= & \sum_{n=1}^{2p-1} O_{\prec}\bigg(\frac{\Psi}{L\eta}+\frac{\Psi^3}{\sqrt{\eta}}\bigg)^{n+1}\cdot \bb E |\cal Q_{yy}|^{2p-n-1}\,,
\end{align*}
which gives the desired result by Lemma \ref{lemma:moment_to_domination}.

The proof of \eqref{eq:P_xy} proceeds as the one for \eqref{eq:P_vy} but the $ S_{\b v i_m}$'s ($m = 1, \ldots, 2p-1 $) are replaced by the $ S_{xi_m}$'s which are bounded by $ W^{-1}$. Finally, \eqref{eq:R_vy} and \eqref{eq:R_xy} are easily obtained by using Lemma \ref{lemma:chain_loop_X_estimates} (iii) and the Ward identity \eqref{eq:ward_identiy}.

\section{Proof of Lemma \ref{lemma:chain_loop_X_estimates}}\label{lemma:chain_loop_X}
In order to prove Lemma \ref{lemma:chain_loop_X_estimates} we will need some auxiliary technical lemmata. The first one concerns a trick to write self-consistent equations for the kind of expectations we are going to use.

\begin{lemma}\label{lemma_self-consistent_equation_trick}
Consider the expectation
\begin{align*}
D_{abc} = \E \sum_i \sigma_{ai} G_{ib}G_{ci} \mathfrak p (G, \overline G)
\end{align*}
where $ \sigma \in \bb  {S}$ and $ \mathfrak p (G, \overline G)$ is a polynomial of $G_{xy}$ and $ \overline G_{x'y'}$ with $x, y, x', y' \not\equiv  i $. Then
\begin{align}\label{eq:self_consist_trick}
D_{abc} = &  \, \E \xi_{ab} G_{cb} \mathfrak p (G, \overline G) - \E \sum_{i,j} \xi_{ai} S_{ij} G_{jb}G_{ci} \partial_{ji}\mathfrak p (G, \overline G)  \\\nonumber
&+  \E \sum_{i,j} \xi_{ai} S_{ij} ((G_{jj}-\mathfrak m)G_{ib}G_{ci} + G_{jb}G_{cj}(G_{ii}-\mathfrak m)) \mathfrak p (G, \overline G)\,
\end{align}
where $ \xi \in \bb S$.
\end{lemma}

\begin{proof}
Applying \eqref{eq:def_G} and the cumulant expansion formula \eqref{eq:cum_exp_gaussian}, we get:
\begin{align}\label{eq:test_expansion_0}
z D_{abc}= \E \sum_i \sigma_{ai} (zG_{ib}) G_{ci} \mathfrak p (G, \overline G) = - \E \sum_i \sigma_{ai} \delta_{ib} G_{ci} \mathfrak p (G, \overline G) + \E \sum_{i,j} \sigma_{ai} S_{ij} \partial_{ji} (G_{jb}G_{ci} \mathfrak p (G, \overline G)) 
 \end{align}
By performing the derivatives and writing the diagonal entries of $G$ as $(G_{ii} - \fra m) + \fra m$ one has
 \begin{align}\nonumber
& z D_{abc} = - \E \sigma_{ab} G_{cb} \mathfrak p (G, \overline G) - \E \sum_{i,j} \sigma_{ai} S_{ij} (G_{jj}G_{ib}G_{ci} + G_{jb}G_{cj}G_{ii}) \mathfrak p (G, \overline G) + \E \sum_{i,j} \sigma_{ai} S_{ij} G_{jb}G_{ci} \partial_{ji}\mathfrak p (G, \overline G) \\\nonumber
& =  - \E \sigma_{ab} G_{cb} \mathfrak p (G, \overline G) - \mathfrak m D_{abc} - \mathfrak m \sum_{i} \sigma_{ai} \E \sum_j S_{ij} G_{jb} G_{cj}\mathfrak p (G, \overline G) \\
&\quad - \E \sum_{i,j} \sigma_{ai} S_{ij} ((G_{jj}-\mathfrak m)G_{ib}G_{ci} + G_{jb}G_{cj}(G_{ii}-\mathfrak m)) \mathfrak p (G, \overline G) + \E \sum_{i,j} \sigma_{ai} S_{ij} G_{jb}G_{ci} \partial_{ji}\mathfrak p (G, \overline G) .\label{eq:test_expansion}
\end{align}
Let us consider the summation of the last term of the expansion in \eqref{eq:test_expansion}
\begin{align*}
\tilde D_{abc} := \E \sum_i S_{ai} G_{ib} G_{ci}\mathfrak p (G, \overline G).
\end{align*}
Note that the only difference between $ D_{abc}$ and $ \tilde D_{abc}$ is that $ \sigma_{ai}$ in $  D_{abc}$ is replaced by $ S_{ai}$ in $ \tilde D_{abc}$. Expanding $ \tilde D_{abc}$ in the same way we have done for $ D_{abc}$, we get
\begin{align}\nonumber
& \sum_{i}((z + m) \delta_{ai} + \mathfrak m S_{ai} )\tilde D_{ibc}= - \E S_{ab} G_{cb} \mathfrak p (G, \overline G) + \E \sum_{i,j} S_{ai} S_{ij} G_{jb}G_{ci} \partial_{ji}\mathfrak p (G, \overline G) \\ 
&- \E \sum_{i,j} S_{ai} S_{ij} ((G_{jj}-\mathfrak m)G_{ib}G_{ci} + G_{jb}G_{cj}(G_{ii}-\mathfrak m)) \mathfrak p (G, \overline G).
\label{eq:test_exp_tilde}
\end{align}
Let us analyse the operator on the left hand side of \eqref{eq:test_exp_tilde}: thanks to \eqref{eq:m_equation} we get
\begin{align}\label{eq:def_L}
\cal L := ((z + \mathfrak{m}) I + \mathfrak m S)^{-1} = - \mathfrak m (I - \mathfrak m^2 S)^{-1}.
\end{align}
Since we are interested in the bulk spectrum of the band matrices, from Proposition B.2 in \cite{EKY13} we see that there is a positive constant $C$ such that
\begin{align}\label{eq:norm_bound_to_1}
\rho := \norm{\cal L}_{\infty \to \infty} \leq C \log L \prec 1.
\end{align} 
From the translational invariance of $\cal L$ and \eqref{eq:norm_bound_to_1} one can easily see that $ \cal L \xi \in \bb S$ for any $ \xi \in \bb S$. 
%In fact, using that $ \cal L$ is translation invariant and \eqref{eq:norm_bound_to_1} we have
%\begin{align*}
%&(\cal L \xi)_{xy} \leq (\sup_i| \xi_{iy})|\sum_{i} |\cal L_{xi}| \prec \rho W^{-1} \prec W^{-1}, \\\nonumber
%&\sum_x |(\cal L \xi)_{xy}| \leq \sum_{x,i} |\cal L_{xi}||\xi_{iy}|= \sum_{x'} |\cal L_{x'0}|\sum_{i}|\xi_{iy}| \prec \rho \prec 1, \\\nonumber
%&\sum_y |(\cal L \xi)_{xy}| \leq \sum_{y,i} |\cal L_{xi}||\xi_{iy}|= \sum_{i} |\cal L_{xi}|\sum_{y}|\xi_{iy}| \prec \rho \prec 1.
%\end{align*}
Thus, \eqref{eq:norm_bound_to_1} and \eqref{eq:test_exp_tilde} imply
\begin{align}\label{eq:bound_tildeA}
\tilde D_{abc} =& -\E \tau_{ab} G_{cb} \mathfrak p (G, \overline G) + \E \sum_{i,j} \tau_{ai} S_{ij} G_{jb}G_{ci} \partial_{ji}\mathfrak p (G, \overline G) \\\nonumber
&- \E \sum_{i,j} \tau_{ai} S_{ij} ((G_{jj}-\mathfrak m)G_{ib}G_{ci} + G_{jb}G_{cj}(G_{ii}-\mathfrak m)) \mathfrak p (G, \overline G)
\end{align}
where $ \tau = \cal L S \in \bb S$.
Coming back to \eqref{eq:test_expansion} and using again \eqref{eq:def_G}, we get
\begin{align*}
& D_{abc} =  \mathfrak m \E \sigma_{ab} G_{cb} \mathfrak p (G, \overline G) - \mathfrak m \E \sum_{i,j} \sigma_{ai} S_{ij} G_{jb}G_{ci} \partial_{ji}\mathfrak p (G, \overline G) \\\nonumber
&+ \mathfrak m \E \sum_{i,j} \sigma_{ai} S_{ij} ((G_{jj}-\mathfrak m)G_{ib}G_{ci} + G_{jb}G_{cj}(G_{ii}-\mathfrak m)) \mathfrak p (G, \overline G) + \mathfrak m^2 \sum_{i} \sigma_{ai} \tilde D_{ibc}.
\end{align*}
In conclusion, using \eqref{eq:bound_tildeA} and the fact that $ \bb S$ is close with respect to matrix addition and multiplication, we get
\begin{align*}
D_{abc} = &  \, \E \xi_{ab} G_{cb} \mathfrak p (G, \overline G) -\E \sum_{i,j} \xi_{ai} S_{ij} G_{jb}G_{ci} \partial_{ji}\mathfrak p (G, \overline G)  \\\nonumber
&+  \E \sum_{i,j} \xi_{ai} S_{ij} ((G_{jj}-\mathfrak m)G_{ib}G_{ci} + G_{jb}G_{cj}(G_{ii}-\mathfrak m)) \mathfrak p (G, \overline G)\,.
\end{align*}
where $ \xi = \fra m \,(\sigma - \fra m \sigma \tau) \in \bb S$.
\end{proof}
It is convenient to define the following transformation on the matrices belonging to $ \bb S$: let $ \sigma \in \bb S$, then we set 
$$
\breve \sigma := \fra m (\sigma - \fra m\sigma \cal L S) \in \bb S
$$
where $\cal L$ is the matrix defined in \eqref{eq:def_L}. With this notation we can write $ \xi$ in Lemma \ref{lemma_self-consistent_equation_trick} as $ \xi = \breve \sigma$.

By using an argument very similar to the one exploited in Lemma \ref{lemma_self-consistent_equation_trick} we can also show that 
\begin{align}\label{eq:bound_calG}
\cal G_a := \sum_i \sigma_{ai} \E(G_{ii}-\mathfrak{m})  \prec \Psi^2
\end{align}
under the assumption that $\Lambda \prec \Psi$. In fact, by the cumulant expansion formula \eqref{eq:cum_exp_gaussian} we have
\begin{align}\label{eq:calG}
z \cal G_a = - \E \sum_{i,j} \sigma_{ai}S_{ij}(G_{jj}-\fra m)(G_{ii}-\fra m) - \fra m \cal G_a - \fra m \sum_{i} \sigma_{ai} \tilde{\cal G}_i
\end{align}
where $\tilde{\cal G}_a := \sum_i S_{ai} \E(G_{ii}-\mathfrak{m}) $. Expanding $\tilde{\cal G}_a$ again via \eqref{eq:cum_exp_gaussian} and using the properties of the operator $\cal L$ as we did for $\tilde D_{abc}$ in the proof of Lemma \ref{lemma_self-consistent_equation_trick}, we get
\begin{align*}
\tilde{ \cal G}_a = - \E \sum_{i,j} (\cal L S)_{ai}S_{ij}(G_{jj}-\fra m)(G_{ii}-\fra m) .
\end{align*} 
Using \eqref{eq:calG}, this implies $ \cal G_a = \, \E \sum_{i,j} \breve\sigma_{ai}S_{ij}(G_{jj}-\fra m)(G_{ii}-\fra m) \prec \Psi^2$.
In the same way one gets also the following estimates
\begin{align}\label{eq:corollary1}
& \E(G_{jj}-\fra m) \prec \Psi^2, \ \ \  \E G_{jk} \prec \Psi^2 \ \mbox{for } k \neq j, \ \ \  \sum_j \sigma_{ij} \E G_{jk} \prec \Psi^2.
\end{align}

We will prove \eqref{eq:chain_bound} and \eqref{eq:loop_bound} by induction over $n$. Therefore, first we should check that those two formulae are true for $n = 1$ and $n =2$ respectively. This is shown in the following lemma.

\begin{lemma}\label{lemma:first_induction_step}
Under the assumption of Lemma \ref{lemma:chain_loop_X_estimates}, for some $\xi \in \bb S$ with nonnegative entries one has
\begin{align}
\label{eq:bound_Y}
Y^{(1)}_{ab;u} \prec \Psi^3 + \delta_{ab}\Psi^2, \\
\label{eq:bound_Z}
Z^{(2)}_{ab} \prec \Psi^4 + \xi_{ab}\,.
\end{align}
\end{lemma}

\begin{proof}
We slightly simplify the notation by setting $Y_{uab} := Y^{(1)}_{ab;u} = \sum_i \sigma_{ui}G_{ai}G_{ib}$ with $ \sigma \in \bb S$. We fix an arbitrary integer $p \geq 1$ and we apply Lemma \ref{lemma_self-consistent_equation_trick} to $\E |Y_{uab}|^{2p} = \E \sum_{i} \sigma_{ui} G_{ai}G_{ib}Y_{sab}^{p-1} \overline{Y}_{sab}^p $, so we get
\begin{align}\label{eq:expansion_callY_pq}
\nonumber
& \E |Y_{uab}|^{2p} \prec \bigg \vert \E \breve\sigma_{ua} G_{ab} (Y_{uab})^{p-1} (\overline{Y}_{uab})^p \bigg\vert + \bigg\vert \E \sum_{i,j} \breve\sigma_{ui} S_{ij} G_{ai} G_{jb} \partial_{ji} (Y_{uab}^{p-1} \overline{Y}_{uab}^p) \bigg\vert \\\nonumber
& + \bigg\vert \E \sum_{i,j} \breve\sigma_{ui} S_{ij} ((G_{jj}-\mathfrak{m}) G_{ai} G_{ib} + G_{aj} G_{jb} (G_{ii}-\mathfrak{m}) ) Y_{uab}^{p-1} \overline{Y}_{uab}^p \bigg \vert \\
& \prec (\Psi^3 + \delta_{ab}\Psi^2)
\E |Y_{uab}|^{2p-1} 
+ \bigg \vert \E \sum_{i,j} \breve\sigma_{ui} S_{ij} G_{aj} G_{ib} \partial_{ji} (Y_{uab}^{p-1} \overline{Y}_{uab}^p) \bigg \vert
\end{align}
where we used the hypothesis $\Lambda \prec \Psi$. 
% the trivial estimate 
%$$
%\sum_{ij} \sigma_{ui} S_{ij} (G_{jj}-\mathfrak{m}) G_{ai} G_{jb} \prec \Psi^3.
%$$
Let us focus on the last term of \eqref{eq:expansion_callY_pq}: we note that
\begin{align}\label{eq:self_similar_derivative_Y}
\partial_{ji} Y_{uab}= - G_{ia} Y_{ujb} - G_{bj}Y_{uai}, \ \ \ \
\partial_{ji} \overline Y_{uab} = - \overline G_{ja} \overline Y_{uib} - \overline G_{bi} \overline Y_{uaj}.
\end{align}
For $a \neq b$ we set the a prior bound
\begin{align}\label{eq:apriori_Y_ab}
Y_{uab} \prec \lambda
\end{align}
for $\lambda \in [\Psi^3, L^C]$ for some $C >0$, 
so that for $a \neq b$, recalling the trivial bound $Y_{uab} \prec \Psi^2$, using \eqref{eq:apriori_Y_ab}, \eqref{eq:self_similar_derivative_Y} and \eqref{eq:expansion_callY_pq}, we have
\begin{align}\label{eq:high_moment_callY}
\E|Y_{uab}|^{2p} \prec \Psi^3 \E |Y_{uab}|^{2p-1} + \lambda \Psi^3 \E |Y_{uab}|^{2p-2} + \Psi^6 \E |Y_{uab}|^{2p-2} \prec \lambda \Psi^3 \E |Y_{uab}|^{2p-2}. 
\end{align}
From Lemma \ref{lemma:moment_to_domination} with $ \vartheta = \lambda$, $ \varphi = \Psi^3$ and $ q=1/2$ we deduce that $ Y_{uab} = Y_{ab;u}^{(1)} \prec \Psi^3$ for $a \neq b$, and therefore in general we obtain $$
Y_{ab;u}^{(1)} \prec \Psi^3 + \delta_{ab} \Psi^2.
$$
Let us now prove \eqref{eq:bound_Z} with the same strategy used for \eqref{eq:bound_Y}. We recall that 
\begin{align*}
Z_{ab}^{(2)}:= \sum_{ij} \sigma_{ai}  \tau_{bj} G_{ij}G_{ji}
\end{align*}
with $ \sigma, \tau \in \bb S$. In the following we will drop the superscript $(2)$. Again from Lemma \ref{lemma_self-consistent_equation_trick} for any fixed integer $p$ we have
\begin{align}\label{eq:expansion_Z_pq}\nonumber
  \E |Z_{ab}|^{2p} \prec & \, \, \bigg\vert \E \sum_{i} \breve\sigma_{ai} \tau_{bi} G_{ii} Z_{ab}^{p-1} \overline Z_{ab}^{p} \bigg\vert 
 + \bigg\vert \E \sum_{i,j,k}\breve\sigma_{ai} \tau_{bj}S_{ik}[(G_{kk}-\mathfrak{m})G_{ji}G_{ij}+ (G_{ii}-\mathfrak{m})G_{kj}G_{jk}] Z_{ab}^{p-1} \overline Z_{ab}^{p} \bigg\vert
\\\nonumber
& + \bigg\vert \E \sum_{i,j,k} \breve\sigma_{ai} \tau_{bj}S_{ik} G_{kj} G_{ji} \partial_{ki} (Z_{ab}^{p-1} \overline Z_{ab}^{p}) \bigg\vert \\
 \prec & \,\, (\xi_{ab} +| B_{ab}| + \Psi^4)\E |Z_{cd}|^{2p-1} + \bigg\vert \E \sum_{i,j,k} \breve\sigma_{ai}\tau_{bj}S_{ik} G_{kj} G_{ji} \partial_{ki} (Z_{ab}^{p-1} \overline Z_{ab}^{p}) \bigg\vert
\end{align}
where $ \xi \in \bb S$ and in the last line we used \eqref{eq:X_bound}. Moreover, 
\begin{align}
\label{eq:B_def}
B_{ab} &:= \sum_{i,j,k}\breve \sigma_{ai} \tau_{bj}S_{ik} G_{ii}^\flat G_{kj}G_{jk}
\end{align}
where we used the notation $g^\flat := g - \E g $. 
We now have to estimate $B_{ab}$: note that by using \eqref{eq:bound_Y} we get $B_{ab} \prec \Psi W^{-1}$, but this is not enough for our purposes. Thus, we are going to expand $B_{ab}$: as before Lemma \ref{lemma_self-consistent_equation_trick} implies
\begin{align}\label{eq:expansion_B_pq}
 \E |B_{ab}|^{2p} & \prec  \bigg\vert \E \sum_{i,j}\breve \sigma_{ai} \breve \tau_{bj}S_{ij} G_{ii}^\flat G_{jj} B_{ab}^{p-1} \overline B_{ab}^{p} \bigg\vert + \bigg\vert \E \sum_{i,j,k,l}\breve \sigma_{ai} \breve \tau_{bj}S_{ik} S_{jl}G_{kj} G_{lk} \partial_{lj}(G_{ii}^\flat B_{ab}^{p-1} \overline B_{ab}^{p})\bigg\vert \\\nonumber
& \prec \Psi^4\E|B_{ab}|^{2p-1} + \Psi^8 \E|B_{ab}|^{2p-2}
\end{align}
where in the second line we used \eqref{eq:X_bound}, \eqref{eq:bound_Y} and the fact that trivially one gets $\partial_{lj} B_{ab} \prec \Psi^4$. Thus, we obtained that
\begin{align}\label{eq:B_basic_bound}
B_{ab} \prec \Psi^4
\end{align}
and 
\begin{align}\label{eq:Z_temporary_bound}
 \E |Z_{ab}|^{2p} \prec (\xi_{ab} + \Psi^4)\E |Z_{ab}|^{2p-1} + \bigg\vert \E \sum_{i,j,k} \breve\sigma_{ai} \tau_{bj}S_{ik} G_{kj} G_{ji} \partial_{ki} (Z_{ab}^{p-1} \overline Z_{ab}^{p}) \bigg\vert \,.
\end{align}
Let us focus on the last term in \eqref{eq:Z_temporary_bound}: we need to compute the derivatives of $ Z$:
\begin{align}\label{eq:derivative_Z}
\partial_{ki} Z_{cd} &= - \sum_{l h} \sigma_{c l} \tau_{d h}(G_{lk}G_{ih}G_{hl} + G_{hk}G_{il}G_{lh})\\
\label{eq:derivative_Z_bar}
\partial_{ki}\overline Z_{cd} &= - \sum_{l h} \ol \sigma_{c l} \ol \tau_{d h}(\overline G_{hl}\overline G_{li} \overline G_{kh}  + \overline G_{lh} \overline G_{hi} \overline G_{kl}).
\end{align}
We consider the first term in \eqref{eq:derivative_Z} (the second one is treated similarly as well as the two terms in \eqref{eq:derivative_Z_bar}): we define 
\begin{align*}
A_{abcd}:= \sum_{i,j} \sigma_{a i} \tau_{b j}G_{ic}G_{dj}G_{ji}\,.
\end{align*}
To bound $ A_{abcd}$ we perform again a cumulant expansions. For $c \neq d$ set the prior bound
\begin{align}\label{eq:apriori_A}
A_{abcd}  \prec \alpha
\end{align}
with $\alpha \in [\Psi^5, L^C]$. For any integer $p \geq 1$ Lemma \ref{lemma_self-consistent_equation_trick} 
%with $ \mathfrak p(G, \ol G) = \sum_{j} \tau_{b j}G_{dj} A_{a'bcd}^{p-1} \overline A_{a'bcd}^{p}$ 
yields 
\begin{align}\label{eq:expansion_A}
& \E |A_{abcd}|^{2p} \prec \bigg\vert \E \sum_{j} \breve\sigma_{a c} \tau_{b j}G_{cc}G_{dj}G_{jc}A_{abcd}^{p-1} \overline A_{abcd}^{p} \bigg\vert \\\nonumber
& + \bigg\vert \E  \sum_{i,j,k} \breve\sigma_{a i} \tau_{b j}S_{ki}[(G_{kk}-\mathfrak{m})G_{ic} G_{dj}G_{jik} + (G_{ii}-\mathfrak{m})G_{kc} G_{dj}G_{jk}] A_{abcd}^{p-1} \overline A_{abcd}^{p} \bigg\vert \\\nonumber
& + \bigg\vert \E  \sum_{i,j,k} \breve\sigma_{a i} \tau_{b j}S_{ki} G_{kc} G_{dk}G_{ij}G_{ji} A_{abcd}^{p-1} \overline A_{abcd}^{p} \bigg\vert  + \bigg\vert \E  \sum_{i,j,k} \breve\sigma_{a i} \tau_{b j}S_{ki}G_{kc}G_{dj}G_{ji} \partial_{ki}(A_{abcd}^{p-1} \overline A_{abcd}^{p}) \bigg\vert\,.
\end{align}
By using \eqref{eq:bound_Y} we get for $c \neq d$
\begin{align}\label{eq:call_A1}
& \E |A_{abcd}|^{2p} \prec \Psi^5 \E |A_{abcd}|^{2p-1} + \bigg\vert \E  \sum_{i,j,k} \breve\sigma_{a i} \tau_{b j}S_{ki}G_{kc}G_{dj}G_{ji} \partial_{ki} (A_{abcd}^{p-1} \overline A_{abcd})^{p} \bigg\vert. 
\end{align}
To bound the second term on the right hand side of \eqref{eq:call_A1} we observe that
\begin{align}\label{eq:derivative_A}
\nonumber
& \partial_{ki} A_{abcd}  = - G_{ic}A_{abkd} - G_{dk}A_{abci}-\sum_{lh} \sigma_{al} \tau_{bh}G_{ld}G_{ch}G_{hk}G_{il} \\\nonumber
& \prec |G_{ic}|(\alpha + \delta_{kd}|A_{abdd}|) + |G_{dk}| (\alpha + \delta_{ci}|A_{abcc}|) + (\Psi^3 + \delta_{dl}W^{-1})(\Psi^3 + \delta_{ck}W^{-1})  \\
& \prec |G_{ic}|(\alpha + \delta_{kd}\Psi^4) + |G_{dk}| (\alpha + \delta_{ci}\Psi^4) + (\Psi^3 + \delta_{di}W^{-1})(\Psi^3 + \delta_{ck}W^{-1})
\end{align}
where we used \eqref{eq:bound_Y}, the trivial estimate $A_{abcd} \prec \Psi^4$ (where we exploited again \eqref{eq:bound_Y}) and the prior estimate \eqref{eq:apriori_A}. Note that a similar bound holds for $ \partial_{ki}\ol A_{abcd}$.
Thus, using \eqref{eq:derivative_A} and \eqref{eq:bound_Y}, we get
\begin{align*}
 \E  \sum_{i,j,k} \breve\sigma_{a i} \tau_{b j}S_{ki}G_{kc}G_{dj}G_{ji} \partial_{ki} (A_{abcd}^{p-1} \overline A_{abcd}^{p})  \prec (\alpha \Psi^5 + \Psi^{10})\E |A_{abcd}|^{2p-2},
\end{align*}
which implies, via Lemma \ref{lemma:moment_to_domination}, that for $c \neq d$ we have $
A_{abcd} \prec \Psi^5$ and in general 
\begin{align}\label{eq:A_bound}
A_{abcd} \prec \Psi^5 + \delta_{cd}\Psi^4.
\end{align}
Thus, by using \eqref{eq:A_bound} and \eqref{eq:bound_Y} in \eqref{eq:expansion_Z_pq} we get
\begin{align}\label{eq:temp_bound_Z}
\E |Z_{ab}|^{2p} \prec (\xi_{ab} + \Psi^4)\E |Z_{ab}|^{2p-1} + \Psi^8 \E |Z_{ab}|^{2p-2}. 
\end{align}
Therefore, Lemma \ref{lemma:moment_to_domination}
implies that $ Z_{ab} \prec \xi_{ab}+ \Psi^4$.
\end{proof}

\subsection{Chain estimates: proof of \eqref{eq:chain_bound}}
We will proceed by induction on $n$: thanks to \eqref{eq:bound_Y}, we know that $Y^{(1)}_{ab;u} \prec \Psi^{3} + \delta_{ab}\Psi^2$ and we assume as induction hypothesis that 
\begin{align}\label{eq:induction_hyp_Y}
Y_{ab;u_1 \cdots u_{n-1}}^{(n-1)} \prec \Psi^{2(n-1) + 1} + \delta_{ab}\Psi^{2(n-1)}.
\end{align}
In the following we will simplify the notation by setting $Y_{ab;u_1 \cdots u_{n}}^{(n)} \equiv Y_{ab}^{(n)}  $. 
For any fixed integer $p \geq 1$ Lemma \ref{lemma_self-consistent_equation_trick} implies that 
%$$
%\E| Y_{ab}^{(n)}|^{2p} = \sum_{i_1,\ldots, i_n} \tilde \sigma^{(1)}_{u_1 i_1} \sigma^{(2)}_{u_2 i_2} \cdots \sigma^{(n)}_{u_n i_n}G_{ai_1}(zG_{i_1 i_2}) \cdots G_{i_n b}(Y_{ab}^{(n)})^{p-1} (\overline Y_{ab}^{(n)})^p\,,
%$$
%where $\tilde \sigma^{(1)}, \sigma^{(2)}, \ldots, \sigma^{(n)} \in \bb S$. By invoking Lemma \ref{lemma_self-consistent_equation_trick} with 
%$$
%\mathfrak p(G, \ol G) = \sum_{i_2,\ldots, i_n} \sigma^{(2)}_{u_2 i_2} \cdots \sigma^{(n)}_{u_n i_n} G_{i_2 i_3} \cdots G_{i_n b}(Y_{ab}^{(n)})^{p-1} (\overline Y_{ab}^{(n)})^p 
%$$ in the same way we have done previously for $Y^{(1)}_{ab;u}$, one obtains for some $\sigma^{(1)} \in \bb S$
\begin{align}\label{eq:expansion_Yn}\nonumber
 \E|Y_{ab}^{(n)}|^{2p}  \prec & \,\,
\bigg\vert \E \sum_{i_2, \ldots, i_{n}} \breve \sigma_{u_1 i_2}^{(1)} \sigma_{u_2 i_2}^{(2)} \cdots \sigma_{u_n i_n}^{(n)} G_{ai_{2}} G_{i_2 i_3} \cdots G_{i_n b} (Y_{ab}^{(n)})^{p-1}(\overline Y_{ab}^{(n)})^{p} \bigg\vert \\\nonumber
& + \bigg\vert \E \sum_{i_1, \ldots, i_{n+1}}\breve \sigma_{u_1 i_1}^{(1)} \cdots \sigma_{u_n i_n}^{(n)} S_{i_1 i_{n+1}} \\\nonumber
&\qquad \times [G_{ai_1}(G_{i_{n+1}i_{n+1}}- \mathfrak{m}) G_{i_1 i_2} + G_{ai_{n+1}}(G_{i_{1}i_{1}}-\mathfrak{m}) G_{i_{n+1} i_2} ] G_{i_2 i_3} \cdots G_{i_n b} (Y_{ab}^{(n)})^{p-1}(\overline Y_{ab}^{(n)})^{p} \bigg\vert \\\nonumber
& + \bigg\vert \E \sum_{i_1, \ldots, i_{n+1}}\breve \sigma_{u_1 i_1}^{(1)} \cdots \sigma_{u_n i_n}^{(n)} S_{i_1 i_{n+1}}G_{ai_{n+1}} G_{i_1 i_2}[\partial_{i_1i_{n+1}}G_{i_2 i_3} \cdots G_{i_n b}] (Y_{ab}^{(n)})^{p-1}(\overline Y_{ab}^{(n)})^{p} \bigg\vert \\\nonumber
& + \bigg\vert \E \sum_{i_1, \ldots, i_{n+1}}\breve \sigma_{u_1 i_1}^{(1)} \cdots \sigma_{u_n i_n}^{(n)} S_{i_1 i_{n+1}}G_{ai_{n+1}} G_{i_1 i_2} G_{i_2 i_3} \cdots G_{i_n b} [\partial_{i_1i_{n+1}}(Y_{ab}^{(n)})^{p-1}(\overline Y_{ab}^{(n)})^{p}] \bigg\vert \\
& = (\mathrm{A}) + (\mathrm{B}) + (\mathrm{C}) + (\mathrm{D})\,.
\end{align}
Let us now deal with the terms $(\mathrm{A})$, $(\mathrm{B})$,  $(\mathrm{C})$ and $(\mathrm{D})$. To simplify the notation we rename $\breve\sigma^{(1)}$ by $\sigma^{(1)}$.

\begin{itemize}
\item[(A)]
Let us define the tensor $\Xi_{ab}^{c}:= W\sigma^{(1)}_{ab}\sigma^{(2)}_{cb}$. It is immediate to verify that for any $c$ we have that $\Xi^c \in \bb S$. Therefore, by using the induction hypothesis \eqref{eq:induction_hyp_Y} one has
\begin{align*}
(\mathrm{A}) & \prec 
W^{-1} \bigg\vert \E \sum_{i_2, \ldots, i_{n}} \Xi^{u_2}_{u_1 i_2} \sigma_{u_3 i_3}^{(3)} \cdots \sigma_{u_n i_n}^{(n)} G_{ai_{2}} G_{i_2 i_3} \cdots G_{i_n b} (Y_{ab}^{(n)})^{p-1}(\overline Y_{ab}^{(n)})^{p} \bigg\vert \\
& \prec (\Psi^{2}\E|Y^{(n-1)}_{ab}| |Y^{(n)}_{ab}|^{2p-1}) \prec (\Psi^{2n+1} + \delta_{ab}\Psi^{2n}) \E |Y^{(n)}_{ab}|^{2p-1}.
\end{align*}
\item[(B)] From \eqref{eq:induction_hyp_Y} we obtain 
\begin{align*}\nonumber
(\mathrm B) & = \bigg\vert \, \, \E \sum_{i_1, i_{n+1}}\sigma_{u_1 i_1}^{(1)} S_{i_1 i_{n+1}}
[(G_{i_{n+1}i_{n+1}}-\mathfrak m)G_{ai_1}Y_{i_1b}^{(n-1)} + (G_{i_{1}i_{1}}-\mathfrak m)G_{ai_{n+1}}Y_{i_{n+1}b}^{(n-1)}](Y_{ab}^{(n)})^{p-1}(\overline Y_{ab}^{(n)})^{p} \bigg\vert \\
&  \prec \Psi^{2n+1} \E |Y_{ab}^{(n)}|^{2p-1}.
\end{align*}
\item[(C)] Using \eqref{eq:induction_hyp_Y}, one has
\begin{align*}\nonumber
(\mathrm C) \leq & \, \, \bigg\vert \E \sum_{i_1, \ldots, i_{n+1}}\sigma_{u_1 i_1}^{(1)} \cdots \sigma_{u_n i_n}^{(n)} S_{i_1 i_{n+1}}G_{ai_{n+1}} G_{i_1 i_2} \sum_{k=2}^{n-1}G_{i_k i_1}G_{i_{n+1} i_{k+1}} \bigg( \prod_{\substack{l=2 \\ l \neq k}}^{n-1} G_{i_l i_{l+1}} \bigg) G_{i_n b} (Y_{ab}^{(n)})^{p-1}(\overline Y_{ab}^{(n)})^{p} \bigg\vert \\\nonumber
& + \bigg\vert \E \sum_{i_1, \ldots, i_{n+1}}\sigma_{u_1 i_1}^{(1)} \cdots \sigma_{u_n i_n}^{(n)} S_{i_1 i_{n+1}}G_{ai_{n+1}} G_{i_1 i_2} \cdots G_{i_ni_1}G_{i_{n+1} b} (Y_{ab}^{(n)})^{p-1}(\overline Y_{ab}^{(n)})^{p} \bigg\vert \\\nonumber
\prec & \, \sum_{k = 1}^{n-1} \E \sum_{i_1} |\sigma_{u_1 i_1}^{(1)}||Y^{(k)}_{i_1i_1}|| Y_{ab}^{(n-k)}| |Y_{ab}^{(n)}|^{2p-1} \prec \, (\Psi^{2n+1} + \delta_{ab}\Psi^{2n})\E|Y^{(n)}_{ab}|^{2p-1}
\end{align*} 

\item[(D)] Note that
\begin{align}\label{eq:diff_Yn1}
\partial_{i_1i_{n+1}} Y^{(n)}_{ab}& = - \sum_{k =0}^{n} Y^{(k)}_{ai_1} Y^{(n-k)}_{i_{n+1}b}, \ \ \
\partial_{i_1i_{n+1}} \overline Y^{(n)}_{ab} = - \sum_{k =0}^{n} \overline Y^{(k)}_{ai_{n+1}} \overline Y^{(n-k)}_{i_{1}b} 
\end{align}
where $Y^{(0)}_{ab} := G_{ab}$. Therefore we get for some $ \xi \in \bb S$ (ignoring the complex conjugation which does not play any role here)
\begin{align*}
(\mathrm D) & \prec \bigg\vert \E \sum_{k=0}^{n}\sum_{i_1, i_{n+1}}\sigma_{u_1 i_1}^{(1)} S_{i_1 i_{n+1}}G_{ai_{n+1}} Y_{ai_1}^{(k)} Y_{i_{n+1}b}^{(n-k)} Y_{i_{1}b}^{(n-1)} (Y_{ab}^{(n)})^{2p-2}\bigg\vert \\\nonumber
& \prec \bigg\vert \E \sum_{k=1}^{n-1}\sum_{i_1, i_{n+1}}\sigma_{u_1 i_1}^{(1)} S_{i_1 i_{n+1}}G_{ai_{n+1}} Y_{ai_1}^{(k)} Y_{i_{n+1}b}^{(n-k)} Y_{i_{1}b}^{(n-1)} (Y_{ab}^{(n)})^{2p-2} \bigg\vert \\\nonumber
& + \bigg\vert \E\sum_{i_1 i_{n+1}}\sigma_{u_1, i_1}^{(1)} S_{i_1 i_{n+1}}G_{ai_{n+1}}[ G_{ai_1} Y_{i_{n+1}b}^{(n)} + G_{i_{n+1}b} Y_{ai_1}^{(n)}] Y_{i_{1}b}^{(n-1)} (Y_{ab}^{(n)})^{2p-2} \bigg\vert \\\nonumber
& \prec (\Psi^{4n+2} + \delta_{ab}\Psi^{4n})\E |Y_{ab}^{(n)}|^{2p-2} + (\mathrm E)
\end{align*}
where
 $$(\mathrm E) = \E\sum_{i_1, \ldots, i_{n+1}}\sigma_{u_1 i_1}^{(1)}S_{i_1 i_{n+1}}G_{a i_{n+1}}[ G_{a i_1} Y_{i_{n+1}b}^{(n)} + G_{i_{n+1}b} Y_{ai_1}^{(n)}] Y_{i_{1}b}^{(n-1)} (Y_{ab}^{(n)})^{2p-2}\,.
 $$
In order to bound (E) let suppose that $a \neq b$ and that 
\begin{align}\label{eq:Y_n_apriori}
Y_{ab}^{(n)} \prec \lambda^{(n)}, \ \ \ \lambda^{(n)} \in [\Psi^{2n+1}, L^C]. 
\end{align}
Note that from \eqref{eq:induction_hyp_Y} we can easily get that $Y_{xy}^{(n)} \prec \Psi^{2n}$ for any $ x, y \in \bb T_L$, so from \eqref{eq:Y_n_apriori} we have
\begin{align*}
& (\mathrm E) \prec (\lambda^{(n)} \Psi^{2n+1} + \Psi^{4n+2})\E|Y_{ab}^{(n)}|^{2p-2}, \ \ \ a \neq b \\\nonumber
& (\mathrm E) \prec \Psi^{4n} \E|Y_{ab}^{(n)}|^{2p-2}, \ \ \ a = b.
\end{align*}
\end{itemize}

Putting together the contribution form all the terms (A), (B), (C), (D), (E) and using Lemma \ref{lemma:moment_to_domination}, we can conclude that $Y_{ab}^{(n)} \prec \Psi^{2n+1} + \delta_{ab}\Psi^{2n}$.

\subsection{Loop estimates: proof of \eqref{eq:loop_bound}}
As for \eqref{eq:chain_bound}, we will proceed by induction on $n$: from \eqref{eq:bound_Z} we know that $Z^{(2)}_{ab} \prec \Psi^{4} + \xi_{ab}$ where $\xi \in \bb S$. 
As induction hypothesis, let us assume that
\begin{align}\label{eq:induction_hyp_Z}
Z^{(n-1)}_{ab;u_3 \cdots u_{n-1}} \prec \Psi^{2(n-1)} + \xi_{ab} \Psi^{2(n-3)}, \ \ \ \xi \in \bb S.
\end{align}
We need the following auxiliary lemma. Recall that $ g^\flat = g - \E g$. 
\begin{lemma}\label{lemma:bound_B}
Let
\begin{align*}
B^{(n)}_{ab}:= \sum_{i_1, \ldots, i_n, i_{n+1}} \sigma_{a i_1}^{(1)} \sigma_{b i_2}^{(2)} \cdots \sigma_{u_n i_n}^{(n)} S_{i_1 i_{n+1}} G_{i_1 i_1}^\flat G_{i_{n+1} i_{2}} G_{i_n i_{n+1}} G_{i_2 i_3} \cdots G_{i_{n-1} i_n}
\end{align*}
with $ {\sigma}^{(1)}, \sigma^{(2)}, \ldots, \sigma^{(n)} \in \bb S$. Then
\begin{align}\label{eq:B_bound}
B^{(n)}_{ab} \prec \Psi^{2n}.
\end{align}
\end{lemma}

\begin{proof}
See Section \ref{sec:bound_B}. 
\end{proof}

In the following we will adopt the simplified notation $Z^{(n)}_{ab} \equiv Z^{(n)}_{ab;u_3 \cdots u_{n}} $. As before, from Lemma \ref{lemma_self-consistent_equation_trick} for any $p \geq 1$ we deduce that 
%consider the generalized moment of $ Z^{(n)}_{ab}$ by defining $ \cal Z^{(n)}_{abcd} := \E(\cal Z^{(n)}_{ab})(\cal Z^{(n)}_{cd})^{p-1}(\overline{\cal Z}^{(n)}_{cd})^{p} $ and we look at
%\begin{align*}
%z\cal Z^{(n)}_{abcd} = \E \sum_{i_1,...,i_n} \tilde \sigma^{(1)}_{ai_1}\sigma^{(2)}_{bi_2}\sigma^{(3)}_{u_3 i_3}\cdots \sigma^{(n)}_{u_n i_n} (zG_{i_1,i_2})\cdots G_{i_{n-1}i_n}G_{i_ni_1}(\cal Z^{(n)}_{cd})^{p-1}(\overline{\cal Z}^{(n)}_{cd})^{p}
%\end{align*}
%with $\tilde \sigma^{(1)}, \ldots, \sigma^{(n)} \in \bb S$. Setting 
%\begin{align*}
%\mathfrak p(G, \ol G) = \E \sum_{i_2,...,i_n} \sigma^{(2)}_{bi_2}\sigma^{(3)}_{u_3 i_3}\cdots \sigma^{(n)}_{u_n i_n} G_{i_2i_3}\cdots G_{i_{n-1}i_n}(\cal Z^{(n)}_{cd})^{p-1}(\overline{\cal Z}^{(n)}_{cd})^{p}, 
%\end{align*}
%we can employ Lemma \ref{lemma_self-consistent_equation_trick} to get for some $\sigma^{(1)} \in \bb S$
\begin{align*}\nonumber
& \E |Z^{(n)}_{ab}|^{2p} \prec \bigg\vert \E \sum_{i_2, \ldots, i_n} \breve\sigma_{a i_2}^{(1)} \sigma_{b i_2}^{(2)} \cdots \sigma_{u_n i_n}^{(n)} G_{i_2 i_3} \cdots G_{i_{n-1} i_n} G_{i_n i_2}(Z_{ab}^{(n)})^{p-1}(\overline Z_{ab}^{(n)})^{p} \bigg\vert \\\nonumber
&  +\bigg\vert \E \sum_{i_1, \ldots, i_{n+1}} \breve\sigma_{a i_1}^{(1)} \sigma_{b i_2}^{(2)} \cdots \sigma_{u_n i_n}^{(n)} S_{i_1 i_{n+1}} (G_{i_{n+1}i_{n+1}} - \mathfrak m) G_{i_ni_1}G_{i_1 i_2} G_{i_2 i_3} \cdots G_{i_{n-1} i_n} (Z_{ab}^{(n)})^{p-1}(\overline Z_{ab}^{(n)})^{p} \bigg\vert \\\nonumber
& +\bigg\vert \E \sum_{i_1, \ldots, i_{n+1}} \breve\sigma_{a i_1}^{(1)} \sigma_{b i_2}^{(2)} \cdots \sigma_{u_n i_n}^{(n)} S_{i_1 i_{n+1}}(G_{i_1 i_1} - \mathfrak m)  G_{i_n i_{n+1}} G_{i_{n+1} i_{2}}G_{i_2 i_3} \cdots G_{i_{n-1} i_n} (Z_{ab}^{(n)})^{p-1}(\overline Z_{ab}^{(n)})^{p}\bigg\vert \\\nonumber
& +\bigg\vert \E \sum_{i_1, \ldots, i_{n+1}} \breve\sigma_{a i_1}^{(1)} \sigma_{b i_2}^{(2)} \cdots \sigma_{u_n i_n}^{(n)} S_{i_1 i_{n+1}}G_{i_n i_1}G_{i_{n+1} i_2} \sum_{k = 2}^{n-1} G_{i_k i_{n+1}} G_{i_1 i_{k+1}} \bigg( \prod_{\substack{l = 2 \\ l \neq k}}^{n-1} G_{i_l i_{l+1}} \bigg) (Z_{ab}^{(n)})^{p-1}(\overline Z_{ab}^{(n)})^{p} \bigg\vert \\\nonumber
& + \bigg\vert \E \sum_{i_1, \ldots, i_{n+1}} \breve\sigma_{a i_1}^{(1)} \sigma_{b i_2}^{(2)} \cdots \sigma_{u_n i_n}^{(n)} S_{i_1 i_{n+1}} G_{i_{n+1} i_2} G_{i_2 i_3} \cdots G_{i_{n-1} i_n} G_{i_n i_1} \partial_{ i_{n+1}i_1}(Z_{ab}^{(n)})^{p-1}(\overline Z_{ab}^{(n)})^{p} \bigg\vert \\
& = (\mathrm A) + (\mathrm B) + (\mathrm C) + (\mathrm D) + (\mathrm E)\,.
\end{align*}
In the following we simplify the notation by renaming $ \breve \sigma^{(1)}$ as $\sigma^{(1)}$.
Let us now look at the different terms. 
\begin{itemize}
\item[(A)] By using \eqref{eq:chain_bound} and the properties of the matrix family $\bb S$, for some $\xi \in \bb S$ one easily gets that
\begin{align*}
(\mathrm A) \prec \Psi^{2(n-2)}\sum_{i_2}  |\sigma_{ai_2}^{(1)}| |\sigma_{bi_2}^{(2)}| \E |Z^{(n)}_{ab}|^{2p-1} \prec \xi_{a b}\Psi^{2(n-2)} \E |Z^{(n)}_{ab}|^{2p-1}.
\end{align*}

\item[(B)] We note that
\begin{align*}
(\mathrm B) \prec \bigg\vert \E \sum_{i_1 ,i_{n+1}} \sigma_{ai_1}^{(1)} S_{i_1i_{n+1}} (G_{i_{n+1}i_{n+1}} - \mathfrak m) Y^{(n-1)}_{i_1 i_1 } (Z_{ab}^{(n)})^{p-1}(\overline Z_{ab}^{(n)})^{p} \bigg\vert,
\end{align*}
thus from \eqref{eq:X_bound} and \eqref{eq:chain_bound} we get $(\mathrm B) \prec \Psi^{2n} \E |Z_{cd}^{(n)}|^{2p-1}.$

\item[(C)] 
From Lemma \ref{lemma:bound_B} and the fact that $\E G_{kk} - \mathfrak{m} \prec \Psi^2$, one can conclude that $(\mathrm C) \prec \Psi^{2n} \E |Z_{ab}^{(n)}|^{2p-1}$.

\item[(D)] We note that
\begin{align*}
(\mathrm D) & \prec \bigg\vert \E \sum_{k=2}^{n-1} \sum_{i_1} \sigma_{a i_1}^{(1)} Z^{(k)}_{i_1 b} Y^{(n-k)}_{i_1 i_1}(Z_{ab}^{(n)})^{p-1}(\overline Z_{ab}^{(n)})^{p} \bigg\vert \prec \bigg( \Psi^{2n} + \Psi^{2(n-2)}\sum_{i_1}  |\sigma_{ai_1}^{(1)}| \xi'_{i_1b} \bigg) \E |Z_{ab}^{(n)}|^{2p-1} \\
& \prec (\Psi^{2n} + \xi_{ab}\Psi^{2(n-2)})\E |Z_{ab}^{(n)}|^{2p-1}
\end{align*}
where we used \eqref{eq:induction_hyp_Z} and \eqref{eq:chain_bound}. Here $\xi_{ab} = \sum_{i} |\sigma_{ai}^{(1)}|\xi'_{ib}$ and $\xi, \xi' \in \bb S$. 

\item[(E)] Note that
\begin{align*}
& \partial_{i_{n+1}i_1} Z^{(n)}_{ab} = - \sum_{k=1}^{n} \sum_{j_1, \ldots,j_n} \sigma_{a j_1}^{(1)} \sigma_{b j_2}^{(2)} \cdots \sigma_{u_n j_n}^{(n)}  G_{j_k i_{n+1}} G_{i_1 j_{k+1}} \prod_{l = 1, l \neq k}^{n} G_{j_l j_{l+1}} \bigg \vert_{j_{n+1}= j_1} \\\nonumber
& \partial_{i_{n+1}i_1} \overline{Z}^{(n)}_{ab} 
= - \sum_{k=1}^{n} \sum_{j_1, \ldots, j_n} \sigma_{a j_1}^{(1)} \sigma_{b j_2}^{(2)} \cdots \sigma_{u_n j_n}^{(n)} 
\overline{G}_{j_k i_{1}}
 \overline{G}_{i_{n+1} j_{k+1}} \prod_{l = 1, l \neq k}^{n} \overline G_{j_l j_{l+1}} \bigg \vert_{j_{n+1}= j_1}, 
\end{align*}
thus from \eqref{eq:chain_bound} we get
\begin{align*}
(\mathrm E) \prec \E \sum_{i_i, i_{n+1}}  |\sigma_{ai_1}^{(1)}| S_{i_1 i_{n+1}} |Y_{i_1 i_{n+1}}^{(n)}|| Y_{i_i i_{n+1}}^{(n-1)}| |Z_{ab}^{(n)}|^{2p-2} \prec \Psi^{4n}\E|Z_{ab}^{(n)}|^{2p-2}.
\end{align*}
\end{itemize}

Collecting all the terms and applying Lemma \ref{lemma:moment_to_domination} we get the claim.

\subsection{Proof of Lemma \ref{lemma:bound_B}}\label{sec:bound_B}
From \eqref{eq:B_basic_bound} we know that $B^{(2)}_{ab} \prec \Psi^4$. We want to show that $B^{(n)}_{ab} \prec \Psi^{2n}$ and we will proceed by induction by setting the hypothesis
\begin{align}\label{eq:induction_hyp_B}
B^{(n-1)}_{ab} \prec \Psi^{2(n-1)}.
\end{align}
For any fixed integer $p \geq 1$ we consider the moment $ \E|B_{ab}|^{2p}$ 
%\begin{align*}
%\E|B^{(n)}_{ab}|^{2p} := \E \sum_{i_1, \ldots, i_{n+1}} \sigma_{a i_1}^{(1)} \sigma_{b i_2}^{(2)} \cdots \sigma_{u_n i_n}^{(n)} S_{i_1 i_{n+1}} G_{i_1 i_1}^\flat  G_{i_n i_{n+1}} G_{i_{n+1} i_{2}}G_{i_2 i_3} \cdots G_{i_{n-1} i_n}( B^{(n)}_{ab})^{p-1} (\overline B^{(n)}_{ab})^{p} 
%\end{align*}
%with $ \sigma^{(1)}, \sigma^{(2)}, \ldots, \sigma^{(n)} \in \bb S$. 
and by invoking Lemma \ref{lemma_self-consistent_equation_trick}, 
%with
%$$
%\mathfrak p(G, \ol G) = \E \sum_{i_1, \ldots, i_{n-2}, i_n, i_{n+1}} \sigma_{a i_1}^{(1)} \sigma_{b i_2}^{(2)} \cdots \sigma_{u_{n-2} i_{n-2}}^{(n-2)} \sigma_{u_n i_n}^{(n)} S_{i_1 i_{n+1}} \langle G_{i_1 i_1} \rangle  G_{i_n i_{n+1}} G_{i_{n+1} i_{2}}G_{i_2 i_3} \cdots G_{i_{n-3} i_{n-2}}( B^{(n)}_{cd})^{p-1} (\overline B^{(n)}_{cd})^{p} 
%$$
we get
\begin{align*}
\E |B^{(n)}_{ab}|^{2p} \prec (\mathrm A) + (\mathrm B) + (\mathrm C) + (\mathrm D) + (\mathrm E) + (\mathrm F)
\end{align*}
where the explicit form of (A), (B), (C), (D), (E) and (F) is given below.
We proceed term by term. 
\begin{itemize}
\item[(A)] By using the induction hypothesis \eqref{eq:induction_hyp_B} and the properties of the matrices in $\bb S$ one has
\begin{align*}
 (\mathrm A) =& \,\, \bigg\vert \E \sum_{i_1, \ldots, i_{n-1}, i_{n+1}} \sigma_{a i_1}^{(1)} \sigma_{b i_2}^{(2)} \cdots \sigma_{u_{n-1} i_{n-1}}^{(n-1)} \breve\sigma_{u_n i_{n-1}}^{(n)} S_{i_1 i_{n+1}} \\\nonumber
 & \times  G_{i_1 i_1}^\flat  G_{i_{n-1} i_{n+1}} G_{i_{n+1} i_{2}} G_{i_2 i_3} \cdots G_{i_{n-2} i_{n-1}}( B^{(n)}_{ab})^{p-1} (\overline B^{(n)}_{ab})^{p} \bigg\vert \\\nonumber
 \prec & \,\, W^{-1} \Psi^{2(n-1)}\E |B_{ab}^{(n)}|^{2p-1} \prec \Psi^{2n}\E |B_{ab}^{(n)}|^{2p-1} \,.
\end{align*}
\item[(B)] From \eqref{eq:chain_bound} we have
\begin{align*}\nonumber
(\mathrm B) =& \,\, \bigg\vert \E \sum_{i_1, \ldots, i_{n}, i_{n+1}, i_{n+2}} \sigma_{a i_1}^{(1)} \sigma_{b i_2}^{(2)} \cdots \sigma_{u_{n-1} i_{n-1}}^{(n-1)} \breve\sigma_{u_n i_{n}}^{(n)} S_{i_1 i_{n+1}}S_{i_{n-1}i_{n+2}} \\\nonumber
& \times G_{i_1 i_1}^\flat  G_{i_{n} i_{n+1}} G_{i_{n+1} i_{2}} G_{i_2 i_3} \cdots G_{i_{n-3} i_{n-2}}G_{i_{n-2}i_{n-1}}G_{i_{n-1}i_n} (G_{i_{n+2}i_{n+2}}-\mathfrak m)( B^{(n)}_{ab})^{p-1} (\overline B^{(n)}_{ab})^{p} \bigg\vert \\
\prec & \,\, \E \sum_{i_1 i_{n-1} i_{n+2}} | \sigma_{a i_1}^{(1)} \sigma_{u_{n-1}i_{n-1}}^{(n-1)} S_{i_{n-1}i_{n+2}} G_{i_1 i_1}^\flat (G_{i_{n+2}i_{n+2}}-\mathfrak m)Y^{(n-1)}_{i_{n-1}i_{n-1}}| |B^{(n)}_{ab}|^{2p-1}\prec \Psi^{2n}\E|B^{(n)}_{ab}|^{2p-1} \,.
\end{align*}
\item[(C)] Similarly from \eqref{eq:chain_bound} we have
\begin{align*}\nonumber
(\mathrm C) =& \, \bigg\vert \E \sum_{i_1, \ldots, i_{n+2}} \sigma_{a i_1}^{(1)} \sigma_{b i_2}^{(2)} \cdots \sigma_{u_{n-1} i_{n-1}}^{(n-1)} \breve\sigma_{u_n i_{n}}^{(n)} S_{i_1 i_{n+1}} S_{i_{n-1}i_{n+2}}\\\nonumber
& \times G_{i_1 i_1}^\flat  G_{i_{n} i_{n+1}} G_{i_{n+1} i_{2}} G_{i_2 i_3} \cdots G_{i_{n-3} i_{n-2}}G_{i_{n-2}i_{n+2}} G_{i_{n+2}i_n}(G_{i_{n-1}i_{n-1}}-\mathfrak m)( B^{(n)}_{ab})^{p-1} (\overline B^{(n)}_{ab})^{p} \bigg\vert \\
\prec & \,\, \E \sum_{i_1, i_{n-1}, i_{n+2}} | \sigma_{a i_1}^{(1)} \sigma_{u_{n-1}i_{n-1}}^{(n-1)} S_{i_{n-1}i_{n+2}} G_{i_1 i_1}^\flat (G_{i_{n-1}i_{n-1}}-\mathfrak m)Y^{(n-1)}_{i_{n+2}i_{n+2}}| |B^{(n)}_{ab}|^{2p-1}\prec \Psi^{2n}\E|B^{(n)}_{ab}|^{2p-1} \,.
\end{align*}
\item[(D)] By \eqref{eq:chain_bound} one has
\begin{align*}\nonumber
(\mathrm D) =& \,\, \bigg\vert \E \sum_{i_1, \ldots, i_{n+2}} \sigma_{a i_1}^{(1)} \sigma_{b i_2}^{(2)} \cdots \sigma_{u_{n-1} i_{n-1}}^{(n-1)}\breve \sigma_{u_n i_{n}}^{(n)} S_{i_1 i_{n+1}} S_{i_{n-1}i_{n+2}}\\\nonumber
& \times G_{i_1 i_{n+2}} G_{i_{n+2}i_n}  G_{i_{n} i_{n+1}} G_{i_{n+1} i_{2}} G_{i_2 i_3} \cdots G_{i_{n-3} i_{n-2}}G_{i_{n-2}i_{n-1}}G_{i_{n-1}i_1} ( B^{(n)}_{ab})^{p-1} (\overline B^{(n)}_{ab})^{p} \bigg\vert \\
 \prec & \,\, \E \sum_{i_1, i_{n+2}, i_{n-1}} |\sigma_{a i_1}^{(1)} \sigma_{u_{n-1} i_{n-1}}^{(n-1)} S_{i_{n-1}i_{n+2}} G_{i_1 i_{n+2}} G_{i_{n-1}i_1}Y^{(n-1)}_{i_{n+2}i_{n-1}}( B^{(n)}_{ab})^{p-1} (\overline B^{(n)}_{ab})^{p}| \prec \Psi^{2n} \E|B_{ab}^{(n)}|^{2p-1}\,.
\end{align*}
\item[(E)] Again using \eqref{eq:chain_bound}, one gets
\begin{align*}\nonumber
(\mathrm E) = & \,\, \bigg\vert \E \sum_{i_1, \ldots, i_{n+2}} \sigma_{a i_1}^{(1)} \sigma_{b i_2}^{(2)} \cdots \sigma_{u_{n-1} i_{n-1}}^{(n-1)} \breve\sigma_{u_n i_{n}}^{(n)} S_{i_1 i_{n+1}} S_{i_{n-1}i_{n+2}}\\\nonumber
& \times G_{i_1 i_1}^\flat G_{i_{n+2}i_n} G_{i_{n-2} i_{n-1}} \bigg[\sum_{k = 0}^{n-3}G_{i_k i_{n+2}}G_{i_{n-1}i_{k+1}} \bigg(\prod_{\substack{l=0 \\ l \neq k}}^{n-3} G_{i_l i_{l+1}}  \bigg)\bigg \vert_{i_0 = i_n} \bigg]( B^{(n)}_{ab})^{p-1} (\overline B^{(n)}_{ab})^{p} \bigg\vert \\\nonumber
 \prec & \,\, \E \sum_{k=0}^{n-3}\sum_{i_{n-1}}|\sigma_{u_{n-1}i_{n-1}}^{(n-1)} B_{ab}^{(k+2)}Y^{(n-2-k)}_{i_{n-1}i_{n-1}}||B^{(n)}_{ab}|^{2p-1} \prec \Psi^{2n}\E |B^{(n)}_{ab}|^{2p-1} \,.
\end{align*}
\item[(F)] We note that
\begin{align}\nonumber
& \partial_{i_{n+2}i_{n-1}} B_{ab}^{(n)} = \sum_{j_1, \ldots, j_n, j_{n+1}} \sigma_{a i_1}^{(1)} \sigma_{b i_2}^{(2)} \cdots \sigma_{u_n i_n}^{(n)} S_{j_1 j_{n+1}} \\\nonumber
& \times \bigg[G_{j_1i_{n+2}}G_{i_{n-1}j_1} \bigg( \prod_{l=2}^{n+1} G_{j_l j_{l+1}} \bigg \vert_{j_{n+2} = j_2 } \bigg) + G_{j_1 j_1}^\flat  \sum_{k=2}^{n+1} G_{j_k i_{n+2}} G_{i_{n-1}j_{k+1}} \bigg( \prod_{l =2, l \neq k}^{n+1} G_{j_l j_{l+1}}\bigg \vert_{j_{n+2} = j_2} \bigg) \bigg] \label{eq:diff_Bn1}\\
& \partial_{i_{n+2}i_{n-1}} \overline B_{ab}^{(n)} = \sum_{j_1, \ldots, j_n, j_{n+1}} \sigma_{a i_1}^{(1)} \sigma_{b i_2}^{(2)} \cdots \sigma_{u_n i_n}^{(n)} S_{j_1 j_{n+1}} \\
& \times \bigg[\overline G_{j_1i_{n-1}} \overline G_{i_{n+2}j_1} \bigg( \prod_{l=2}^{n+1} \overline G_{j_l j_{l+1}} \bigg \vert_{j_{n+2} = j_2 } \bigg) + \overline G_{j_1 j_1}^\flat  \sum_{k=2}^{n+1} \overline G_{j_k i_{n-1}} \overline G_{i_{n+2}j_{k+1}} \bigg( \prod_{l =2, l \neq k}^{n+1} \overline G_{j_l j_{l+1}}\bigg \vert_{j_{n+2} = j_2} \bigg) \bigg] \label{eq:diff_Bn2},
\end{align}
thus from \eqref{eq:diff_Bn1} and \eqref{eq:diff_Bn2} and by using the trivial bound $ Z_{ab}^{(n)} \prec \Psi^{2(n-1)}$ we get
\begin{align*}\nonumber
(\mathrm F) =& \,\bigg\vert \E \sum_{i_1, \ldots, i_{n+2}} \sigma_{a i_1}^{(1)} \sigma_{b i_2}^{(2)} \cdots \breve\sigma_{u_n i_n}^{(n)} S_{i_1 i_{n+1}} S_{i_{n+2} i_{n-1}} \\\nonumber
& \times G_{i_1 i_1}^\flat  G_{i_{n+2} i_n} G_{i_n i_{n+1}} G_{i_{n+1} i_{2}}G_{i_2 i_3} \cdots G_{i_{n-2}i_{n-1}} \partial_{i_{n+2}i_{n-1}}( B^{(n)}_{ab})^{p-1} (\overline B^{(n)}_{ab})^{p} \bigg\vert \\\nonumber
\prec & \,\E \sum_{i_1, j_1, i_{n-1}, i_{n+2}} |\sigma_{ai_1}^{(1)} \sigma_{a j_1}^{(1)} \sigma_{u_{n-1}i_{n-1}}^{(n-1)} S_{i_{n+2} i_{n-1}} G_{i_1 i_1}^\flat | \\\nonumber
& \times (|G_{j_1 i_{n+2}} G_{i_{n-1}j_1}| + |G_{j_1 i_{n-1}} G_{i_{n+2}}j_1|)| Y^{(n-1)}_{i_{n+2} i_{n-1}} Z^{(n)}_{j_1b} | |B^{(n)}_{ab}|^{2p-2} \\\nonumber
& + \E \sum_{i_1, j_1, i_{n-1}, i_{n+2}} |\sigma_{ai_1}^{(1)} \sigma_{a j_1}^{(1)} \sigma_{u_{n-1}i_{n-1}}^{(n-1)} S_{i_{n+2} i_{n-1}} G_{i_1 i_1}^\flat G_{j_1 j_1}^\flat Y^{(n-1)}_{i_{n+2} i_{n-1}} Y^{(n)}_{i_{n+2} i_{n-1}} | |B^{(n)}_{ab}|^{2p-2} \\\nonumber
& \prec \Psi^{4n} \E |B^{(n)}_{ab}|^{2p-2}\,.
\end{align*}
\end{itemize}
In conclusion, we obtained $\E|B^{(n)}_{ab}|^{2p} \prec \Psi^{2n}\E |B^{(n)}_{ab}|^{2p-1} + \Psi^{4n}\E |B^{(n)}_{ab}|^{2p-2}$ which implies the desired estimate thanks to Lemma \ref{lemma:moment_to_domination}.

\subsection{Proof of \eqref{eq:X_bound}}
Since from \eqref{eq:bound_calG} we have that $ \E \sum_a \sigma_{ia} (G_{aa} - \mathfrak m) \prec \Psi^2$, to prove the claim we have to estimate
$$
X_i' = \sum_a \sigma_{ia} G_{aa}^\flat.
$$
We proceed as in the proof of Lemma \ref{lemma_self-consistent_equation_trick}:  consider $ z\bb E|X_i'|^{2p}= \E \sum_a \sigma_{ia}  z G_{aa}^\flat (X_i')^{p-1}(\overline X_i')^{p}$. The cumulant expansion yields

\begin{align}\label{eq:expansion_X_ji}
\nonumber
& z\bb E|X_i'|^{2p}= -\E \sum_{a,b} \sigma_{ia}S_{ab} (G_{aa} G_{bb})^\flat (X_i')^{p-1}(\overline X_i')^{p} + \E \sum_{a,b} \sigma_{ia}S_{ab} G_{ba} \partial_{ba} (X_i')^{p-1}(\overline X_i')^{p} \\\nonumber
& = - \E \sum_{a,b} \sigma_{ia}S_{ab}[G_{aa}^\flat G_{bb}^\flat - \E G_{aa}^\flat G_{bb}^\flat +  G_{aa}^\flat \E G_{bb} + G_{bb}^\flat \E G_{aa}] (X_i')^{p-1}(\overline X_i')^{p} \\\nonumber
& \quad+ \E \sum_{a,b} \sigma_{ia}S_{ab} G_{ba} \partial_{ba} (X'_i)^{p-1}(\overline X'_i)^{p} \\\nonumber
& = - \mathfrak{m} \bb E|X_i'|^{2p} - \mathfrak{m}\bb E \sum_{a,b} \sigma_{ia}  S_{ab}   G_{bb}^\flat (X_i')^{p-1}(\overline X_i')^{p} + \E \sum_{a,b} \sigma_{ia}S_{ab} G_{ba} \partial_{ba} (X_i')^{p-1}(\overline X_i')^{p} \\
& \quad - \E \sum_{a,b} \sigma_{ia}S_{ab}[G_{aa}^\flat G_{bb}^\flat - \E G_{aa}^\flat G_{bb}^\flat + G_{aa}^\flat (\E G_{bb} - \mathfrak m) + G_{bb}^\flat (\E G_{aa}-\mathfrak m)] (X_i')^{p-1}(\overline X_i')^{p}. 
\end{align}
where we used that
$$
(fg)^\flat = f^\flat g^\flat - \E f\flat g^\flat + f^\flat \E g + g^\flat \E f\,.
$$

Let $\tilde X_i\deq\sum_a S_{ia} G_{aa}^\flat$. Expanding $\bb E|\tilde{ X_i}|^{2p}$ as we did for $\E|X_i'|^{2p}$ and using the properties of the matrix $\cal L$ defined in \eqref{eq:def_L}, we get
\begin{align}\nonumber
\bb E|\tilde{ X_i}|^{2p} = & \,\E \sum_{a,b} \tau_{ia}S_{ab} G_{ba} \partial_{ba} (X_i')^{p-1}(\overline X_i')^{p} \\
& \quad - \E \sum_{a,b} \tau_{ia}S_{ab}[G_{aa}^\flat G_{bb}^\flat - \E G_{aa}^\flat G_{bb}^\flat + G_{aa}^\flat (\E G_{bb} - \mathfrak m) + G_{bb}^\flat (\E G_{aa}-\mathfrak m)] (X_i')^{p-1}(\overline X_i')^{p}\,,\label{eq:tildeX}
\end{align}
where $\tau = \cal L S$. Thus, from \eqref{eq:tildeX} and \eqref{eq:expansion_X_ji} we get \begin{align}\label{eq:X_temp_bound}
\nonumber
\bb E|X_i'|^{2p} \prec & \,\,\bigg\vert  \E \sum_{a,b} \breve\sigma_{ia}S_{ab} G_{ba} \partial_{ba} (X_i')^{p-1}(\overline X_i')^{p} \bigg\vert \\\nonumber
& + \bigg\vert \E \sum_{a,b} \breve\sigma_{ia}S_{ab}[G_{aa}^\flat G_{bb}^\flat - \E G_{aa}^\flat G_{bb}^\flat + G_{aa}^\flat (\E G_{bb} - \mathfrak m) + G_{bb}^\flat (\E G_{aa}-\mathfrak m)] (X_i')^{p-1}(\overline X_i')^{p} \bigg\vert \\
\prec & \,\, \bigg\vert  \E \sum_{a,b} \breve\sigma_{ia}S_{ab} G_{ba} \partial_{ba} (X_i')^{p-1}(\overline {X_i'})^{p} \bigg\vert + \Psi^2 \E|X_i' |^{2p},
\end{align}
where in the last line we used $ \E(G_{ii} - \mathfrak m) \prec \Psi$ in \eqref{eq:corollary1} and $ G_{ii}^\flat \prec \Psi$. To finish the proof, we compute the derivative 
\begin{align*}
\partial_{ab} X'_i = - \sum_{c} \sigma_{ic} G_{ca} G_{bc} \prec \Psi^3+ \delta_{ab} \Psi^2
\end{align*}
where we used \eqref{eq:bound_Y}. Therefore, \eqref{eq:X_temp_bound} becomes
\begin{align*}
\bb E|\tilde{ X_i}|^{2p} \prec \Psi^2 \E | X_i'|^{2p-1} + \Psi^4 \E | X_i'|^{2p-2}.
\end{align*}
By invoking Lemma \ref{lemma:moment_to_domination}, we finally get $ X_i' \prec \Psi^2$. The proof of the second estimate in \eqref{eq:X_bound}, i.e.\ $\sum_k \sigma_{ik} G_{kj} \prec \Psi^2$, is completely analogous to the one just presented for $X_i \prec \Psi^2$ and it is omitted.

\section{Extension to non-Gaussian distribution and general complex case}\label{sec:nonGaussian_generalComplex}

In this section we explain how to deal with the more general case when $ H_{ij}$ is not necessarily Gaussian distributed and when $ \E H_{ij}^2 \neq 0$. 

Let us start with the non Gaussian corrections: we point out that the way we are going to control them is insensitive to whether $\E H_{ij}^2 = 0 $ or not, therefore, for the sake of simplicity, we will assume that $ \E H_{ij}^2 = 0$. 
We recall that the Gaussian case is easier because, adopting the notation of Lemma \ref{lem:cumulant_factos_estimate}, the cumulants of order $ p+q=k \geq 3$ vanish, which implies that the cumulant expansion formula \eqref{eq:cumulant_expansion} is truncated at $ \ell =1$ with $ R_{2} =0 $. When $ H_{ij}$ is not Gaussian we have to estimate all the higher order terms in \eqref{eq:cumulant_expansion}. 

We will show explicitly that the non Gaussian terms do not modify the bounds that we got earlier in \eqref{eq:P_vy}, \eqref{eq:bound_Y} and \eqref{eq:bound_Z} since they are at most of the same magnitude as the Gaussian terms. 
Heuristically, the reason is that, as stated in Lemma \ref{lem:cumulant_factos_estimate}, we have
\begin{align}\label{eq:bound_cumulants}
\cal C^{(p,q)}(H_{ij}) = O( S_{ij}^{k/2}) \prec S_{ij}\Psi^{k -2}, \ \ \ k = p+q\,.
\end{align}
Throughout this whole Section \ref{sec:nonGaussian_generalComplex} we will assume that $\Psi$ is an admissible control parameter such that the a prior bound $\Lambda^2 \prec \Psi^2$ holds true. 

Note that \eqref{eq:bound_cumulants} implies that in the higher order contributions of the full cumulant expansion \eqref{eq:cumulant_expansion} the huge number of resolvent entries produced by the derivatives is compensated by the smallness of the high order cumulants. 

In particular, the strategy is to expand these derivatives via the Leibniz rule (see \eqref{8.3}, \eqref{eq:nonGauss_Y_full_expansion} and \eqref{eq:U_n2} below): only few terms in this expansion need to be treated carefully (see for example \eqref{eq:ell=1}, \eqref{eq:ell=0}, \eqref{eq:a} and \eqref{eq:3_cumulant}), while all the others are easily bounded by using \eqref{eq:bound_cumulants}, the Ward identity and $\Lambda^2 \prec \Psi^2$. 
 
Moreover, for \eqref{eq:bound_Y} and \eqref{eq:bound_Z} we can still employ the same self-consistent scheme that we set up in Lemma \ref{lemma_self-consistent_equation_trick} in the non Gaussian case: the only modification is that the non Gaussian terms will appear just as additional contributions on the right hand side of \eqref{eq:self_consist_trick}. 

 The same method can be applied to prove that all the other results in Proposition \ref{prop:main_estimates} and lemmata \ref{lemma:chain_loop_X_estimates}, \ref{lemma:first_induction_step} and \ref{lemma:bound_B} remain valid when $H$ is non Gaussian. 
Finally, we will briefly discuss also how to control the additional terms arising in the cumulant expansion when $ \E H_{ij}^2 \neq 0$.

\subsection{Non Gaussian terms in \eqref{eq:P_vy}}
In this section we explain the proof of \eqref{eq:P_vy} when the entries of $H$ are not Gaussian. Suppose there is some  $\lambda\in [L^{-1},L]$ such that $\cal Q_{xy} \prec \lambda$ for all $x,y \in \bb T_L$. As in Section \ref{sec:proof_main_estimates}, we would like to have a bound on $\bb E|Q_{yy}|^{2p}$. By looking at the proof in Section \ref{sec:proof_main_estimates} carefully, we see that we used two cumulant expansions in the proof: the first one is in \eqref{eq:Q^2p} and the second one is in \eqref{eq:rhs_Q}. We now want to control the additional non Gaussian terms arising from them. 

  Let us look at the estimate of \eqref{eq:rhs_Q}. Starting from the LHS of \eqref{eq:rhs_Q}, we need to consider the additional terms
\begin{equation} \label{8.1}
\begin{aligned}
&\Big(\frac{1}{\sqrt{L}}\Big)^n\sum_{\substack{w,t \geq 0 \\w+t=2}}^K \frac{1}{w! t!}\cal C^{(w,t+1)}(H_{j_{n+1}i_{n+1}}) \bb E \sum_{i_1,j_1,...,i_{n+1},j_{n+1}} \sigma^{(1)}_{i_1j_1}\cdots \sigma^{(n)}_{i_nj_n} \\
&\partial_{j_{n+1}i_{n+1}}^w \partial_{i_{n+1}j_{n+1}}^t \big(G_{j_1y}\ol G_{i_1y}V_{1,2}\cdots V_{n-1,n} S_{\b vi_{n+1}}G_{j_{n+1}j_n}\ol G_{i_{n+1}y}G_{i_ny}\cdot \cal Q_{yy}^{\alpha}\ol{\cal Q}_{yy}^{\beta}\big)+\sum_{j_{n+1},i_{n+1}}R^{j_{n+1}i_{n+1}}_{K+1}\,,
\end{aligned}
\end{equation}
where $R^{j_{n+1}i_{n+1}}_{K+1}$ is the remainder term defined similarly to $R_{K+1}$ in \eqref{eq:R_cumulant_expansion}. Following a routine verification (one may refer to the proof of Lemma 4.6(i) in \cite{HK17}), we see that for any $D>0$, there is $K=K(D) \in \bb N$ such that 
$$
\sum_{j_{n+1},i_{n+1}}R^{j_{n+1}i_{n+1}}_{K+1}=O(L^{-D})\,.
$$
Now we are left with the estimate of the first sum in \eqref{8.1}. Let us fix $(w,t)\in \bb N\times \bb N$ with $w+t=k \ge 2$. W.L.O.G we assume $w=k$ and $t=0$, and the general cases of $w,t$ follow in a similar fashion. Let us define
$$
 \sigma^{(n+1)}_{j_{n+1}i_{n+1}}=\sqrt{L}W^{(k-1)/2}S_{\b v i_{n+1}}\cal C^{(k,1)}(H_{j_{n+1}i_{n+1}})\,,
$$
which belongs to $\bb S$ thanks to \eqref{eq:bound_cumulants}.  

We would like to bound 
\begin{equation} \label{8.2}
\frac{1}{L^{(n+1)/2}W^{(k-1)/2}}\bb E \sum_{i_1,j_1,...,i_{n+1},j_{n+1}}  \sigma^{(1)}_{i_1j_1}\cdots \sigma^{(n+1)}_{i_{n+1}j_{n+1}}\partial_{j_{n+1}i_{n+1}}^k \big(G_{j_1y}\ol G_{i_1y}V_{1,2}\cdots V_{n-1,n}
G_{j_{n+1}j_{n}}\ol G_{i_{n+1}y}G_{i_{n}y}\cdot \cal Q_{yy}^{\alpha}\ol{\cal Q}_{yy}^{\beta}\big)\,.
\end{equation} 
By Leibniz's rule, we look at \eqref{8.2} by considering 
\begin{equation} \label{8.3}
\begin{aligned}
\frac{1}{L^{(n+1)/2}W^{(k-1)/2}}\bb E \sum_{i_1,j_1,...,i_{n+1},j_{n+1}} \sigma^{(1)}_{i_1j_1}\cdots \sigma^{(n+1)}_{i_{n+1}j_{n+1}}\big[\partial_{j_{n+1}i_{n+1}}^s \big(G_{j_1y}\ol G_{i_1y}V_{1,2}\cdots V_{n-1,n}
G_{j_{n+1}j_n}\ol G_{i_{n+1}y}G_{i_ny}\big)\big]
\\
\cdot  \big(\partial^{t_1}_{j_{n+1}i_{n+1}}\cal Q_{yy} \cdots \partial^{t_\ell}_{j_{n+1}i_{n+1}}\cal Q_{yy}\big) \cdot Q_{yy}^{\alpha-\ell}\ol{\cal Q}_{yy}^{\beta}\,,
\end{aligned}
\end{equation}
where $s+t_1+\cdots+t_{\ell}=k \ge 2$. Here we consider the case where the differential $\partial_{j_{n+1}i_{n+1}}$ is only applied to $Q_{yy}$, and the general case when $\partial_{j_{n+1}i_{n+1}}$ is also applied to $\ol{Q}_{yy}$ can be estimated in the similar fashion. Thus a bound on \eqref{8.3} implies the same bound on \eqref{8.2}.

{\bf Case 1.} Suppose $\ell \ge 2$. Let $\tilde{\cal T}$ be the collection of factors $G$ and $\ol{G}$ in \eqref{8.3} such that at least one of the two indices belongs to the class $\{i_1,j_1,...,i_n,j_n\}$, and we have $|\tilde{ \cal T}|=3n+1$. As in Section \ref{prop:main_estimates}, we can use Lemma \ref{lemma:chain_loop_X_estimates}(i) to show 
\begin{equation} \label{8.5}
\Big(\frac{1}{\sqrt{L}}\Big)^n\sum_{i_1,j_1,...,i_n,j_n}\sigma^{(1)}_{i_1j_1}\sigma^{(n)}_{i_nj_n} \prod_{t \in \tilde{\cal  T}}  t\prec \Big(\frac{\Psi^3}{\sqrt{\eta}}+\frac{\Psi}{\sqrt{L\eta}}\Big)^n \cdot \Psi\,.
\end{equation}
By 
\begin{equation} \label{85}
\partial_{ji} \cal Q_{yy} =-\cal Q_{jy}G_{iy}-\cal Q_{yi}\ol G_{jy}+S_{\b vj}G_{iy}\ol G_{jy}-\sum_{k,l}S_{\b vk}S_{kl}G_{lj}G_{il}|G_{ky}|^2-\sum_{k,l}S_{\b v k}S_{kl}\ol G_{ki}\ol G_{jk}|G_{ly}|^2\,,
\end{equation}
and
\begin{equation*}
\partial_{ji} \ol{\cal Q}_{yy} =-\ol{\cal Q}_{iy}\ol G_{jy}-\ol{\cal Q}_{yj} G_{iy}+\ol S_{\b vj}G_{iy}\ol G_{jy}-\sum_{k,l}\ol S_{\b vl}S_{lk} G_{lj} G_{il}|G_{ky}|^2-\sum_{k,l}\ol S_{\b vl}S_{lk} \ol G_{ki}\ol G_{jk}|G_{ly}|^2\,,
\end{equation*}
together with Ward identity, we see that for any fixed $m \ge 1$,
\begin{equation*} 
\partial^{m}_{j_{n+1}i_{n+1}} \cal Q_{yy} \prec \lambda + \frac{1}{\sqrt{L}}+\frac{\Psi^3}{\sqrt{\eta}}\,.
\end{equation*}
By using the above bound for $\ell-2$ many factors in \eqref{8.3}, together with \eqref{8.5} we have for some $p_{n+1} \in \{i_{n+1},j_{n+1}\}$ with
\begin{equation} \label{8.9}
\begin{aligned}
\eqref{8.3} \prec&\, \frac{1}{L^{1/2}W^{(k-1)/2}}\cdot \Big(\frac{\Psi^3}{\sqrt{\eta}}+\frac{\Psi}{\sqrt{L\eta}}\Big)^n \cdot \Psi \cdot \Big(\lambda + \frac{1}{\sqrt{L}}+\frac{\Psi^3}{\sqrt{\eta}}\Big)^{\ell-2}\\
&\,\cdot \bb E \big|\sum_{i_{n+1},j_{n+1}}  \sigma^{(n+1)}_{i_{n+1}j_{n+1}} \ol G_{p_{n+1}y} \partial^{t_1}_{j_{n+1}i_{n+1}}\cal Q_{yy}\partial^{t_2}_{j_{n+1}i_{n+1}}\cal Q_{yy}\big| \cdot |\cal Q_{yy}|^{\alpha+\beta-\ell}\\
\prec&\, \frac{1}{L^{1/2}W^{(k-\ell+1)/2}} \cdot \Big(\frac{\Psi^3}{\sqrt{\eta}}+\frac{\Psi}{\sqrt{L\eta}}\Big)^n \cdot \Psi \cdot \Big(\lambda\Psi + \frac{\Psi}{\sqrt{L}}+\frac{\Psi^3}{\sqrt{\eta}}\Big)^{\ell-2}\\
&\, \cdot \bb E \big|\sum_{i_{n+1},j_{n+1}}  \sigma^{(n+1)}_{i_{n+1}j_{n+1}} \ol G_{p_{n+1}y} \partial^{t_1}_{j_{n+1}i_{n+1}}\cal Q_{yy}\partial^{t_2}_{j_{n+1}i_{n+1}}\cal Q_{yy}\big| \cdot |\cal Q_{yy}|^{\alpha+\beta-\ell}\,.
\end{aligned}
\end{equation}
For $\max\{t_1,t_2\} \ge 2$, we have $k-\ell+1 \ge 2$, thus by $1/\sqrt{W} \prec \Psi$ and using Ward identity for the term $\ol G_{p_{n+1}y}$ and another $G$ hidden in $\partial^{t_1}_{j_{n+1}i_{n+1}}\cal Q_{yy}$ we have
\begin{equation*}
\begin{aligned}
\eqref{8.9} \prec&\, \frac{\Psi^2}{\sqrt{L}} \cdot \Big(\frac{\Psi^3}{\sqrt{\eta}}+\frac{\Psi}{\sqrt{L\eta}}\Big)^n \cdot \Psi \cdot \Big(\lambda\Psi + \frac{\Psi}{\sqrt{L}}+\frac{\Psi^3}{\sqrt{\eta}}\Big)^{\ell-2} \cdot \frac{1}{\eta} \Big(\lambda + \frac{1}{\sqrt{L}}+\frac{\Psi^3}{\sqrt{\eta}}\Big)^2 \bb E |\cal Q_{yy}|^{2p-n-\ell-1}\\
\prec&\,\Big(\frac{\Psi^3}{\sqrt{\eta}}+\frac{\Psi}{\sqrt{L\eta}}+\lambda\Psi\Big)^{n+\ell-2}\cdot \bigg(\frac{\Psi}{\eta^{1/6}}\lambda^{2/3}+\frac{\Psi^3}{\sqrt{\eta}}+\frac{\Psi}{\sqrt{L\eta}}\bigg)^3 \bb E |\cal Q_{yy}|^{2p-n-\ell-1}\\
\prec & \Big(\frac{\Psi^3}{\sqrt{\eta}}+\frac{\Psi}{\sqrt{L\eta}}+\lambda\Psi+\frac{\Psi}{\eta^{1/6}}\lambda^{2/3}\Big)^{n+\ell+1}\bb E |\cal Q_{yy}|^{2p-n-\ell-1}\,.
\end{aligned}
\end{equation*}
For $t_1=t_2=1$, we have $k-\ell+1 \ge 1$, thus by exploring \eqref{85} carefully and use Ward identity we have
\begin{equation*}
\begin{aligned}
\eqref{8.9} \prec&\, \frac{\Psi}{\sqrt{L}} \cdot \Big(\frac{\Psi^3}{\sqrt{\eta}}+\frac{\Psi}{\sqrt{L\eta}}\Big)^n \cdot \Psi \cdot \Big(\lambda\Psi + \frac{\Psi}{\sqrt{L}}+\frac{\Psi^3}{\sqrt{\eta}}\Big)^{\ell-2} \cdot  \Big(\frac{1}{\eta}\lambda^2 + \frac{\Psi}{\eta} \frac{1}{L}+\frac{1}{\eta}\frac{\Psi^6}{\eta}\Big)\cdot \bb E |\cal Q_{yy}|^{2p-n-\ell-1}\\
\prec & \Big(\frac{\Psi^3}{\sqrt{\eta}}+\frac{\Psi}{\sqrt{L\eta}}+\lambda\Psi+\frac{\Psi}{\eta^{1/6}}\lambda^{2/3}\Big)^{n+\ell+1}\bb E |\cal Q_{yy}|^{2p-n-\ell-1}\,.
\end{aligned}
\end{equation*}
Thus for Case 1 we have
\begin{equation*}
\eqref{8.3} \prec \sum_{\ell=1}^{2p-n-1} \Big(\frac{\Psi^3}{\sqrt{\eta}}+\frac{\Psi}{\sqrt{L\eta}}+\lambda\Psi+\frac{\Psi}{\eta^{1/6}}\lambda^{2/3}\Big)^{n+\ell+1}\bb E |\cal Q_{yy}|^{2p-n-\ell-1}\,.
\end{equation*}
{\bf Case 2.} Suppose $\ell=1$. We have $1/W^{(k-1)/2} \prec \Psi$. By \eqref{8.5} and Ward identity we have
\begin{equation}\label{eq:ell=1}
\begin{aligned}
\eqref{8.3} \prec&\,  \frac{\Psi}{\sqrt{L}}\Big(\frac{\Psi^3}{\sqrt{\eta}}+\frac{\Psi}{\sqrt{L\eta}}\Big)^n \cdot \Psi\cdot \bb E \big|\sum_{i_{n+1},j_{n+1}} \sigma^{(n+1)}_{i_{n+1}j_{n+1}} \ol G_{i_{n+1}y} \partial^{t}_{j_{n+1}i_{n+1}}\cal Q_{yy}\big| \cdot |\cal Q_{yy}|^{\alpha+\beta-1}\\
\prec&\, \Big(\frac{\Psi^3}{\sqrt{\eta}}+\frac{\Psi}{\sqrt{L\eta}}\Big)^n \cdot \frac{\Psi^2}{\sqrt{L}} \cdot \frac{1}{\eta}\Big(\lambda+\frac{1}{\sqrt{L}}+\frac{\Psi^3}{\sqrt{\eta}}\Big)\cdot \bb E |\cal Q_{yy}|^{2p-n-2}\\
\prec & \Big(\frac{\Psi^3}{\sqrt{\eta}}+\frac{\Psi}{\sqrt{L\eta}}+\frac{\Psi}{L^{1/4}\eta^{1/2}}\lambda^{1/2}\Big)^{n+2} \cdot \bb E |\cal Q_{yy}|^{2p-n-2}\,.
\end{aligned}
\end{equation}
{\bf Case 3.} Suppose $\ell=0$, and we see that similar as in Section 5.1, we can use Lemma \ref{lemma:chain_loop_X_estimates} and shown that
\begin{equation} \label{eq:ell=0}
\begin{aligned}
\frac{1}{L^{(n+1)/2}W^{(k-1)/2}} \sum_{i_1,j_1,...,i_{n+1},j_{n+1}} \sigma^{(1)}_{i_1j_1}\cdots \sigma^{(n+1)}_{i_{n+1}j_{n+1}}\partial_{j_{n+1}i_{n+1}}^k \big(G_{j_1y}\ol G_{i_1y}V_{1,2}\cdots V_{n-1,n}
G_{j_{n+1}j_n}\ol G_{i_{n+1}y}G_{i_ny}\big)\\
\prec \Big(\frac{\Psi^3}{\sqrt{\eta}}+\frac{\Psi}{\sqrt{L\eta}}\Big)^{n+1}\,,
\end{aligned}
\end{equation}
which implies
\begin{equation*}
\eqref{8.3} \prec \Big(\frac{\Psi^3}{\sqrt{\eta}}+\frac{\Psi}{\sqrt{L\eta}}\Big)^{n+1} \bb E |\cal Q_{yy}|^{2p-n-1}\,.
\end{equation*}
Thus from Cases 1-3 we have
\begin{equation*}
\begin{aligned}
\eqref{8.3} \prec&\, \sum_{k=1}^{2p-n} \Big(\frac{\Psi^3}{\sqrt{\eta}}+\frac{\Psi}{\sqrt{L\eta}}+\lambda\Psi+\frac{\Psi}{\eta^{1/6}}\lambda^{2/3}+\frac{\Psi}{L^{1/4}\eta^{1/2}}\lambda^{1/2}\Big)^{n+k}\bb E |\cal Q_{yy}|^{2p-n-k}\\
\prec &\, \sum_{k=1}^{2p} \Big(\frac{\Psi^3}{\sqrt{\eta}}+\frac{\Psi}{\sqrt{L\eta}}\Big)^{k/3}\lambda^{2k/3}\bb E |\cal Q_{yy}|^{2p-k} \,.
\end{aligned}
\end{equation*}
By estimating \eqref{eq:Q^2p} in a similar fashion and using the steps in Section \ref{sec:proof_main_estimates}, we see that 
\begin{equation*}
	\bb E |\cal Q_{yy}|^{2p} \prec  \sum_{k=1}^{2p} \Big(\frac{\Psi^3}{\sqrt{\eta}}+\frac{\Psi}{\sqrt{L\eta}}\Big)^{k/3}\lambda^{2k/3}\bb E |\cal Q_{yy}|^{2p-k}\,.
\end{equation*}
Note that we can also estimate $\bb E |\cal Q_{xy}|^{2p}$ exactly in the same way, i.e.\
$$
\bb E |\cal Q_{xy}|^{2p} \prec  \sum_{k=1}^{2p} \Big(\frac{\Psi^3}{\sqrt{\eta}}+\frac{\Psi}{\sqrt{L\eta}}\Big)^{k/3}\lambda^{2k/3}\bb E |\cal Q_{xy}|^{2p-k}\,
$$ 
whenever $\cal Q_{xy} \prec \lambda$ for all $x,y \in \bb T_L$. Hence, it suffices to apply Lemma \ref{lemma:moment_to_domination} to conclude the proof.

\subsection{Non Gaussian terms in \eqref{eq:bound_Y}}\label{sec:nonGaussian_Y}

In this section we want to estimate the non Gaussian terms in the cumulant expansion \eqref{eq:expansion_callY_pq}: using Lemma \ref{lem:cumulant_expansion}, we see that the non Gaussian terms yields the additional contributions $\cal U_{uab} + |R_{K+1}| $ to the right hand side of \eqref{eq:expansion_callY_pq} where 
\begin{align}\label{eq:nonGauss_U}
\cal U_{uab} = \bigg\vert \E  \sum_{i,j} \sigma_{ui} \sum_{\substack{w,t \geq 0 \\ w + t = 2}}^{K} \frac{1}{w! t!} \cal C^{(w, t+1)}(H_{ji})\partial_{ij}^t \partial_{ji}^w [ G_{ai} G_{jb} (Y_{uab})^{p-1} (\overline{Y}_{uab})^p ]\bigg\vert, \ \ \ \sigma \in \bb S
\end{align}
and $ R_{K+1}$ is given by \eqref{eq:R_cumulant_expansion}. As in the previous section, from Lemma 3.4 (iii) in \cite{HKR17}, we know that, given a large constant $D$, then one can choose $K = K(D)$ such that $ R_{K+1} = O(L^{-D})$. 

Let us now focus on $\cal U_{uab}$: 
%By using the trivial estimate $ Y_{uab} \prec \Psi^2 $ and $ G_{xy} \prec \Psi + \delta_{xy}$, a simple induction shows that
%\begin{align}\label{eq:multiple_derivative_Y}
%& \partial^{t}_{ij} \partial^{w}_{ji} Y_{iab} \prec \Psi^3 + (\delta_{ai} + \delta_{bi} + \delta_{aj} + \delta_{bj} + \delta_{ij}) \Psi^2 \\\nonumber
%& \partial^{t}_{ij} \partial^{w}_{ji} \overline Y_{iab} \prec \Psi^3 + (\delta_{ai} + \delta_{bi} + \delta_{aj} + \delta_{bj} + \delta_{ij}) \Psi^2
%\end{align}
%for any $ t, w$ such that $ t + w \geq 1$. 
since we are going to use very rough estimates where the complex conjugation does not play any role, in the following we will simplify the notation by neglecting it: i.e.\ we will replace $\ol Y_{uab}$ by $ Y_{uab}$ and also $\partial_{ji} = \partial/\partial H_{ji} = \partial/\partial \ol H_{ji} $ by $ \partial_{ij} = \partial/\partial H_{ij}$. 

Using these notation conventions and \eqref{eq:bound_cumulants}, we can write
%Furthermore, the effect of differentiation with respect to $ H_{ck}$ or $ H_{kc}$ is not going to affect the velidity of the following (very rough) estimates, so we will simplify even more the notation by the following substitution
%\begin{align}
%\sum_{u + t = 3}^{K} \frac{1}{u! t!} \cal C^{(u, t+1)}(H_{ck})\frac{\partial^{u+t}}{\partial H^u_{ck} \partial H^t_{kc}} \to \sum_{n = 3}^{K} \frac{1}{n!} \cal C^{(n+1)}(H_{ck})\frac{\partial^{n}}{\partial H^n_{ck}},
%\end{align}
%therefore
\begin{align}
\cal U_{uab} \prec \sum_{n = 2}^{K} U^{(n)}_{uab}, \ \ \ U^{(n)}_{uab} := \E  \sum_{i,j} |\sigma_{ui}|  S_{ij}^{(n+1)/2}  |\partial^{n}_{ij}  G_{ai} G_{jb} Y_{uab}^{2p-1}|\,.
\end{align}
Applying the Leibniz rule for derivatives we have
\begin{align*}
\partial^{n}_{ij}  G_{ai} G_{jb} Y_{uab}^{2p-1} = \sum_{r=0}^n (\partial^{r}_{ij}  G_{ai} G_{jb} ) (\partial^{n-r}_{ij}  Y_{uab}^{2p-1} ),
\end{align*}
thus a second application of the Leibniz rule allows us to estimate $ U^{(n)}_{uab}$ as a sum of contributions of the form
\begin{align}\label{eq:nonGauss_Y_full_expansion}
 \E  \sum_{i,j} |\sigma_{ui}| S_{ij}^{(n+1)/2}  \bigg\vert (\partial^{r}_{ij}  G_{ai} G_{jb} )  \bigg( \prod_{t=1}^{h} \partial^{\ell_t}_{ij} Y_{uab} \bigg)Y_{uab}^{2p-1-h} \bigg\vert
\end{align}
where the sum runs over the integers $ h = 0, \ldots, (n-r) \wedge (2p-1)$ and $ \ell_1, \ldots, \ell_h \geq 1$ with $ \ell_1 + \cdots + \ell_h = n-r$.
Let us split $ U^{(n)}$ in three terms: 
\begin{itemize}
\item[(a)] the one corresponding to $r = 0$ and $ h = n$ (so that $\ell_1 = \ell_2 = \cdots = \ell_n = 1$) is stochastically dominated by
\begin{align}\label{eq:a}
(\mathrm a) := \,\E  \sum_{i,j} S_{ij}^{(n+1)/2}  |\sigma_{ui}| |G_{ai} G_{jb} | 
|\partial_{ij}  Y_{uab} |^n |Y_{uab}|^{2p-n-1},
\end{align}
\item[(b)] the one where $ r = 0 $ and $ h \leq (n-1) \wedge (2p-1)$ is dominated by
\begin{align*}
(\mathrm b) := 
\E  \sum_{i,j}S_{ij}^{(n+1)/2} |\sigma_{ui} G_{ai} G_{jb}| \sum_{h=0}^{(n-1) \wedge (2p-1)}
\bigg \vert \prod_{t=1}^{h} 
\partial^{\ell_t}_{ij}Y_{uab}  \bigg\vert |Y_{uab}|^{2p-1-h},
\end{align*}
\item[(c)] the one where $ r \geq 1$ is bounded by
\begin{align*}
(\mathrm c) :=  \E  \sum_{i,j} S_{ij}^{(n+1)/2}|\sigma_{ui}| \sum_{r=1}^{n} |\partial^{r}_{ij}   G_{ai} G_{jb}| \sum_{h=0}^{(n-r) \wedge (2p-1)}\bigg \vert \prod_{t=1}^{h} 
\partial^{\ell_t}_{ij}Y_{uab}  \bigg\vert |Y_{uab}|^{2p-1-h}.
\end{align*}
\end{itemize}
Note that, since $ K$ in \eqref{eq:nonGauss_U} is big but fixed, we have that 
\begin{align}\label{eq:bound_Un}
\cal U_{uab} \prec (\mathrm a)+(\mathrm b)+(\mathrm c).
\end{align}
We are now going to establish more explicit bounds for (a), (b) and (c). For (a), set the prior estimate $ Y_{uab} \prec \lambda$ with $\lambda \in [\Psi^3, L^C]$ for $ a \neq b$. By using \eqref{eq:self_similar_derivative_Y} and the trivial bound $ Y_{uab} \prec \Psi^2$, we get
\begin{align}\label{eq:first_diff_Y}
\partial_{ij} Y_{uab} & \prec \lambda\Psi + \lambda(\delta_{ai} + \delta_{bi} +\delta_{aj} +\delta_{bj}) + \Psi^2(\delta_{ai}\delta_{bj} + \delta_{aj}\delta_{bi}).
\end{align}
 Therefore, from \eqref{eq:first_diff_Y} and \eqref{eq:a} we get that
\begin{align}\label{eq:u_{rs}}
(\mathrm a) & \prec  \, (\Psi^{4n+1} + \lambda^n \Psi^{n+2} + \Psi^{3n+3}) \E|Y_{uab}|^{2p-n-1} \\\nonumber
& \prec \, (\Psi^{3} \lambda^2)^{\frac{n+1}{3}} \E|Y_{uab}|^{2p-(n+1)},
\end{align}
where in the last line we used that $ n \geq 2$.

Let us now deal with (b) and (c): a simple induction shows that for any $ r \geq 0$ and $r' \geq 1$ and $ \tau, \omega \in \bb S$
\begin{align}\label{eq:diff_GG}
& \sum_{i,j} \tau_{ui} \omega_{ij} |\partial^{r}_{ij} G_{ai} G_{jb} | \prec \Psi^2, \\
\label{eq:diff_GG_deltas}
& \sum_{i,j} \tau_{ui} \omega_{ij} |\partial^{r}_{ij} G_{ai} G_{jb} |(\delta_{ai} + \delta_{bi} + \delta_{aj} + \delta_{bj} + \delta_{ij}) \prec \Psi^3, \\
\label{eq:multiple_derivative_Y}
& \partial^{r'}_{ij} Y_{uab} \prec \Psi^3 + (\delta_{ai} + \delta_{bi} + \delta_{aj} + \delta_{bj} + \delta_{ij}) \Psi^2\,. 
\end{align}
By \eqref{eq:multiple_derivative_Y}, \eqref{eq:diff_GG}, \eqref{eq:diff_GG_deltas} and \eqref{eq:u_{rs}} we get
\begin{align}
\nonumber
(\mathrm b)  & \prec   
\E  \sum_{i,j}S_{ij}^{(n+1)/2} |\sigma_{ui} G_{ai} G_{jb}| \sum_{h=0}^{(n-1) \wedge (2p-1)}
(\Psi^{3h} + (\delta_{ai} + \delta_{bi} + \delta_{aj} + \delta_{bj} + \delta_{ij}) \Psi^{2h}) |Y_{uab}|^{2p-1-h} \\\nonumber
& \prec \Psi^{n+1} \sum_{h=0}^{2p-1} \Psi^{3h} \E |Y_{uab}|^{2p-1-h} + \Psi^{n+2} \sum_{h=0}^{(n-1) \wedge (2p-1)} \Psi^{2h} \E |Y_{uab}|^{2p-1-h} \\
& \prec \Psi^{n-2} \sum_{h'=1}^{2p} \Psi^{3h'} \E |Y_{uab}|^{2p-h'} + \Psi^{n} \sum_{h'=1}^{n \wedge 2p} \Psi^{2h'} \E |Y_{uab}|^{2p-h'} \prec \sum_{h'=1}^{2p} \Psi^{3h'} \E |Y_{uab}|^{2p-h'},\label{eq:nonGaussian_Y}
\end{align}
where in the last passage we used that $ n \geq 2$, $n \geq h'$ and $ \Psi \leq 1$. The same argument shows that for (c) the same bound holds. Therefore, recalling \eqref{eq:bound_Un} and the definition of $ \lambda$, we have
\begin{align*}
\cal U_{uab} \prec  (\Psi^{3} \lambda^2)^{\frac{n+1}{3}} \E|Y_{uab}|^{2p-(n+1)} + \sum_{l=1}^{2p} \Psi^{3l} \E |Y_{uab}|^{2p-l} \prec \sum_{l=1}^{2p} (\Psi^{3} \lambda^2)^{l/3} \E |Y_{uab}|^{2p-l}.
\end{align*}
This implies that in the general case the bound \eqref{eq:high_moment_callY} becomes
\begin{align*}
\E|Y_{uab}|^{2p} \prec \lambda \Psi^3 \E |Y_{uab}|^{2p-2} + \sum_{l=1}^{2p} (\Psi^{3} \lambda^2)^{l/3} \E |Y_{uab}|^{2p-l} \prec  \sum_{l=1}^{2p} (\Psi^{3} \lambda^2)^{l/3} \E |Y_{uab}|^{2p-l}.
\end{align*}
Applying Lemma \ref{lemma:moment_to_domination} with $ q=2/3$, $ \vartheta = \lambda$ and $ \varphi = \Psi^3$ concludes the proof for $a \neq b$. For $a = b$ the bound $ Y_{uaa} \prec \Psi^2$ is trivial.

\subsection{Non Gaussian terms in \eqref{eq:bound_Z}}
First, we note that the same method used to treat the non Gaussian terms in $Y_{uab}$ can be employed to show that the bounds \eqref{eq:A_bound} for $A_{abcd}$ and \eqref{eq:B_basic_bound} for $B_{ab}$ remain valid in the non Gaussian case. 

Here we will focus on the additional terms arising from the expansion \eqref{eq:expansion_Z_pq}:
\begin{align}\label{eq:nonGauss_Z_original}
\cal U_{ab} = \bigg\vert \E \sum_{i,j,k}\sigma_{ai}\tau_{bj} \sum_{w + t = 2}^{K} \frac{1}{w! t!} \cal C^{(w, t+1)}(H_{ik})\partial_{ik}^t \partial_{ki}^w(G_{kj}G_{ji} Z_{ab}^{p-1} \overline Z_{ab}^{p} ) \bigg\vert\,, \ \ \ \sigma, \tau \in \bb S\,.
\end{align}

By using the same notation simplification adopted for the non Gaussian terms of \eqref{eq:bound_Y} in the previous Section \ref{sec:nonGaussian_Y}, we have
\begin{align*}
\cal U_{ab} \prec \sum_{n = 2}^{K} U^{(n)}_{ab}, \ \ \ U^{(n)}_{ab} := \E  \sum_{i,j,k} |\sigma_{ai}\tau_{bj}|  S_{ki}^{(n+1)/2}  |\partial^{n}_{ki}  G_{kj}G_{ji} Z_{ab}^{2p-1}|\,.
\end{align*}
As before, a double application of the Leibniz rule implies that $ U^{(n)}_{ab} $ is bounded by a sum of terms of the form
\begin{align}\label{eq:U_n2}
%U^{(n)}_{ab} 
%& := \E \sum_{ijk}\ee^{-\ii p j} S_{ai}S_{bj} \frac{1}{n!} \cal C^{(n)}(H_{ik})\frac{\partial^{n}}{\partial H^n_{ki}} G_{kj}G_{ji} (Z_{cd}^{(q)})^{2p-1} \\\nonumber
%& = 
\E \sum_{i,j,k} |\sigma_{ai}\tau_{bj}|  S_{ki}^{(n+1)/2} \bigg\vert (\partial^{r}_{ki} G_{kj}G_{ji}) \bigg( \prod_{t=1}^{h} \partial^{\ell_t}_{ki}Z_{ab}\bigg) Z_{ab}^{2p-1-h} \bigg\vert
\end{align}
where the sum runs over the integers $r=0, \ldots, n$, $ h = 0, \ldots, (n-r) \wedge (2p-1)$ and $ \ell_1, \ldots, \ell_h \geq 1$ with $ \ell_1 + \cdots + \ell_h = n-r$.
An induction argument involving \eqref{eq:bound_Y} yields for $ \ell \geq 1$, $ r \geq 0$ and $\sigma, \tau, \omega \in \bb S$ that
\begin{align}\label{eq:GG}
& \E \sum_{i,j,k} \sigma_{ai}\tau_{bj}\omega_{ki}| \partial^r_{ki} G_{kj}G_{ji}| \prec \Psi^2, \\
\label{eq:diff_Z} 
& \partial^{\ell}_{ki} Z_{ab} \prec \Psi^4.
\end{align}
Thus, by using \eqref{eq:U_n2}, \eqref{eq:GG} and \eqref{eq:diff_Z}, we get
\begin{align}\label{eq:nonGaussian_Z}
\cal U_{ab} & \prec \Psi^{n+1} \sum_{h = 0}^{2p-1} \Psi^{4h}\E |Z_{ab}|^{2p-(h+1)} \prec \Psi^{n-3} \sum_{h' = 0}^{2p-1} \Psi^{4h'}\E |Z_{ab}|^{2p-h')} \prec \sum_{h' = 1}^{2p} \Psi^{4h'}\E |Z_{ab}|^{2p-h'} 
\end{align}
where in the last passage we used that $ n \geq 3$ and $ \Psi \leq 1$. 
To complete the proof we examine in more detail what happens when $ n=2$ which corresponds to the third order cumulant. We come back to the original expression \eqref{eq:nonGauss_Z_original} and we observe that the terms where $ w + t = 2$ are stochastically dominated by
\begin{align}\label{eq:3_cumulant}\nonumber
& \bigg\vert \E \sum_{i,j,k}\sigma_{ai}\tau_{bj} \sum_{w + t = 2} \frac{1}{w! t!} \cal C^{(w, t+1)}(H_{ik})( \partial^t_{ik} \partial^w_{ki} G_{kj}G_{ji} ) Z_{ab}^{p-1} \overline Z_{ab}^{p} \bigg\vert \\\nonumber
& + \bigg\vert \E  \sum_{i,j,k}\sigma_{ai}\tau_{bj} \sum_{w + t = 2} \frac{1}{w! t!} \cal C^{(w, t+1)}(H_{ik}) G_{kj}G_{ji} (\partial^w_{ki} \partial^t_{ik}Z_{ab}^{p-1} \overline Z_{ab}^{p}) \bigg\vert \\\nonumber
& +\bigg\vert \E \sum_{i,j,k}\sigma_{ai}\tau_{bj}  \cal C^{(1, 2)}(H_{ik}) (\partial_{ki} G_{kj}G_{ji}) (\partial_{ik}Z_{ab}^{p-1} \overline Z_{ab}^{p}) \bigg\vert  \\\nonumber
& +\bigg\vert \E  \sum_{i,j,k}\sigma_{ai}\tau_{bj}  \cal C^{(1, 2)}(H_{ik})(\partial_{ik} G_{kj}G_{ji}) (\partial_{ki}Z_{ab}^{p-1} \overline Z_{ab}^{p}) \bigg\vert  \\\nonumber
& +\bigg\vert \E  \sum_{i,j,k}\sigma_{ai}\tau_{bj}  \cal C^{(0, 3)}(H_{ik}) (\partial_{ik} G_{kj}G_{ji}) (\partial_{ik}Z_{ab}^{p-1} \overline Z_{ab}^{p}) \bigg\vert  \\
& +\bigg\vert \E  \sum_{i,j,k}\sigma_{ai}\tau_{bj}  \cal C^{(2, 1)}(H_{ik}) (\partial_{ki} G_{kj}G_{ji}) (\partial_{ki}Z_{ab}^{p-1} \overline Z_{ab}^{p})  \bigg\vert  =  (\mathrm a) + (\mathrm b) + (\mathrm c),
\end{align}
where $(\mathrm c)$ denotes the sum of the last four terms on the left hand side of \eqref{eq:3_cumulant}.
Let us proceed term by term.
\begin{itemize}
\item[(a)] We note that $ \partial^t_{ik} \partial^w_{ki} G_{kj}G_{ji}$ can only generate of the form $ G_{ii}G_{kk}G_{kj}G_{ji}$ and $ G_{ki}G_{kk}G_{ij}G_{ji}$, up to switching $ i$ to $k$. Therefore, plugging the first type of contribution in $(\mathrm a)$ and using \eqref{eq:bound_cumulants} and \eqref{eq:bound_Y} for the summation over $j$, we get
\begin{align*}
& \bigg\vert \E  \sum_{i,j,k}\sigma_{ai}\tau_{bj} \sum_{w + t = 2} \frac{1}{w! t!} \cal C^{(w, t+1)}(H_{ik})G_{ii}G_{kk}G_{kj}G_{ji} Z_{ab}^{p-1} \overline Z_{ab}^{p} \bigg\vert \\\nonumber
& \prec \Psi \E \sum_{i,k}|\sigma_{ai}S_{ik}| (\Psi^3 + \delta_{ik}\Psi^2) |Z_{ab}|^{2p-1}  \prec \Psi^4 \E |Z_{ab}|^{2p-1}.
\end{align*}
The second type of contribution has three non-diagonal entries of $ G$, so we trivially get the same bound as for the former term. This implies that $ (\mathrm a) \prec \Psi^4 \E |Z_{ab}|^{2p-1}$.
\item[(b)] By summing over $ j$ and using \eqref{eq:diff_Z} we get
\begin{align*}
(\mathrm b) & \prec \bigg\vert \E  \sum_{i,j,k}\sigma_{ai}\tau_{bj} \sum_{w + t = 2} \frac{1}{w! t!} \cal C^{(w, t+1)}(H_{ik}) G_{kj}G_{ji} (\partial^w_{ki} \partial^t_{ik}Z_{ab}^{p-1} \overline Z_{ab}^{p}) \bigg\vert \\\nonumber
& \prec \Psi \, \E \sum_{i,k}|\sigma_{ai} S_{ik}|(\Psi^3 + \delta_{ik}\Psi^2) |\partial^w_{ki} \partial^t_{ik}Z_{ab}^{p-1} \overline Z_{ab}^{p}| \\\nonumber
& \prec \Psi \,  \sum_{i,k}|\sigma_{ai} S_{ik}|(\Psi^3 + \delta_{ik}\Psi^2) (\Psi^4 \E |Z_{ab}|^{2p-2} +  \Psi^8 \E |Z_{ab}|^{2p-3}) \prec \Psi^8 \E |Z_{ab}|^{2p-2} +  \Psi^{12} \E |Z_{ab}|^{2p-3}.
\end{align*} 
\item[(c)] The three terms in $ (\mathrm c)$ have the same structure, so we will examine only the first one. According to \eqref{eq:A_bound}, which is valid also in the non Gaussian case, the derivative of $ Z$ is controlled as follows 
\begin{align}
A_{abcd}  = \partial_{cd} Z_{ab} \prec \Psi^5 + \delta_{cd} \Psi^4.
\end{align}
Thus, from \eqref{eq:bound_cumulants} we have
\begin{align*}
(\mathrm c) \prec \Psi \E \sum_{i,j,k} |\sigma_{ai} \tau_{bj} S_{ik}| |G_{kk}G_{ij}G_{ji} + G_{kj}G_{ik}G_{ij}||A_{abki} | |Z_{ab}|^{2p-2} \prec \Psi^8 \E |Z_{ab}|^{2p-2}.
\end{align*}
\end{itemize}
By putting together all the contributions and invoking Lemma \ref{lemma:moment_to_domination}, one concludes the proof.

\subsection{General complex case}\label{sec:complex_general_case}
Here we explain how our results extend to the case when 
\begin{align}\label{eq:complex_general_case}
\E H_{ij}^2 \neq 0.
\end{align}
Since the non-Gaussian terms in this case are treated exactly as when $ \E H_{ij}^2 = 0 $, we will assume that $ H_{ij}$ is Gaussian. 

The bounds in Proposition \ref{prop:main_estimates} revolve around the self consistent equation trick implemented in Lemma \ref{lemma_self-consistent_equation_trick}, so we need to examine how \eqref{eq:complex_general_case} affects this lemma. 
From the cumulant expansion formula in Lemma \ref{lem:cumulant_expansion} we see that equation \eqref{eq:test_expansion_0} becomes
\begin{align}\label{eq:add_term_general_complex}
z D_{abc}= - \E \sum_i \sigma_{ai} \delta_{ib} G_{ci} \mathfrak p (G, \overline G) + \E \sum_i \sigma_{ai} (S_{ij} \partial_{ji} + \cal C^{(2,0)}(H_{ij})\partial_{ij})(G_{jb}G_{ci} \mathfrak p (G, \overline G)). 
\end{align}
Note that $ \cal C^{(2,0)}(H_{ij})= \E H_{ij}^2 \in \bb S$ since $ |\cal C^{(2,0)}(H_{ij})| \leq S_{ij}$ and $ \partial_{ij}G_{jb}G_{ci} = - 2G_{ji} G_{jb} G_{ci}$: this means that the additional term never generates any diagonal entry of $G$ and therefore it is always smaller than the contribution proportional to $ \E |H_{ij}|^2 = S_{ij}$.

As an example, let us write down in details the estimate of the additional terms for the estimate of the elementary chain $ Y_{ab;u}^{(1)} \equiv Y_{uab}$ (Lemma \ref{lemma:chain_loop_X_estimates}). From Lemma \ref{lemma_self-consistent_equation_trick} modified according to \eqref{eq:add_term_general_complex}, we get 
\begin{align}\label{eq:Y_complex}\nonumber
& \E|Y_{uab}|^{2p} \prec \bigg \vert \E \breve\sigma_{ua} G_{ab} (Y_{uab})^{p-1} (\overline{Y}_{uab})^p \bigg\vert + \bigg\vert \E \sum_{i,j} \breve\sigma_{ui} \omega_{ij} G_{ai} G_{jb} \partial_{ji} (Y_{uab}^{p-1} \overline{Y}_{uab}^p) \bigg\vert \\\nonumber
& + \bigg\vert \E \sum_{i,j} \breve\sigma_{ui} S_{ij} ((G_{jj}-\mathfrak{m}) G_{ai} G_{ib} + G_{aj} G_{jb} (G_{ii}-\mathfrak{m}) ) Y_{uab}^{p-1} \overline{Y}_{uab}^p \bigg \vert \\
& + \bigg\vert \E \sum_{i,j} \breve\sigma_{ui} \tau_{ij} G_{ji} G_{ai} G_{jb} Y_{uab}^{p-1} \overline{Y}_{uab}^p \bigg \vert 
\end{align}
where $ \tau_{ij} = \cal C^{(2,0)}(H_{ij})$, $ \omega = S + \tau$ and $\tau, \omega,  \breve\sigma \in \bb S$. Note the last term on the right hand side of \eqref{eq:Y_complex} is the only new additional contribution with respect to the former estimate \eqref{eq:expansion_callY_pq}. By using the hypothesis $ \Lambda \prec \Psi$, we can easily control it:
$$
 \E \sum_{i,j} \breve\sigma_{ui} \tau_{ij} G_{ji} G_{ai} G_{jb} Y_{uab}^{p-1} \overline{Y}_{uab}^p  \prec \Psi^3 \E|Y_{uab}|^{2p-1}.
$$ 
Thus, \eqref{eq:expansion_callY_pq} (and consequently \eqref{eq:bound_Y}) actually holds true even when $ \E H_{ij}^2 \neq 0$. Analogous straightforward arguments show that also all the other bounds in Proposition \ref{prop:main_estimates} and lemmata \ref{lemma:chain_loop_X_estimates}, \ref{lemma:first_induction_step} and \ref{lemma:bound_B} remain valid. 

\section{High dimensions}\label{sec:high_dim}
Fix $d \geq 2$ and recall that $N = L^d$ and $M \asymp W^d$. 
From Proposition 2.8 in \cite{EKYY13} for $ d \geq 2$ we have
\begin{align}\label{eq:bound_Theta_high_dim}
\Theta_{xy} \leq C \tilde \Phi^2\,,
\end{align}
where the explicit expression for $\Theta$ is given in Lemma 8.2 in \cite{EKYY13} and
\begin{align*}
\tilde \Phi^2 :=  \frac{1}{M}  + \frac{1}{N \eta} \,.
\end{align*}
In this setting we can show the analogues of Theorem \ref{th:local_law_d=1}, Corollary \ref{cor:deloc_d=1} and Theorem \ref{th:diffusion_d=1} in \cite{EKYY13}. 

\begin{theorem}[High dimensions]\label{th:high_dim_local_law}
Let $d \geq 2$ and assume \eqref{decay of f}. 
\begin{itemize}
\item[(i)] Suppose that
\begin{align}\label{eq:cond_high_dim_local_law}
L \ll W^{1+d/2}\,, \ \ \ \  \eta \gg \frac{L}{W^{1+d}}\,.
\end{align}
Then for $z\in \f S$ we have
\begin{equation*}
   \Lambda^2 \;\prec\; \tilde \Phi^2\,.
\end{equation*}

\item[(ii)] If  $L \ll W^{1+ \frac{d}{d+1}}$, then the eigenvectors of $H$ are completely delocalized in the sense of Proposition 7.1 in \cite{EKYY13}.

\item[(iii)] 
Assume that 
$$
L \ll W^{1 + d/3}, \ \ \ (W/L)^2 \leq \eta \leq 1\,.
$$
Then
\begin{equation}\label{eq:T_diff_2}
 T_{xy} -\Theta_{xy} \;\prec\; \frac{1}{M^{3/2} \eta}\,.
\end{equation}
Moreover, the analogues of \eqref{Tfin} and \eqref{Tfin1} hold with the explicit expression for $ \Upsilon_{xy}$ given in Theorem 8.6 in \cite{EKYY13}.

\end{itemize}

%\begin{theorem}[Diffusion profile]\label{th:high_dim_diffusion}
%\begin{equation*}
%\Upsilon_{xy} \;\equiv\; \Upsilon^{(K)}_{xy} \;\deq\; \frac{\eta^{d/2 - 1}}{W^d} \pb{V * \varphi_{\sqrt{\eta}}} \pbb{\frac{\sqrt{\alpha \eta}}{W}D^{-1/2} (x - y)} + \frac{1}{W^d} \avgbb{\frac{x-y}{W}}^{-K} + \frac{\eta^{d/2}}{W^d} \avgbb{\frac{\sqrt{\eta} (x-y)}{W}}^{-K}\,,
%\end{equation*}
%where $K$ is an arbitrary, fixed, positive integer.
\end{theorem}

\begin{remark}[Comparison with the analogous results in \cite{EKYY13}]
Theorems \ref{th:high_dim_local_law} improves Theorem 8.4 and Theorem 8.6 in \cite{EKYY13}.
In fact, in Theorem 8.4 \cite{EKYY13} it is assumed that $L \ll W^{1 + d/4}$ and $\eta \gg \frac{L^{2}}{W^{d+2}}$,
in Theorem 8.6 that $L \ll W^{1+d/4}$. 

Furthermore, in Corollary 8.5 in \cite{EKYY13} the stated condition, i.e. $L \ll W^{1+d/4}$, is wrong: it should be $L \ll W^{1+\frac{d}{d+2}} $ because one must have that $ \frac{L^{2}}{W^{d+2}} \ll \eta \leq N/M$ in order to apply Proposition 7.1. Thus Theorem \ref{th:high_dim_local_law}(iii) improves also Corollary 8.5 in \cite{EKYY13}.
\end{remark}

\begin{proof}
\begin{itemize}
\item[(i)] For $d > 1$ we treat the error term $\tilde{ \cal E}$ of equation \eqref{self-const intro} very similarly to what we did in Proposition \ref{prop:fourier_analysis} for $d=1$. We define the analogue of the matrix $Q$ as $Q^{(\alpha)}$ such that $\widehat q^{(\alpha)}(p) = 1-\chi(p W^{1-\alpha}L^{\alpha})$ for $p \in \bb T_L^d$, where $q^{(\alpha)}_x = Q^{(\alpha)}_{x0}$ for $x \in \bb T_L^d$ and $\chi$ is a smooth bump function as defined in the proof of Proposition \ref{prop:fourier_analysis}. Here we set $\alpha \in [0, 1)$ since for $\alpha \geq 1$, $ \chi(p W^{1-\alpha}L^{\alpha})=0$ for $p \neq 0$, thus we would be back to the old analysis performed in \cite{EKYY13}. We will tune $\alpha$ in order to get the optimal conditions for the local law.

Let $w_x := (\Pi \cal E)_{xy}$, then 
$$
\sup_{x,y}|\tilde{ \cal E}_{xy}| \prec \sup_y \Big \Vert \frac{Q_{\alpha}}{I- |\fra m|^2 S} \b w \Big \Vert_{\infty} + \sup_y \Big \Vert \frac{I -Q_{\alpha}}{I- |\fra m|^2 S} \b w \Big \Vert_{\infty}
$$
where
\begin{align}\nonumber
& \sup_y \Big \Vert \frac{I - Q_{\alpha}}{I- |\fra m|^2 S} \b w \Big \Vert_{\infty} = \sup_y \bigg \Vert \sum_{p \in (\bb T_L^d)^* , p \neq 0} \frac{\chi(pW^{1+\alpha})}{1-|\mathfrak{m}|^2 \widehat{s}(p)} \b e(p) \langle \b e(p), \cal E \rangle_y \bigg \Vert_{\infty} \\
& 
%\prec \frac{1}{\sqrt N}  \sup_y \sup_{p \neq 0} |\langle \b e(p), \cal E \rangle_y| \Big(\frac{L}{W} \Big)^2 \sum_{\substack{j \in \Z^d \\ 0 < |j| \leq L/W^{1+\alpha}}} |j|^{-2}
 \prec \frac{1}{\sqrt{N}} \min \Big\{ \frac{1}{\eta}, \Big(\frac{L}{W} \Big)^2
\Big\} \Big( \frac{L}{W} \Big)^{(d - 2)(1-\alpha)}\sup_y \sup_{p \neq 0} |\langle \b e(p), \cal E \rangle_y|\label{eq:I-Q_alpha}
\end{align} 
and
%, $|\widehat S(p) \ind{|p|> W^{-(1+\alpha)}}| \leq 1 - c W^{-2 \alpha}$, so that
$$
\sup_y \bigg \Vert \frac{ Q_{\alpha}}{1-|\mathfrak{m}|^2S} \b w \bigg \Vert_{\infty} 
\prec \frac{1}{\eta + (W/L)^{2 \alpha }} \sup_{x,y} |\cal E_{xy}|\,.
$$
Hence, for $\alpha < 1$
\begin{align}\label{eq:err_high_dim}
\sup_{x,y}|\tilde{ \cal E}_{xy}| \prec \frac{1}{\sqrt{N}} \min \Big\{ \frac{1}{\eta}, \Big(\frac{L}{W} \Big)^2
\Big\} \Big( \frac{L}{W} \Big)^{(d - 2)(1-\alpha)} \sup_y \sup_{p \neq 0} |\langle \b e(p), \cal E \rangle_y| + \frac{1}{\eta + (W/L)^{2 \alpha }}\sup_{x,y} |\cal E_{xy}|\,,
\end{align}
while for $\alpha \geq 1$ we get the same bound as in \cite{EKYY13}, i.e.\ 
\begin{align}\label{eq:err_old_bound}
\sup_{x,y}|\tilde{ \cal E}_{xy}| \prec \frac{1}{\eta + (W/L)^2} \sup_{x,y} |\cal E_{xy}|\,.
\end{align}
Moreover, from \cite{EKYY13} we know that 
\begin{align}\label{eq:old_bound}
\sup_{x,y} |\cal E_{xy}| \prec \Psi^4 + \Psi^2 M^{-1/2}.
\end{align}
Note that \eqref{eq:old_bound} is derived by using the fluctuation averaging bounds in \cite{EKY13}, but we believe that it can be obtained also by the cumulant expansion method by performing nested expansions of $\cal P_{xy}$ and $\cal R_{xy}$ defined in \eqref{eq:def_P} and \eqref{eq:def_R}. 

 By combining \eqref{eq:old_bound} with the \eqref{eq:P_vy} and \eqref{eq:R_vy} in Proposition \ref{prop:main_estimates} we get
$$
\sup_{x,y}|\tilde{ \cal E}_{xy}| \prec \Big( \frac{L}{W} \Big)^{2 \alpha} \Psi^2 \Big[\Big(\frac{N}{M} \Big)^{1-\alpha}\Big( \frac{1}{N \eta} + \frac{\Psi^{-1}}{N \sqrt \eta} + \frac{\Psi}{\sqrt{N \eta}} \Big) + \Psi^2 + M^{-1/2}\Big]\,.
$$
Recall that initially, thanks to Lemma \ref{lm:lsc}, $\Psi^2 = (M\eta)^{-1}$ and that $\Psi^{-2} \leq N \eta + M$. In order to apply Lemma \ref{lemma:self-improving_bound}, we need to have
\begin{align}
 \ \ \ L \ll W^{1 + \frac{d}{4\alpha+d(1-2\alpha)}}\,,\ \ \ \eta\gg \frac{L^{2\alpha+d/2-d\alpha}}{W^{2\alpha+3d/2-d\alpha}}\,, \ \ \ \eta \gg \frac{L^{2 \alpha}}{W^{2\alpha + d}}\,,\ \ \ L \ll W^{1+d/4\alpha}\,.
\end{align}
By comparing the above conditions, it is easy to see that the optimal conditions are attained for $\alpha = 1/2$, i.e.\
$$
L \ll W^{1+d/2}, \ \ \ \eta \gg \frac{L}{W^{1+d}}.
$$
Observe that from \eqref{eq:err_high_dim} we see that for $d=2$ we can choose any $\alpha \in [0, 1/2]$ because the sum over the moments is only logarithmically divergent:
$$
\sum_{\substack{p \in (\bb T_L^2)^* \\ p \neq 0}} \frac{\chi(pW^{1+\alpha})}{1-|\mathfrak{m}|^2 \widehat{s}(p)} \leq C \min \Big\{\frac{1}{\eta}, \frac{L^2}{W^2}  \Big\} \sum_{\substack{j \in \Z^2 \\ 0 < |j| \leq L/W^{1+\alpha}}} |j|^{-2} = \min \Big\{\frac{1}{\eta}, \frac{L^2}{W^2}  \Big\} O \Big(\log \frac{L}{W^{1+\alpha}} \Big).
$$
\item[(ii)]
From \eqref{large eta estimate} and (i) we know that for the eigenvector delocalization we need $\eta < M/N$ and  $\eta \gg L/W^{1+d}$, which is true when $L \ll W^{1 + \frac{d}{d+1}}$. 
\item[(iii)] From (i) we know that $\Lambda^2 \prec M^{-1} + (N\eta)^{-1} $ when $L \ll W^{1+d/2}$ and $\eta \gg L / W^{1+d}$. Here we want $\eta \geq (W/L)^2$, therefore, in order to use (i), we need to require
$$
\frac{L}{W^{1+d}} \ll \frac{W^2}{L^2},
$$
i.e.\ $L \ll W^{1+d/3}$. Note also that $\eta \geq (W/L)^2$ implies $\eta \geq M/N$ for $d \geq 2$, thus $\Lambda^2 \prec M^{-1}$. This means that \eqref{eq:err_high_dim} holds with $\Psi = M^{-1/2}$. In this setting it is easy to check from \eqref{eq:err_high_dim} and \eqref{eq:err_old_bound} that \eqref{eq:T_diff_2} and the analogues of \eqref{Tfin} and \eqref{Tfin1} are valid for any $\alpha \geq 1/2$.
\end{itemize}

\end{proof}

%\begin{remark} In the proof of Theorem \ref{th:high_dim_local_law}(i) from \eqref{eq:err_high_dim} we see that for $d=2$ we can choose any $\beta \in [0, 1/2]$ because the sum over the moments is only logarithmically divergent, in fact 
%$$
%\sum_{\substack{p \in (\bb T_L^2)^* \\ p \neq 0}} \frac{\chi(pW^{1+\alpha})}{1-|\mathfrak{m}|^2 \widehat{s}(p)} \leq C \min \Big\{\frac{1}{\eta}, \frac{L^2}{W^2}  \Big\} \sum_{\substack{j \in \Z^2 \\ 0 < |j| \leq L/W^{1+\alpha}}} |j|^{-2} = \min \Big\{\frac{1}{\eta}, \frac{L^2}{W^2}  \Big\} O \Big(\log \frac{L}{W^{1+\alpha}} \Big).
%$$
%\end{remark}

\appendix

\section{Proof of Lemma \ref{lemma:improv_bound}}

We recall that Lemma \ref{lemma:improv_bound} basically coincides with Corollary 5.4 in \cite{EKYY13}. Here we give a proof which does not rely on the averaging fluctuations estimate in \cite{EKY13}. 
%Instead, we adopt the strategy followed in \cite{HKR17} to prove the local law for Wigner matrices with non-zero expection.

%\begin{Lemma}[corollary 5.4 in \cite{EKYY13}]
%Suppose that $\Lambda \prec \Psi$ for some admissible parameter $\Psi$ and $T_{ab}, T'_{ab} \prec \Omega^2_{ab}$ for a family of admissible control parameters $\Omega_{ab}$ indexed by a pair $(a,b)$. Then 
%\begin{align}\label{eq:app1}
%|G_{ab}- \fra m|^2 \prec \Omega_{ab}^2 + \Psi^4.
%\end{align}
%Moreover, setting $\tilde \Omega = \sup_{ab} \Omega_{ab}$ and assuming that $\tilde \Omega$ is admissible, one has
%\begin{align}\label{eq:app2}
%\Lambda^2 \prec \tilde \Omega^2.
%\end{align}
%\end{Lemma}
%
%\begin{proof}
We will assume that $ H$ is Hermitian with Gaussian entries and such that $ \E H_{ij}^2 =0$, but the result holds also in the general complex case and the additional terms are treated as we saw in Section \ref{sec:nonGaussian_generalComplex}.
%As in \cite{HKR17}, we define the matrix map $ \Pi:\C^{N \times N} \to \C^{N \times N}$ such that 
%\begin{align}\label{eq:Pi_equation}
%\Pi(G) = I + zG + \fra m(z)G = HG + \fra m(z)G. 
%\end{align}
%\textcolor{red}{We want to show that $\Pi(G) \prec \Omega$ and this implies, by the stability of the equation \eqref{eq:Pi_equation}, that $ \Lambda = G-\fra m I \prec \Omega$.}

To get the desired bounds, we consider the expectation $ \cal F_{ab} := \E|F_{ab}|^{2p}$ where $F_{ab}:=G_{ab} - \fra m \, \delta_{ab} $ and $p$ is an arbitrary strictly positive integer. 
The cumulant expansion yields
\begin{align*}
& z \cal F_{ab} = \E(zG_{ab}-z \fra m \, \delta_{ab})F_{ab}^{p-1} \ol F_{ab}^{p} = \E \bigg(\sum_i H_{ai}G_{ib}-(1+z \fra m) \, \delta_{ab}\bigg)F_{ab}^{p-1} \ol F_{ab}^{p} \\\nonumber
& = \fra m^2 \delta_{ab}\E F_{ab}^{p-1} \ol F_{ab}^{p} - \E \sum_i S_{ai} \partial_{ia} G_{ib} F_{ab}^{p-1} \ol F_{ab}^{p} \\\nonumber
& = \fra m^2 \delta_{ab} \E F_{ab}^{p-1} \ol F_{ab}^{p} - \E \sum_i S_{ai} G_{ii}G_{ab} F_{ab}^{p-1} \ol F_{ab}^{p} + \E \sum_i S_{ai}  G_{ib} \partial_{ia}F_{ab}^{p-1} \ol F_{ab}^{p}
\end{align*}
where we used \eqref{eq:m_equation}. Using the trivial identities $ G_{ii} = G_{ii} - \fra m + \fra m$ and $ G_{ab} = F_{ab} + \fra m \delta_{ab}$, we get
\begin{align*}
 (z + \fra m) \cal F_{ab} =& - \E \sum_i S_{ai} (G_{ii}-\fra m)F_{ab} F_{ab}^{p-1} \ol F_{ab}^{p} - \fra m \delta_{ab} \E \sum_i S_{ai} (G_{ii}-\fra m) F_{ab}^{p-1} \ol F_{ab}^{p} \\\nonumber
&+ \E \sum_i S_{ai}  G_{ib} \partial_{ia}F_{ab}^{p-1} \ol F_{ab}^{p}\,.
\end{align*}
Using \eqref{eq:X_bound}, \eqref{eq:m_equation} and that $|F_{ab}| \leq \Lambda \prec \Psi$, one obtains
\begin{align}\label{eq:F_temp_bound}
\cal F_{ab} = O_{\prec}(\Psi^2 \E| F_{ab}|^{2p-1}) + O_{\prec}\bigg(\bigg\vert \E \sum_i S_{ai}  G_{ib} \partial_{ia}F_{ab}^{p-1} \ol F_{ab}^{p} \bigg\vert \bigg).
\end{align}
Let us examine the second term on the right hand side of \eqref{eq:F_temp_bound}: we note that
\begin{align*}
\partial_{ia} F_{ab} = - G_{ai}G_{ab}, \ \ \ \partial_{ia} \ol F_{ab} = - \ol G_{aa} \ol G_{ib},
\end{align*}
therefore \eqref{eq:F_temp_bound} becomes
\begin{align}\nonumber
\cal F_{ab} =& \, O_{\prec}(\Psi^2 \E| F_{ab}|^{2p-1}) + O_{\prec}\bigg(\bigg\vert \E \sum_i S_{ai}  G_{ib} G_{ai}G_{ab}F_{ab}^{p-2} \ol F_{ab}^{p} \bigg\vert \bigg) \\
& + O_{\prec}\bigg(\bigg\vert \E \sum_i S_{ai} | G_{ib}|^2 \ol G_{aa} F_{ab}^{p-1} \ol F_{ab}^{p-1} \bigg\vert \bigg).
\label{eq:temp_bound_F}
\end{align}
In the second term on the right hand side of \eqref{eq:temp_bound_F}, when $a \neq b$ we can use \eqref{eq:bound_Y} so that
$$
\cal F_{ab} \prec \Psi^2 \E| F_{ab}|^{2p-1} + \E(\Psi^4 + T_{ab})| F_{ab}|^{2p-2},
$$
while, when $a = b$, the trivial inequality $|xy| \leq (|x|^2 + |y|^2)/2$ yields
$$
\sum_i S_{ai}G_{ib} G_{ai} \prec \sum_i S_{ai}(|G_{ai}|^2 + |G_{ia}|^2) = T_{aa} + T_{aa}',
$$
so that
$$
\cal F_{aa} \prec \Psi^2 \E| F_{aa}|^{2p-1} + \E(T_{aa} + T'_{aa})| F_{aa}|^{2p-2}.
$$
Thus, we finally have
\begin{align*}
\cal F_{ab} & \prec \Psi^2 \E| F_{ab}|^{2p-1} + (\Psi^4 + \Omega_{ab})\E| F_{ab}|^{2p-2}  \prec (\Psi^2 + \Omega_{ab})\E| F_{ab}|^{2p-1} + (\Psi^4 + \Omega^2_{ab}) \E| F_{ab}|^{2p-2}.
\end{align*}
%where we used that $ T_{ab}' := \sum_i |G_{ai}|^2 S_{ib} \prec \Omega^2$ (I think that this should be an additional hypothesis, which is true because $ T'$ is as big as $ T$. It would be nice to have a short clear argumet that show that if $ T \prec \Omega^2$, then $ T' \prec \Omega^2$). 
Lemma \ref{lemma:moment_to_domination} completes the proof. 

\begin{remark}
The estimate in Lemma \ref{lemma:improv_bound} cannot be improved in the sense that $F_{ab}$ will always be bounded at least by $T_{ab} = \sum_i S_{ai}|G_{ib}|^2$ because $T_{ab}$ appears explicitly in the third term on the right hand side of \eqref{eq:temp_bound_F}.
\end{remark}

\end{document}